%% file: main.tex
\documentclass[12pt]{article}
\usepackage{geometry}
\usepackage[T1]{fontenc}
\usepackage[utf8]{inputenc}
\usepackage{enumerate,enumitem}
\usepackage{natbib}
\usepackage{url} 
\usepackage{bbm}
\usepackage{float}
\usepackage{xspace}
\usepackage{comment}
\usepackage{diagbox}
\usepackage{multirow, multicol}
\usepackage{boxedminipage}
\usepackage[dvipsnames]{xcolor}
\usepackage{authblk}
\usepackage{natbib}
\usepackage{wrapfig}
\usepackage{relsize}
\usepackage{comment}
\usepackage{diagbox}
\usepackage{enumitem}
\usepackage{multirow}
\usepackage{boxedminipage}
\usepackage{algorithm,algpseudocode}

\usepackage{colortbl} 
\usepackage{tcolorbox}
\usepackage{wrapfig}
\usepackage{setspace}
\usepackage{graphicx,psfrag,epsf,subcaption}
\usepackage{amssymb,amsthm,amsmath,mathrsfs,bbm,nicefrac}
\usepackage{booktabs}
\usepackage[bookmarks,bookmarksnumbered,%
    allbordercolors={0.8 0.8 0.8},%
    linktocpage%
    ]{hyperref} 
\hypersetup{colorlinks,
  citecolor={RoyalBlue},linkcolor={ForestGreen},urlcolor={BrickRed}}
\usepackage[capitalise,noabbrev]{cleveref}
\usepackage{textcomp}
\usepackage{times}
\usepackage[toc,page,title,titletoc,header]{appendix}
\input{mathcommands}
\newcommand{\negred}[1]{\textcolor{Mahogany}{\textbf{#1}}}
\newcommand{\bracket}[1]{(#1)}

\title{Neyman-Pearson multiclass classification under label noise via empirical likelihood}
\date{}
\author[1]{Qiong Zhang}
\author[2]{Qinglong Tian}
\author[2]{Pengfei Li\thanks{Correspondence to: pengfei.li@uwaterloo.ca. Qiong Zhang gratefully acknowledges the funding support by the National Key R \& D Program of China Grant 2024YFA1015800.}}
\affil[1]{\small Institute of Statistics and Big Data, Renmin University of China, Beijing, China}
\affil[2]{\small Department of Statistics and Actuarial Science, University of Waterloo}

\begin{document}

\maketitle

\begin{abstract}
In many classification problems, misclassification costs are highly asymmetric, while training labels are often corrupted due to measurement error, annotator variability, or adversarial noise. 
The Neyman--Pearson multiclass classification (NPMC) framework addresses such asymmetry by controlling class-specific errors, but existing methods assume that training labels are correctly observed. 
To our knowledge, no existing approach handles NPMC under label noise in the multiclass setting, and the only binary method requires prior knowledge of the noise mechanism. 
A fundamental difficulty is that, without structural assumptions, noisy-label models are non-identifiable: distinct combinations of class-conditional distributions and noise mechanisms can induce the same observed distribution, preventing recovery of the quantities required for error control.
We show that the exponential tilting density ratio model restores identifiability, and leverage this structure to develop an empirical likelihood approach for NPMC with noisy labels. 
The proposed method jointly estimates clean-label class proportions, posterior probabilities, and the noise mechanism from noisy data, without requiring prior knowledge of the confusion matrix. 
An expectation--maximization algorithm enables efficient computation. 
The resulting estimators are $\sqrt{n}$-consistent and asymptotically normal, and the induced classifiers satisfy Neyman--Pearson oracle inequalities in both binary and multiclass settings. 
Simulation and real-data experiments demonstrate near-oracle performance.
\end{abstract}

\noindent%

\input{sections/intro}
\input{sections/setting}

\input{sections/estimation}

\input{sections/binary}
\input{sections/multiclass}

\input{sections/real_data}
\input{sections/conclusion}

\bibliography{biblio}
\bibliographystyle{apalike}

\clearpage
\newpage
\appendix
\section*{Appendices}
\pagenumbering{arabic}
\input{sections/appendix}

\end{document}

%% file: mathcommands.tex
\usepackage{amsmath,amsfonts,bm}

\newcommand{\ie}{{\em i.e.,~}} 

\def\vp{{\bm{p}}}

\def\vs{{\bm{s}}}

\def\vv{{\bm{v}}}
\def\vw{{\bm{w}}}

\def\mM{{\bm{M}}}

\def\mS{{\bm{S}}}
\def\mT{{\bm{T}}}

\def\mV{{\bm{V}}}
\def\mW{{\bm{W}}}

\DeclareMathAlphabet{\mathsfit}{\encodingdefault}{\sfdefault}{m}{sl}
\SetMathAlphabet{\mathsfit}{bold}{\encodingdefault}{\sfdefault}{bx}{n}

\def\gD{{\mathcal{D}}}

\def\gF{{\mathcal{F}}}
\def\gG{{\mathcal{G}}}
\def\gH{{\mathcal{H}}}

\def\gK{{\mathcal{K}}}
\def\gL{{\mathcal{L}}}

\def\gN{{\mathcal{N}}}

\def\gS{{\mathcal{S}}}

\def\gX{{\mathcal{X}}}

\def\gZ{{\mathcal{Z}}}

\def\sE{{\mathbb{E}}}

\def\sP{{\mathbb{P}}}

\def\sR{{\mathbb{R}}}
\def\sS{{\mathbb{S}}}

\DeclareMathOperator*{\argmax}{arg\,max}

\newcommand{\iid}{\overset{\text{i.i.d.}}{\sim}}
 \newcommand{\ind}{\perp\!\!\!\!\perp}

\newtheorem{remark}{Remark}[section]
\newtheorem{theorem}{Theorem}[section]
\newtheorem{lemma}{Lemma}[section]
\newtheorem{assumption}{Assumption}[section]
\newtheorem{example}{Example}[section]
\newtheorem{definition}{Definition}[section]
\newtheorem{proposition}{Proposition}[section]

\newcommand{\bbeta}{\mbox{\boldmath $\beta$}}

\newcommand{\bnu}{\mbox{\boldmath $\nu$}}
\newcommand{\btheta}{\mbox{\boldmath $\theta$}}
\newcommand{\blambda}{\mbox{\boldmath $\lambda$}}

\newcommand{\bgamma}{\mbox{\boldmath $\gamma$}}

\newcommand{\bea}{\begin{eqnarray*}}
\newcommand{\eea}{\end{eqnarray*}}
\newcommand{\ba}{\begin{eqnarray*}}
\newcommand{\ea}{\end{eqnarray*}}
\newcommand{\be}{\begin{equation}}
\newcommand{\ee}{\end{equation}}
\newcommand{\bi}{\begin{itemize}}
\newcommand{\ei}{\end{itemize}}

\makeatletter
\newcommand*\rel@kern[1]{\kern#1\dimexpr\macc@kerna}
\newcommand*\widebar[1]{%
  \begingroup
  \def\mathaccent##1##2{%
    \rel@kern{0.8}%
    \overline{\rel@kern{-0.8}\macc@nucleus\rel@kern{0.2}}%
    \rel@kern{-0.2}%
  }%
  \macc@depth\@ne
  \let\math@bgroup\@empty \let\math@egroup\macc@set@skewchar
  \mathsurround\z@ \frozen@everymath{\mathgroup\macc@group\relax}%
  \macc@set@skewchar\relax
  \let\mathaccentV\macc@nested@a
  \macc@nested@a\relax111{#1}%
  \endgroup
}
\makeatother

%% file: sections/intro.tex
\section{Introduction}
\subsection{Label noise is pervasive in real-world applications} 
\label{sec:label_noise_pervasiveness}
In supervised learning, \emph{label noise} refers to discrepancies between observed labels and the (unobserved) ground truth. 
Such contamination is ubiquitous in practice. 
Empirical studies suggest that roughly 5\% of labels in real-world datasets are incorrect, and even widely used benchmarks such as MNIST and ImageNet contain non-negligible annotation errors~\citep{northcutt2021confident}. 
In domain-specific settings, for example in CheXpert~\citep{irvin2019chexpert}, ambiguities in radiology reports give rise to label errors that are difficult to quantify yet unavoidable. 
More recently, outputs from LLMs have also been shown to induce label noise~\citep{zhang2025stochasticity}.

Label noise may arise from sources including measurement error, annotator variability, and adversarial corruption. 
Its impact on statistical learning procedures is well documented~\citep{frenay2013classification}: contaminated labels can distort decision boundaries, bias parameter estimates, and lead to overly complex models whose apparent performance, when evaluated against noisy labels, substantially overstates their accuracy on clean data~\citep{northcutt2021confident}. 
A substantial body of work~\citep{natarajan2013learning,patrini2017making,han2018masking,han2018co,wu2021class2simi,wang2022seminll} addresses label noise in general classification settings. 
However, they largely focus on improving overall classification accuracy, and \emph{does not account for settings where different types of errors must be controlled in an asymmetric manner}, as we discuss below.

\subsection{Neyman--Pearson multiclass classification under clean labeled data}
In many high-stakes applications, minimizing overall misclassification error is not the appropriate objective. 
The consequences of different types of errors are often highly asymmetric: for example, misclassifying a malignant tumor as benign is far more consequential than the reverse error~\citep{bokhari2021multi}, and approving a fraudulent loan carries greater risk than rejecting a legitimate one. 
In such settings, the goal is not to minimize average error, but to control errors for specific high-risk classes while reducing errors elsewhere.

The Neyman--Pearson (NP) classification framework formalizes this principle by seeking a classifier that minimizes a weighted misclassification risk subject to explicit upper bounds on class-specific error probabilities.
Let $(X, Y)$ be a pair of random variables, where $X \in \gX\subset \sR^p$ denotes the feature vector and $Y \in [K] = \{0,1,\dots,K-1\}$ the class label. 
The Neyman--Pearson multiclass classification (NPMC) problem\footnote{The NPMC problem reduces to classical overall misclassification error minimization when $\rho_k = w_k^* = \sP(Y = k)$ for all $k$ and $\gS=\varnothing$.} seeks a classifier $\phi: \gX \to [K]$ that solves:
\begin{equation}
\label{eq:NPMC}
\begin{split}
\min_{\phi} &\quad\gL(\phi) = \sum_{k=0}^{K-1}\rho_k P_k^*(\{\phi(X)\neq k\}) \\
\text{subject to} &\quad P_k^*(\{\phi(X)\neq k\}) \leq \alpha_k,~\forall~k\in \gS \subset [K],   
\end{split}
\end{equation}
where $\rho_k$, $\alpha_k$, and $\gS$ are pre-specified weights, target error levels, and constrained class indices respectively, and $P_k^*(A) := \sP(X \in A |Y=k)$ for any Borel set $A\subset \gX$. 
When $K=2$, problem~\eqref{eq:NPMC} reduces to the classical NP binary classification task of minimizing the type II error subject to a constraint on the type I error, a setting that has been extensively studied~\citep{scott2005neyman, scott2007performance, rigollet2011neyman, tong2013plug, zhao2016neyman, tong2018neyman, tong2020neyman, li2020bridging, xia2021intentional, wang2024non}. 

The general multiclass case ($K \geq 3$) is considerably more challenging and has received comparatively less attention. 
In particular, problem~\eqref{eq:NPMC} may be infeasible for certain specifications of ${\rho_k,\alpha_k}$.
\citet{tian2024neyman} addressed this issue via a Lagrangian dual formulation, showing that the dual objective is concave irrespective of the convexity of the primal problem, thereby enabling efficient optimization using standard methods. Their framework further provides a feasibility check and establishes NP oracle inequalities for the resulting classifier. 
A key feature of this approach is its plug-in structure: given consistent estimators of the class proportions ${w_k^* = \sP(Y=k)}$ and the class-posterior probabilities ${\pi_k^*(x) = \sP(Y=k | X=x)}$ under clean labels, the optimal classifier can be obtained by substitution into the dual solution.
Despite these appealing properties, the framework \emph{explicitly assumes that training labels are correctly observed}—an assumption that is routinely violated in precisely the applications as we show in Section~\ref{sec:label_noise_pervasiveness}.

\subsection{Challenges of NPMC under noisy labeled data}
The plug-in structure of \citet{tian2024neyman} suggests, at least in principle, a natural route to handling label noise: rather than redesigning the solver, one may instead supply consistent estimators of $w_k^*$ and $\pi_k^*(x)$. 
However, \emph{constructing such estimators from noisy labeled data alone is far from straightforward and constitutes the central statistical challenge} addressed in this paper.

When only noisy labels $\widetilde{Y}$ are observed, naive estimators target the contaminated quantities $\{\widetilde{w}_k^* = \sP(\widetilde{Y}=k)\}$ and $\{\widetilde{\pi}_k^*(x)=\sP(\widetilde{Y}=k | X=x)\}$ rather than their clean-label counterparts.
Substituting these into the NPMC solver yields classifiers calibrated to the noisy label distribution, which in turn leads to systematic violations of the error constraints with respect to the true labels. 
For instance,~\citet{yao2023asymmetric} shows that ignoring label noise results in overly conservative type~I error control and, consequently, inflated type~II error under the binary NP setting.
A natural strategy is to apply existing noisy-label methods as a preprocessing step before solving the NPMC problem. 
However, this approach is inadequate for our purposes. 
Loss-correction methods such as~\citet{natarajan2013learning} and~\citet{patrini2017making} are designed to recover the Bayes-optimal classifier for overall accuracy under noisy data, rather than to provide consistent estimators of the class proportions $w_k^*$ and posterior probabilities $\pi_k^*(x)$ required by the NPMC solver. 
Moreover, they rely on strong assumptions, including knowledge of the confusion matrix $\mM^*=\{\sP(Y=k|\widetilde{Y}=l)\}$ and, in the case of~\citet{patrini2017making}, the existence of anchor points. 
The confident learning approach of~\citet{northcutt2021confident} likewise targets improvements in overall classification performance and does not yield the estimators needed for NPMC.

To the best of our knowledge,~\citet{yao2023asymmetric} is the only work that addresses asymmetric error control under label noise. 
Their analysis is restricted to the binary NP setting, where they propose a noise-adjusted NP umbrella algorithm. 
However, their method requires prior knowledge of noise confusion matrix $\mM^*$. 
Such information is rarely available in practice, and the method does not provide a way to estimate these quantities. 
In the absence of such information, one may instead impose bounds on these quantities; however, loose bounds can lead to performance comparable to that of naive methods that ignore label noise.
Moreover, their approach does not extend to the multiclass setting, leaving the multiclass problem under label noise unaddressed.

A common feature of these existing approaches is their reliance on knowledge of the noise confusion matrix $\mM^*$, which is typically unavailable in practice. 
A natural alternative is therefore to estimate this quantity from the observed data. 
However, this leads to a more fundamental difficulty: \emph{\underline{without additional structural assumptions, the noisy-label model is non-identifiable}}. 
In particular, distinct combinations of class-conditional distributions $\{P_k^*\}$ and confusion matrices can induce the same observed data distribution, so that the quantities required for NPMC error control cannot, in general, be recovered. 
This non-identifiability is not merely a theoretical concern but arises in concrete ways that preclude naive estimation. 
For example, convex perturbations of the class-conditional distributions can be absorbed into a modified noise mechanism without altering the observed distribution (Example~\ref{ex:algebraic_unidentifiability}). 
These observations indicate that the problem is inherently ill-posed without suitable structural constraints, motivating the approach in this paper.

\subsection{Our contributions and overview}
We address the Neyman--Pearson multiclass classification problem under label noise by resolving a fundamental identifiability barrier. 
We begin by formalizing the non-identifiability of noisy-label classification through concrete examples, showing that the problem is ill-posed without additional structural assumptions (Example~\ref{ex:algebraic_unidentifiability}). 
To overcome this difficulty, we introduce an exponential tilting density ratio model (DRM), under which the class-conditional distributions are modeled relative to a nonparametric reference distribution. 
We show that, in the absence of feature collinearity, this structure is sufficient to uniquely identify the confusion matrix, the clean-label class proportions $\{w_k^*\}$, and the posterior probabilities $\{\pi_k^*(x)\}$ from noisy observations alone, and we formalize this via a strict identifiability result for both the DRM parameters and the reference distribution (Theorem~\ref{thm:identifiability}).
To the best of our knowledge, this result has no direct analogue in the existing DRM literature and forms the foundation of our approach.

Building on this identification result, we develop an empirical likelihood (EL) framework for estimation and inference in the NPMC problem under the standard instance-independent noise assumption $\widetilde{Y} \ind X|Y$.
The parameters $\{w_k^*\}$ and $\{\pi_k^*(x)\}$ are jointly estimated via the maximum empirical likelihood estimators (MELEs) from noisy training data, without requiring knowledge of the noise transition matrix. 
We show that these estimators are $\sqrt{n}$-consistent and asymptotically normal (Theorem~\ref{thm:estimator_asymptotic_normality}).

For efficient computation, we treat the true labels as latent variables and develop an expectation-maximization (EM) algorithm to maximize the profile empirical likelihood. 
The resulting procedure is straightforward to implement, as the M-step reduces to a weighted multinomial logistic regression. 
We further establish monotonic ascent of the objective function (Proposition~\ref{thm:em_convergence}), which guarantees convergence of the algorithm.

Finally, we show that the resulting plug-in classifiers satisfy Neyman--Pearson oracle inequalities with respect to the true labels, in both binary (Theorem~\ref{thm:binary_type1}) and multiclass settings (Theorem~\ref{thm:multiclass_error}). 
In particular, our method achieves valid error control without requiring prior knowledge of the noise transition matrix and provides, to the best of our knowledge, the first such guarantees for multiclass NPMC under label noise. 
Empirical results on simulated and real datasets demonstrate near-oracle performance and substantial improvements over existing approaches.

The rest of the paper is organized as follows. 
Section~\ref{sec:problem_setting} introduces the problem setup and the relationship between noisy and clean labels. 
Section~\ref{sec:el_general_form} presents the EL-based estimator, its asymptotic properties, and the EM algorithm. 
Section~\ref{sec:binary} studies the binary NP problem, including the classifier, oracle guarantees, and simulations. 
Section~\ref{sec:multiclass} extends the framework to the multiclass setting. Section~\ref{sec:real_data} presents real-data experiments, and Section~\ref{sec:conclusion} concludes.

%% file: sections/setting.tex
\section{Problem setting}
\label{sec:problem_setting}
Suppose we observe $n$ independent and identically distributed samples $\{(X_i, \widetilde{Y}_i)\}_{i=1}^n$, where $X_i$ denotes the feature vector, $Y_i$ the unobserved true label, and $\widetilde{Y}_i$ its noisy observed version.
To enable valid statistical inference, we must formalize the relationship between latent $Y$ and observed $\widetilde{Y}$. 
Following established work~\citep{angluin1988learning, yao2023asymmetric, sesia2024adaptive, bortolotti2025noise}, we adopt the \emph{standard instance independence assumption}:

\begin{assumption}[Instance-independent noise]
\label{assumption:instance_independent_noise}
$\widetilde{Y} \ind X \mid Y$ (i.e., the noisy label $\widetilde{Y}$ is conditionally independent of features $X$ given the true label $Y$).
\end{assumption}

This assumption captures the common scenario as shown in~\citet[Appendix A.3]{sesia2024adaptive} where label noise depends primarily on class membership rather than specific feature values, as frequently observed in real-world settings like crowdsourcing~\citep{snow2008cheap}. 
Beyond its practical relevance, this assumption provides crucial tractability by establishing a clear relationship between the contaminated distribution $\widetilde{P}_k^*$ of $\widetilde{X}^k\triangleq X |(\widetilde{Y}=k)$ and the true conditional distribution $P_k^*$ of $X^{k}\triangleq X |(Y=k)$, as formalized in the proposition below. 
\begin{proposition}[Distributional relationships under label noise]
\label{proposition:x_conditional_relationships}
Let $\widetilde{w}_l^* = \sP(\widetilde{Y}=l)$ and $w_k^* = \sP(Y=k)$ denote the noisy and true class proportions, respectively. 
Let $\pi_k^*(x) = \sP(Y=k| X=x)$ and $\widetilde{\pi}_l^*(x) =\sP(\widetilde{Y}=l| X=x)$.
Define:
\begin{itemize}
\item $\mM^*$: The confusion matrix with entries $M_{lk}^* = \sP(Y=k| \widetilde{Y}=l)$,
\item $\mT^*$: The noise transition matrix with entries $T_{lk}^* = \sP(\widetilde{Y}=l| Y=k) = \widetilde{w}_l^* M_{lk}^* / w_k^*$.
\end{itemize}
Then, under Assumption~\ref{assumption:instance_independent_noise}, the following relationships hold:
\begin{equation}
\label{eq:marginal_link} 
\text{(Marginal)} \quad w_k^* = \sum_{l\in[K]} M_{lk}^*\widetilde{w}_l^*, \quad \widetilde{w}_l^* = \sum_{k\in[K]} T_{lk}^*w_k^*, 
\end{equation}
\begin{equation}
\label{eq:conditional_distribution_link} 
\text{(Conditional)}  \quad \widetilde{P}_l^* = \sum_{k\in[K]} M_{lk}^*P_k^*, \quad P_k^* = \sum_{l\in[K]} T_{lk}^*\widetilde{P}_l^*, \end{equation}
\begin{equation}
\label{eq:posterior_link}
\text{(Posterior)} \quad \widetilde{\pi}_l^*(x) = \sum_{k\in[K]} T_{lk}^* \pi_k^*(x). \quad\quad\quad\quad\quad 
\end{equation}
\end{proposition}
The proof is deferred to Appendix~\ref{app:proposition_proof}.    
Based on these relationships, a naive two-step approach to estimating $\pi_k^*(x)$ via~\eqref{eq:posterior_link} would require prior knowledge of the transition matrix $\mT^*$, which is typically unavailable in practice.
In the absence of such knowledge, we emphasize that, without additional structural assumptions, the model is fundamentally unidentifiable, beyond the usual label permutation ambiguity encountered in mixture models. 
In particular, multiple distinct parameterizations can induce exactly the same observed distribution. 
We illustrate this phenomenon through the following example.

\begin{example}[Algebraic unidentifiability of true distributions]
\label{ex:algebraic_unidentifiability}
Consider a binary classification problem ($K=2$) where we observe noisy distributions $\widetilde{P}_0(x)$ and $\widetilde{P}_1(x)$. Let the true underlying distributions be $P_0^*(x)$ and $P_1^*(x)$, and let the true noise transition matrix be $\mM^*$, such that:
\[
\widetilde{P}_0(x) = (1-\rho_0^*)P_0^*(x) + \rho_1^* P_1^*(x),~\widetilde{P}_1(x) = \rho_0^* P_0^*(x) + (1-\rho_1^*) P_1^*(x),
\]
where $\rho_0^* = \sP(\widetilde{Y}=1 | Y=0)$ and $\rho_1^* = \sP(\widetilde{Y}=0 | Y=1)$.

If $P_0^*(x)$ and $P_1^*(x)$ are unconstrained, the model is unidentifiable. For any small constant $\epsilon \in (0, 1)$, we can define a new, "contaminated" candidate distribution for class 1: $P_1^{**}(x) \triangleq (1-\epsilon)P_1^*(x) + \epsilon P_0^*(x).$
Because $P_1^{**}(x)$ is a convex combination of two valid probability density functions, it is also a valid density function. We can rewrite the true $P_1^*(x)$ in terms of this new distribution:
$$ P_1^*(x) = \frac{1}{1-\epsilon}P_1^{**}(x) - \frac{\epsilon}{1-\epsilon}P_0^*(x).$$
Substituting this back into the equation for the observed $\widetilde{P}_0(x)$, we obtain:
\[
\mathsmaller{\widetilde{P}_0(x) = (1-\rho_0^*)P_0^*(x) + \rho_1^* \left[ \frac{1}{1-\epsilon}P_1^{**}(x) - \frac{\epsilon}{1-\epsilon}P_0^*(x) \right] = \left( 1 - \rho_0^* - \frac{\epsilon \rho_1^*}{1-\epsilon} \right) P_0^*(x) + \left( \frac{\rho_1^*}{1-\epsilon} \right) P_1^{**}(x).}
\]
Let $\rho_1^{**} = \frac{\rho_1^*}{1-\epsilon}$ and $(1-\rho_0^{**}) = 1 - \rho_0^* - \frac{\epsilon \rho_1^*}{1-\epsilon}$. As long as $\epsilon$ is chosen to be sufficiently small such that $\rho_1^{**} < 1$ and $(1-\rho_0^{**}) > 0$, the new parameters $\{\rho_0^{**}, \rho_1^{**}, P_0^*(x), P_1^{**}(x)\}$ are perfectly valid and satisfy the mixture equations to produce the exact same observed distributions $\widetilde{P}_0(x)$ and $\widetilde{P}_1(x)$. Therefore, the true noise rates and class-conditional distributions are non-identifiable without further structural assumptions.
\end{example}

This example demonstrates that simultaneous estimation of $\mT^*$ and $\{\pi_k^*(x)\}$ requires additional structural assumptions on these distributions, which are introduced in the next section.

%% file: sections/estimation.tex
\section{Empirical likelihood based estimation under label noise}
\label{sec:el_general_form}
Motivated by the need for structural assumptions identified in the previous section, we develop an estimation framework that directly recovers $\{w_k^*, \pi_k^*(x)\}_{k=0}^{K-1}$ from noisy labeled data. 
Our approach combines density ratio modeling with empirical likelihood and does not require prior knowledge of the label corruption probabilities.

\subsection{Exponential tilting and density ratio model}
\label{sec:DRM_for_noisy_label}
Consider $(X,Y,\widetilde{Y})$ where $(X,Y)\sim P^*_{X,Y}$ and $\widetilde{Y}$ satisfies Assumption~\ref{assumption:instance_independent_noise}.
Let $g(x)\in \sR^{d}$ be some representation of feature $x$ such that:
\begin{equation}
\label{eq:posterior}
\pi_k^*(x):=\sP(Y=k| X=x) = \frac{\exp((\gamma_k^*)^{\dagger} + \langle \beta_k^*,g(x)\rangle)}{\sum_{k'\in [K]} \exp((\gamma_{k'}^*)^{\dagger} + \langle \beta_{k'}^*,g(x)\rangle)}.
\end{equation}
with $(\gamma_0^*)^{\dagger}=0$ and $\beta_0^*=0$ for identifiability. 
This leads to the exponential tilting relationship between the class-conditional distributions:
\be
\label{eq:drm}
dP_k^*/dP_0^*(x) = \exp(\gamma_k^* + \langle \beta_k^*, g(x)\rangle),
\ee
where $dP_k^*/dP_0^*$ denotes the Radon-Nikodym derivative of $P_k^*$ with respect to $P_0^*$, $\gamma_k^* = (\gamma_k^*)^{\dagger} + \log (w_0^*/w_k^*)$.
With this reparameterization, we also have $\gamma_0^*= 0$.

To ensure $P_k^*$s are valid probability measures, the parameters $\bgamma^* = \{\gamma_k^*\}_{k=1}^{K-1}$ and $\bbeta^* = \{\beta_k^*\}_{k=1}^{K-1}$ must satisfy the constraints:
\be
\label{eq:drm_constraints}
\int \exp(\gamma_k^* + \langle\beta_k^*,g(x)\rangle)P_0^*(dx) = 1, \quad \forall~k\in \{1,\ldots,K-1\}.
\ee

Under the instance-independent noise assumption (Assumption~\ref{assumption:instance_independent_noise}), the contaminated distributions of $\widetilde{X}^{l}\triangleq X |(\widetilde{Y}=l) \sim \widetilde{P}_l^*$ relate to the true distributions as follows:
\begin{equation}
\label{eq:drm_noisy}
\frac{d\widetilde{P}_l^*}{dP_0^*}(x) 
=\sum_{k\in[K]} M_{lk}^* \frac{d P_{k}^*}{dP_0^*}(x) 
=\frac{1}{\widetilde{w}_l^*}\sum_{k\in[K]} \left\{w_k^* T_{lk}^*\exp(\gamma_k^* + \langle \beta_k^*, g(x)\rangle)\right\}.
\end{equation}
Both~\eqref{eq:drm} and~\eqref{eq:drm_noisy} are instances of the density ratio model (DRM)~\citep{anderson1979multivariate,qin2017biased}, where all measures ($\{\widetilde{P}_l^*\}_{l=0}^{K-1}$, $\{P_k^*\}_{k=1}^{K-1}$) are connected to the base measure $P_0^*$ via a parametric form.
The DRM is highly flexible and encompasses many commonly used distribution families, including all members of the exponential family with appropriate specifications for $P_0^*$ and $g(x)$. 
For example, if the base distribution is normal and $g(x) = (x, x^2)^{\top}$, the DRM reduces to the normal distribution family.
\begin{remark}[Basis function selection]
The function $g$ is commonly referred to as the basis function.
In our simulations and experiments, the simple choice $g(x)=x$ performs well across a range of scenarios.
For more complex data, such as images, performance may be improved by using features extracted from pretrained deep neural networks or by adopting data-adaptive approaches, such as that of~\citet{zhang2022density}.
\end{remark}

\subsection{Identifiability}
\label{sec:identifiability}
In this section, we study the identifiability of the model in~\eqref{eq:drm_noisy}. 
We first show that the DRM is not identifiable when the features are linearly dependent, highlighting that the absence of collinearity is essential for establishing identifiability of the model in~\eqref{eq:drm_noisy}.
\begin{example}[Unidentifiability due to linearly dependent features]
\label{ex:drm_unidentifiability}
Recall our density ratio model formulation from \eqref{eq:drm}: $dP_1^*/dP_0^*(x) = \exp(\gamma_1^* + \langle \beta_1^*, g(x)\rangle)$.
Suppose the chosen feature representation $g(x) = (g_1(x), g_2(x))^\top \in \sR^2$ is not linearly independent across the support of $P_0^*$. Specifically, assume the data lies on a manifold where $g_1(x) + g_2(x) = c$ for some constant $c$.

Let $(\gamma_1^*, \beta_{11}^*, \beta_{12}^*)$ be the true parameters generating the data. The exponential tilt is: $\gamma_1^* + \beta_{11}^* g_1(x) + \beta_{12}^* g_2(x)$.
Now, consider a different set of parameters defined by shifting the feature coefficients by an arbitrary constant $\delta \neq 0$: $\beta_{11}^{**} = \beta_{11}^* + \delta$, $\beta_{12}^{**} = \beta_{12}^* + \delta.$
Evaluating the tilt with these new parameters yields:
\[
\begin{split}
\gamma_1^* + \beta_{11}^{**} g_1(x) + \beta_{12}^{**} g_2(x) &= \gamma_1^* + (\beta_{11}^* + \delta)g_1(x) + (\beta_{12}^* + \delta)g_2(x) \\
&= (\gamma_1^* + \delta c) + \beta_{11}^* g_1(x) + \beta_{12}^* g_2(x).    
\end{split}
\]
By defining a new baseline parameter $\gamma_1^{**} = \gamma_1^* + \delta c$, we perfectly absorb the shift. The parameter sets $(\gamma_1^*, \beta_{11}^*, \beta_{12}^*)$ and $(\gamma_1^{**}, \beta_{11}^{**}, \beta_{12}^{**})$ produce the exact same density ratio $\exp(\gamma_1^{**} + \langle \beta_1^{**}, g(x)\rangle) \equiv \exp(\gamma_1^* + \langle \beta_1^*, g(x)\rangle)$ for all observable $x$. Thus, without assumptions guaranteeing the richness of the feature subspace (i.e., ruling out collinearity), the tilt parameters are non-identifiable.
\end{example}
This example motivates the introduction of a set of mild but essential assumptions under which identifiability can be restored, as formalized in the following theorem.
\begin{theorem}[Identifiability]
\label{thm:identifiability}
Under mild assumptions (Assumptions~\ref{assump:distinct}--~\ref{assump:recoverability}), the parameters $(\mM^*, \bgamma^*, \bbeta^*)$ and the reference distribution $P_0^*$ in~\eqref{eq:drm_noisy} are strictly identifiable.
\end{theorem}
The formal mathematical statements of these assumptions, along with the complete proof, are deferred to Appendix~\ref{app:identifiability}.
Strict identifiability implies a true one-to-one mapping between the parameter space and the probability distributions of the observed data. 
For our model, this means there is exactly one unique set of parameters $(\mM^*, \bgamma^*, \bbeta^*)$ and one unique reference distribution $P_0^*$ that can generate the observed data, completely eliminating the label-switching (permutation) ambiguity that typically plagues mixture models.
To establish this strict identifiability, we require three standard regularity conditions. 
Broadly, these conditions ensure that the latent mixture components are physically distinct and that the observed data provides sufficient variation to decouple them. 
Specifically, we require the feature coefficients of the exponential tilts to be strictly distinct (Assumption~\ref{assump:distinct}) and the feature mapping to explore a sufficiently rich subspace (Assumption~\ref{assump:richness}). 
These are common requirements in the exponential family and finite mixture model literature to prevent component collapse and guarantee that the underlying densities are linearly independent. 
Furthermore, we require the confusion matrix $\mM^*$ to be full rank, with specific structural constraints on its inverse (Assumption \ref{assump:recoverability}). 
Crucially, this final condition on the inverse matrix not only ensures the signal from the true latent components is preserved, but it effectively breaks the permutation symmetry inherent in standard mixture models, allowing for the strict, unique identification of all parameters.

\subsection{Empirical likelihood based inference}
\label{sec:el_inference}
To avoid restrictive parametric assumptions on the base distribution, we adopt a non-parametric approach on the estimation of $\pi_k^*(x)$ and the rest of the parameters $\mT^*$ and $\vw^*$, given a set of i.i.d. samples $\{(X_i,\widetilde{Y}_i)\}_{i=1}^{n}$ from $P^*_{X,\widetilde{Y}}$.

The log-likelihood function based on the contaminated set $\gD=\{(X_i,\widetilde{Y}_i)\}_{i=1}^{n}\iid P^*_{X,\widetilde{Y}}$ is $\ell =\sum_{i=1}^{n}\sum_{l=0}^{K-1}\left[\mathbbm{1}(\widetilde{Y}_i=l)\log\left\{\widetilde{w}_l^* \widetilde{P}_l^*(\{X_i\})\right\}\right]$.
Using the exponential tilting relationship in~\eqref{eq:drm_noisy}, we obtain
\[
\ell=\sum_{i=1}^{n}\sum_{l\in[K]}\mathbbm{1}(\widetilde{Y}_i=l)\log \left[\sum_{k\in[K]} \left\{w_k^* T_{lk}^*\exp(\gamma_k^* + \langle \beta_k^*,g(X_i)\rangle)\right\}\right] + \sum_{i=1}^{n}\log P_0^*(\{X_i\}).
\]
The empirical likelihood (EL)~\citep{owen2001empirical,qin1994empirical} allows us to estimate parameters without assumptions on the form of $P_0^*$.
Let $p_i = P_0^*(\{X_i\})$ for $i = 1, \ldots, n$ and view them as parameters.
Let $\vw = \{w_k\}_{k=0}^{K-1}$ and $\btheta = \{\vw, \bgamma, \bbeta, \mT\}$\footnote{
Proposition~\ref{proposition:x_conditional_relationships} demonstrates a fundamental duality: the pair $(\mM^*, \widetilde{\vw}^*)$ contains equivalent information to $(\mT^*, \vw^*)$. This means we can focus our estimation efforts on either parameter group, as the other can be derived accordingly.
We consider $(\mT^*,\vw^*)$ since $\vw^*$ is used directly in the NPMC numerical solver in Section~\ref{sec:NPMC_review}.
}.
The log-EL function becomes:
\be
\label{eq:log-EL}
\ell(\vp,\btheta)=\sum_{i=1}^{n}\sum_{l\in[K]}\mathbbm{1}(\widetilde{Y}_i=l)\log \left[\sum_{k\in[K]} \left\{w_k T_{lk}\exp(\gamma_k + \beta_k^{\top}g(X_i))\right\}\right] + \sum_{i=1}^{n}\log p_i.
\ee
The feasible values of $\vp$ must satisfy the following constraints:
\be
\label{eq:constraints}
p_{i}\geq 0,~\sum_{i=1}^{n}p_i=1,~\sum_{i=1}^{n}p_i\exp(\gamma_k +\beta_k^{\top}g(X_i))=1, ~\forall~k=1,\ldots,K-1.
\ee
The first two conditions in~\eqref{eq:constraints} ensure that $P_0^*$ is a valid probability measure, while the last follows from~\eqref{eq:drm_constraints}.

The inference for $\btheta$ is usually made by first profiling the log-EL with respect to $\vp$.
That is, we define $p\ell_{n}(\btheta)=\sup_{\vp}\ell_n(\vp, \btheta)$ subject to the constraints in~\eqref{eq:constraints}. 
By the Lagrange multiplier method (see Appendix~\ref{app:profile_loglik} for details), the maximizer is attained when
\be
\label{eq:vp_lagrange}
p_i = p_i(\btheta) = n^{-1}\left[1 + \sum_{k=1}^{K-1}\nu_k \left\{\exp(\gamma_k +\beta_k^{\top}g(X_i))-1\right\}\right]^{-1}
\ee
where the Lagrange multipliers $\nu_1,\ldots,\nu_{K-1}$ are the solution to
\be
\label{eq:lagrange-multiplier-solution}
\sum_{i=1}^{n}\frac{\exp(\gamma_k +\beta_k^{\top}g(X_i)) -1}{1+\sum_{k'=1}^{K-1}
\nu_{k'}\left\{\exp(\gamma_{k'} +\beta_{k'}^{\top}g(X_i))-1\right\}} = 0,\quad k=1,\ldots, K-1.
\ee
The profile log-EL of $\btheta$ (after maximizing out $p_i$) is given by
\[
p\ell_{n}(\btheta) =\sum_{i=1}^{n}\sum_{l\in[K]}\mathbbm{1}(\widetilde{Y}_i=l)\log \left[\sum_{k\in[K]} \left\{w_k T_{lk}\exp(\gamma_k + \beta_k^{\top}g(X_i))\right\}\right] + \sum_{i=1}^{n}\log p_i(\btheta),
\]
and the Maximum Empirical Likelihood Estimator (MELE) of $\btheta$ is defined as 
\be
\label{eq:mele_estimator}
\widehat\btheta = \argmax p\ell_{n}(\btheta).
\ee
The conditional distribution estimators are then
\be
\label{eq:conditional_distribution_estimate}    
\widehat P_k =\sum_{i=1}^{n} p_i(\widehat\btheta)\exp(\widehat\gamma_k + \widehat\beta_k^{\top}g(X_i))\delta_{X_i},\quad\widehat{\pi}_k(x) =\frac{\exp(\widehat{\gamma}_k^{\dagger} + \widehat{\beta}_l^{\top}g(x))}{\sum_{k'} \exp(\widehat{\gamma}_{k'}^{\dagger} + \widehat{\beta}_{k'}^{\top}g(x))}, 
\ee
where $\widehat{\gamma}_k^{\dagger} = \widehat{\gamma}_k + \log(\widehat{w}_0/\widehat{w}_k)$.

\subsection{Statistical guarantees}
\begin{theorem}[Estimator rate of convergence]
\label{thm:estimator_asymptotic_normality} 
Let $\btheta^* = (\vw^*, \bgamma^*, \bbeta^*, \mT^*)$ be the true values of $\btheta$ and $\widehat\btheta$ be its MELE estimator in~\eqref{eq:mele_estimator}.
We have $\sqrt{n}(\widehat\btheta-\btheta^*)\to N(0,\Sigma)$ as $n\to\infty$ where $\Sigma>0$ is some non-degenerated covariance matrix.
\end{theorem}
The proof of the theorem is deferred to Appendix~\ref{app:asymptotic_normality}.
The result establishes that the MELE estimator is consistent and asymptotically normal, achieving the optimal convergence rate.
This asymptotic normality plays a key role in deriving the NP oracle inequalities for the final classifier.

\subsection{EM algorithm for profile EL maximization}
\label{sec:EM_algorithm}
Since the true labels \((Y_1,\ldots,Y_n)\) are unobserved, $\widetilde{P}_l^*$ become mixtures of $P_0^*, \ldots, P_{K-1}^*$.
This complicates the maximization of~\eqref{eq:mele_estimator} due to non-convexity. 
To address the challenge, we treat the true labels as missing values and employ an Expectation-Maximization (EM) algorithm for its numerical computation. 
We summarize the algorithm in Algorithm~\ref{alg:noisy_label} and explain it below.
\begin{algorithm}[htbp]
\caption{EM algorithm for classification with noisy labels via empirical likelihood}
\label{alg:noisy_label}
\begin{algorithmic}[1]
\Require Noisy labeled data $\{(X_i,\widetilde{Y}_i)\}_{i=1}^n$, basis $g(x)$, \# of classes $K$, convergence threshold $\epsilon$
\Ensure Estimates $\widehat{w}_k$, $\widehat{\pi}_{k}(x)$ for all $k \in [K]$

\State \textbf{Initialize:} Choose initial parameters $\btheta^{(0)} = (\vw^{(0)}, \bgamma^{(0)}, \bbeta^{(0)}, \mT^{(0)})$
\Statex Set initial profile log-EL $p\ell_n(\btheta^{(0)})\gets -\infty$

\Repeat
\State \textbf{E-step:} Compute posterior probabilities $\omega_{ik}^{(t)}$ for each observation according to~\eqref{eq:e_step_multiclass}

\State \textbf{M-step:} Update parameters:
\State Update class proportions $w_k^{(t+1)}$ via~\eqref{eq:m_step_class_probability}

\State Update transition matrix $\mT^{(t+1)}$ via~\eqref{eq:m_step_transition_matrix}
        
\State Update $(\bar{\bgamma}^{(t+1)}, \bbeta^{(t+1)})$ by solving the weighted multinomial logistic regression in~\eqref{eq:weighted_multinomial_logistic}
\State Adjust and update to obtain $\gamma_k^{(t+1)} \gets \bar\gamma_k^{(t+1)} - \log(\sum_{i=1}^n \omega_{ik}^{(t)}/\sum_{i=1}^n \omega_{i0}^{(t)})$

\State Update weights $p_i^{(t+1)}$ via~\eqref{eq:m_step_pi} to compute $p\ell_n(\btheta^{(t+1)})$

\Until{$p\ell_n(\btheta^{(t+1)}) - p\ell_n(\btheta^{(t)}) < \epsilon$}
\State \textbf{Return:}
\Statex $\widehat{w}_k \gets w_k^{(t+1)}$, $\widehat{\pi}_{k}(x) \gets \frac{\exp(\widehat{\gamma}_k^{\dagger} + \widehat{\beta}_k^{\top}g(x))}{\sum_{k'} \exp(\widehat{\gamma}_{k'}^{\dagger} + \widehat{\beta}_{k'}^{\top}g(x))}$ where $\widehat{\gamma}_k^{\dagger} = \gamma_k^{(t+1)} + \log\frac{w_0^{(t+1)}}{w_k^{(t+1)}}$
\end{algorithmic}
\end{algorithm}

The EM framework is particularly well-suited for problems with latent variables, as it iteratively handles the missing information via: an \textbf{expectation} step that computes conditional probabilities of the missing true labels; a \textbf{maximization} step that updates parameter estimates.

If the complete data $\{(X_i,Y_i,\widetilde{Y}_i)\}_{i=1}^{n}$ were available, the complete data log-EL would be
\[
\ell_{n}^{c}(\vp, \btheta) = \sum_{i=1}^{n} \sum_{l\in [K]} \sum_{k\in [K]} \mathbbm{1}(\widetilde{Y}_i=l, Y_i=k) \log \sP(X = X_i, \widetilde{Y}_i=l, Y_i=k).
\]
Using the chain rule and Assumption~\ref{assumption:instance_independent_noise}, we decompose the joint probability as:
\[
\sP(X = X_i, \widetilde{Y}_i=l, Y_i=k) 
= \sP(X = X_i | Y_i = k) T_{lk}^* w_k^* = P_0^*(\{X_i\}) \exp\{\gamma_k^* + \langle \beta_k^*, g(X_i)\rangle\} T_{lk}^* w_k^*,
\]
where the last equality is from the DRM in~\eqref{eq:drm}. 
This yields the complete data profile log-EL:
\begin{equation*}
\ell_{n}^{c}(\vp,\btheta) =\sum_{i=1}^{n} \sum_{l\in [K]} \sum_{k\in [K]}\mathbbm{1}(\widetilde{Y}_i=l)\mathbbm{1}(Y_i=k)\left[\log p_{i} + \{\gamma_k + \beta_k^{\top}g(X_i)\} + \log T_{lk} + \log w_k\right].
\end{equation*}
We define the profile complete data log-EL as  $p\ell_{n}^{c}(\btheta) = \sup_{\vp} \ell_{n}^{c}(\vp, \btheta)$, where the maximization is subject to the constraints in~\eqref{eq:constraints}. 
Using the same derivation for~\eqref{eq:vp_lagrange}, the profile complete data log-EL becomes
\begin{equation*}
p\ell_{n}^{c}(\btheta) =\sum_{i=1}^{n} \sum_{l\in [K]} \sum_{k\in [K]}\mathbbm{1}(\widetilde{Y}_i=l)\mathbbm{1}(Y_i=k)\left(\gamma_k + \beta_k^{\top}g(X_i) + \log T_{lk} + \log w_k\right) + \sum_{i=1}^{n}\log p_i(\btheta)
\end{equation*}
where $p_i(\btheta)$ is given in~\eqref{eq:vp_lagrange}, and the $\nu_k$'s are the solutions to~\eqref{eq:lagrange-multiplier-solution}.

The EM algorithm is an iterative procedure initialized at $\btheta^{(0)}$. 
At iteration $t$, given the current estimate $\btheta^{(t)}$, it proceeds in two steps:
\begin{itemize}[leftmargin=*]
\item \textbf{E-step (Expectation step)}: The missing data is ``filled in'' by computing their posterior given the observed data and $\btheta^{(t)}$:
\begin{equation}
\label{eq:e_step_multiclass}
\omega_{ik}^{(t)} 
=\sE(\mathbbm{1}(Y_i=k)| \btheta^{(t)};\widetilde{Y}_i, X_i) =\frac{\exp(\gamma_k^{(t)}+\langle \beta_k^{(t)}, g(X_i)\rangle)T_{\widetilde{Y}_i,k}^{(t)}w_k^{(t)}}{\sum_{k'=1}^{K}\exp(\gamma_{k'}^{(t)}+\langle \beta_{k'}^{(t)}, g(X_i)\rangle)T_{\widetilde{Y}_i,k'}^{(t)}w_{k'}^{(t)}}.
\end{equation}
Then the expected complete data profile empirical log-EL at iteration $t$ is
\begin{equation*}
\begin{split}
Q(\btheta;\btheta^{(t)})=&~\sE\{p\ell_n^{(c)}(\btheta)| \btheta^{(t)};\gD\}\\
=&~\sum_{i=1}^{n}\sum_{l\in[K]}\mathbbm{1}(\widetilde{Y}_i=l)\sum_{k\in [K]}\omega_{ik}^{(t)}(\log w_k + \gamma_k + \beta_k^{\top}g(X_i) + \log T_{lk})+\sum_{i=1}^{n} \log p_i(\btheta)    
\end{split}
\end{equation*}

\item \textbf{M-step (Maximization step)}: Instead of maximizing log-EL, the M-step maximizes $Q$ with respect to $\btheta$ instead, \ie $\btheta^{(t+1)} = \argmax Q(\btheta, \btheta^{(t)})$.

Since $Q$ is separable in $\vw$, $\mT$, and $(\bgamma, \bbeta)$, we can optimize each set of parameters separately.
Below are the updated parameters; the detailed derivations are provided in Appendix~\ref{app:em_detail}.
\begin{align}
w_{k}^{(t+1)} =&~\frac{1}{n}\sum_{i=1}^{n}\omega_{ik}^{(t)}\label{eq:m_step_class_probability}\\
T_{lk}^{(t+1)} =&~\frac{\sum_{i=1}^{n}\mathbbm{1}(\widetilde{Y}_i=l)\omega_{ik}^{(t)}}{\sum_{i=1}^{n}\omega_{ik}^{(t)}}.\label{eq:m_step_transition_matrix}
\end{align}

The update of $(\bgamma^{(t+1)}, \bbeta^{(t+1)})$ is more involved, as it maximizes
{\small{
\[
Q_3(\bgamma,\bbeta;\btheta^{(t)}) =\sum_{i=1}^{n}\sum_{k\in [K]}\omega_{ik}^{(t)}(\gamma_k + \beta_k^{\top}g(X_i)) -\sum_{i=1}^{n} \log\left\{1+\sum_{k=1}^{K-1}\nu_k \left(\exp(\gamma_k +\beta_k^{\top}g(X_i))-1\right)\right\}.
\]}}
At the optimal point, we show in Appendix~\ref{app:em_detail} that the Lagrange multipliers 
\[\nu_k = n^{-1} \sum_{i=1}^{n} \omega_{ik}^{(t)} := \omega_{\cdot k}^{(t)}.\] 
Thus, the stationary point of $Q_3(\bgamma, \bbeta;\btheta^{(t)})$ coincides with the stationary point of
\[
\widetilde{Q}(\bgamma,\bbeta;\btheta^{(t)}) =\sum_{i=1}^{n}\sum_{k\in [K]}\omega_{ik}^{(t)}(\gamma_k + \beta_k^{\top}g(X_i)) -\sum_{i=1}^{n} \log\left\{\sum_{k\in[K]}w_{\cdot k}\exp(\gamma_k +\beta_k^{\top}g(X_i))\right\}.
\]
Let $\bar\gamma_k = \gamma_k + \log (\omega_{\cdot k}^{(t)}/\omega_{\cdot 0}^{(t)})$, then the function $\widetilde{Q}$ becomes
\be
\label{eq:weighted_multinomial_logistic}
\widetilde{Q}(\bar\bgamma,\bbeta;\btheta^{(t)}) =\sum_{i=1}^{n}\sum_{k\in [K]}\omega_{ik}^{(t)}(\bar\gamma_k + \beta_k^{\top}g(X_i)) -\sum_{i=1}^{n} \log\left\{\sum_{k\in[K]}\exp(\bar\gamma_k +\beta_k^{\top}g(X_i))\right\},
\ee
which is the weighted log-likelihood of multinomial logistic regression model based on the dataset in Table~\ref{tab:weighted_glm} in Appendix~\ref{app:em_detail}.
Thus, $(\bar{\bgamma}^{(t+1)}, \bbeta^{(t+1)})$ maximizes the weighted log-likelihood.
We use standar packages in \texttt{R} or \texttt{Python} to find the numerical value.
Finally, we get 
\be
\gamma_k^{(t+1)} = \bar\gamma_k^{(t+1)} - \log (\omega_{\cdot k}^{(t)}/\omega_{\cdot 0}^{(t)}),~k\in[K].
\ee
The weights are 
\be
\label{eq:m_step_pi}
p_i^{(t+1)} = p_i(\btheta^{(t+1)}) = n^{-1}\left\{\sum_{k\in [K]} \omega_{\cdot k}^{(t)} \exp(\gamma_k^{(t+1)}+\langle \beta_k^{(t+1)}, g(X_i)\rangle)\right\}^{-1}
\ee
since $\nu_k = \omega_{\cdot k}^{(t)}$ at the optimal point.
\end{itemize}
The E- and M-steps are iterated until the change in the profile log-EL falls below a threshold.
This criterion is met since the EM algorithm produces estimates that monotonically increase the true objective—the profile log-EL, which we detail below.
    
\begin{proposition}[Convergence of EM algorithm]
\label{thm:em_convergence}
With the EM algorithm described above, we have for $t\geq 1$, $p\ell_n(\btheta^{(t+1)})\geq p\ell_n(\btheta^{(t)})$.
\end{proposition}
The proof of Proposition~\ref{thm:em_convergence} is in Appendix~\ref{app:em_convergence}.
Note that $p\ell_n(\btheta^{(t+1)})$ is a sum of the logarithm of probabilities, implying that $p\ell_n(\btheta)\leq 0$ for any $\btheta$ that is feasible.
Proposition~\ref{thm:em_convergence} then guarantees convergence to at least a local maximum for the given initial value $\btheta^{(0)}$.
We recommend using multiple initial values to explore the likelihood
function to ensure that the algorithm reaches the global maximum. 
Second, in practice, we may stop the
algorithm when the increment in the log-EL after an iteration is no greater than, say, $10^{-6}$.

\begin{remark}[Incorporating prior knowledge]
When prior knowledge about the transition matrix is available (e.g., $T_{kk}^*\geq \xi_k >0$), the EM algorithm’s M-step can be modified to incorporate such prior information. 
The analytical solution for this constrained update is provided in~\ref{app:em_prior_knowledge_transition_matrix}.
To further improve numerical stability and convergence speed, penalties can be imposed on the diagonal elements $\{T_{kk}\}_{k=0}^{K-1}$.
The corresponding analytical solution for this penalized update is detailed in Appendix~\ref{app:emp_penalization}.
\end{remark}

%% file: sections/binary.tex
\section{Neyman-Pearson binary classification}
\label{sec:binary}
This section focuses on binary classification, a fundamental setting widely used in practice. 
We begin by reviewing the NP framework for clean-label binary classification, where the optimal classifier is derived from a likelihood ratio test. 
We then incorporate the EL-based estimator from the previous section into this framework to achieve type I error control under noisy labels. 
The resulting classifier provably controls the empirical type I error asymptotically.

\subsection{Revisiting NP binary classification with clean labels}
We begin by reviewing the NP binary classification framework without label noise. 
Consider random variables $(X,Y)\sim P^*_{X,Y}$, where $w^* = \mathbb{P}(Y=1)$, $ \pi^*(x) = \mathbb{P}(Y=1|X=x)$, and $X^{y} \triangleq X|(Y=y) \sim P_y^*$.
For a classifier $\phi: \mathcal{X} \to \{0,1\}$, the NP paradigm solves:
\begin{equation}
\label{eq:binary_np_clean}
\min_{\phi}~~P_1^*(\{\phi(X)\neq 1\}) \quad\quad \text{s.t.}~~P_0^*(\{\phi(X)\neq 0\}) \leq \alpha,   
\end{equation}
where $\alpha \in (0,1)$ is the target type I error, and the objective corresponds to the type II error. This is a special case of~\eqref{eq:NPMC} with $\rho_0=0$, $\rho_1=1$, $\gS=\{0\}$, and $\alpha_0=\alpha$.

When the joint distribution $P_{X,Y}^*$ is known, the NP lemma~\citep[Lemma 1]{sadinle2019least} establishes that the optimal classifier takes the form:
\begin{equation}
\label{eq:binary_optimal_classifier}
\phi_{\lambda}^{*}(x) = \mathbbm{1}\left(\frac{dP_1^*}{dP_0^*}(x) \geq \lambda \right)
= \mathbbm{1}\left(\lambda \leq r(x, \pi^*, w^*)\right)
\end{equation}
where $r(x,\pi, w) = \{(1-w) \pi(x)\}/\{w(1- \pi(x))\}$.
The value of $\lambda$ controls the trade-off between type I and type II errors: a smaller $\lambda$ increases the type I error while reducing the type II error.  
Using Bayes' theorem and the law of total expectation, it can be shown that the optimal $\lambda^*$ in $\phi^*_{\lambda}$ is determined by solving:
\be
\label{eq:type1_error}
\alpha = P_0^*(\{\phi_{\lambda}^*(X) \neq 0\}) = (1-w^*)^{-1}\sE_{X}\left[\{1- \pi^*(X)\}\mathbbm{1}\left(\lambda \leq r(X, \pi^*, w^*)\right)\right].
\ee

At the sample level, let $\gD=\{(X_i, Y_i)\}_{i=1}^{n} \overset{\text{i.i.d.}}{\sim} P_{X,Y}$.
Given estimators $\widehat{w}$ and $\widehat{\pi}(x)$ for $w$ and $\pi(x)$ based on $\gD$, respectively, we estimate $\lambda^*$ by solving the empirical estimate of~\eqref{eq:type1_error}:
\begin{equation}
\label{eq:binary_npc_clean_lambda_sample}
\frac{1}{n}\sum_{i=1}^n \{1-\widehat{\pi}(X_i)\}\mathbbm{1}\left(\lambda \leq r(X_i,\widehat{\pi}, \widehat{w})\right) = \alpha(1-\widehat w).
\end{equation}
The left-hand side of~\eqref{eq:binary_npc_clean_lambda_sample} is piecewise constant in $\lambda$, with jumps at $\{r(X_i,\widehat{\pi}, \widehat{w})\}_{i=1}^{n}$. 
The estimator $\widehat{\lambda}$ corresponds to the smallest jump point where the constraint is satisfied.

Once $\widehat{\lambda}$ is obtained, we substitute it, along with the estimators for $\pi(x)$ and $w$, into~\eqref{eq:binary_optimal_classifier} to obtain the error-controlled classifier, denoted as $\widehat{\phi}_{\widehat{\lambda}}$.

\subsection{NP binary classification with noisy labels with statistical guarantees}
Since our EL-based estimation allows us to directly estimate $\widehat{w}$ and $\widehat{\pi}$ from noisy labeled observations, we can simply plug in these estimators, as discussed in the clean label case, to obtain the final classifier with controlled type I error.

Specifically, let $\widehat{w}, \widehat{\gamma}, \widehat{\beta}$ be the MELE estimators defined in~\eqref{eq:mele_estimator} with $\widehat{\gamma}^{\dagger} = \widehat{\gamma} + \log\{(1-\widehat{w})/\widehat{w}\}$.
We define
\be
\label{eq:binary_conditional}
\widehat{\pi}(x) = \left[1 + \exp(-\{\widehat{\gamma}^{\dagger} + \widehat{\beta}^{\top}g(x)\})\right]^{-1}.
\ee
Let $\widehat{\lambda}$ be the solution to~\eqref{eq:binary_npc_clean_lambda_sample} with the estimators $\widehat{\pi}$ and $\widehat{w}$ as defined above. 
Our final binary classifier with type I error control is given by
\be
\label{eq:binary_el_classifier}
\widehat{\phi}_{\widehat{\lambda}}(x) = \mathbbm{1}(\widehat{\lambda} < r(x,\widehat{\pi}, \widehat{w})) = \mathbbm{1}\left(\widehat{\pi}(x) \geq \frac{\widehat{\lambda}\widehat{w}}{1-\widehat{w}+\widehat{\lambda}\widehat{w}}\right),
\ee
where the complete expression for the threshold is implied by the context.

We now show that the type I error of this estimator is asymptotically controlled at the target level. 
The following assumptions are required for our theoretical results:
\begin{assumption}[Compactness]
\label{assumption:compactness}
The parameter space {\small{$\Theta=\{(\gamma, \beta): \gamma\in \sR, \gamma\in\sR^{d}\}$}} is compact.
\end{assumption}

\begin{assumption}[Bounded second moment]
\label{assumption:bounded_second_moment}
There exists a constant $R>0$ s.t. $\sE\|g(X)\|_2^2\leq R^2$.
\end{assumption}

\begin{theorem}[Error control guarantee]
\label{thm:binary_type1}
Let $\widehat\phi_{\widehat\lambda}(x)$ be the binary classifier defined in~\eqref{eq:binary_el_classifier}. 
Under Assumptions~\ref{assumption:compactness} and~\ref{assumption:bounded_second_moment}, the type I error is controlled as: $P_{0}^*\left(\left\{\widehat\phi_{\widehat\lambda}(X)\neq 0\right\}\right)\leq \alpha + O_{p}(n^{-1/2})$.
\end{theorem}
The proof is deferred to Appendix~\ref{app:binary_proof}. 
This theorem establishes that our classifier asymptotically controls the type I error at the target level $\alpha$ under Assumptions~\ref{assumption:compactness} and~\ref{assumption:bounded_second_moment}, which are standard in the literature (e.g.,~\citealt{tian2024neyman}). 
These conditions are realistic in practice: Assumption~\ref{assumption:compactness} requires only that the parameter space be closed and bounded, which can be enforced by imposing $|\gamma| + \|\beta\|_2 \leq B$ for a sufficiently large $B$, while Assumption~\ref{assumption:bounded_second_moment} demands only finite second moments of the feature map $g(X)$. 
Under these mild conditions, the Rademacher complexity of the class $\{\pi(x): (\gamma,\beta)\in \Theta\}$ is of order $n^{-1/2}$. 
Combined with the fact that our EL-based estimators $(\widehat{\gamma},\widehat{\beta},\widehat{w})$ converge at the usual $n^{-1/2}$ rate, this ensures that the final classifier $\widehat{\phi}_{\widehat{\lambda}}$ preserves the Neyman--Pearson guarantee up to a vanishing $O_p(n^{-1/2})$ term.

\begin{remark}[NP umbrella algorithm type control]
The NP umbrella algorithm was developed by~\citet{tong2018neyman} to control errors in binary classification with uncontaminated data, and later extended to the noisy-label setting by~\citet{yao2023asymmetric}. 
When domain knowledge of the noise transition matrix $\mM^{*}$ is available, they estimate the discrepancy between true and corrupted type I errors and propose a label-noise-adjusted NP umbrella algorithm with theoretical guarantees. 
Specifically, for any $\delta > 0$, they show that their classifier $\phi_{\delta}$ satisfies
\be
\label{eq:np_control}
\sP(P_0^*(\{\phi_{\delta}(X)\neq 0\})\leq \alpha) > 1-\delta + O(\varepsilon(n)),
\ee
where $\varepsilon(n)$ vanishes as $n$ increases.
To our knowledge, this is the only work addressing error control under noisy labels. 
It is therefore natural to compare our error control guarantee (Theorem~\ref{thm:binary_type1}) with their result. 
At first glance, our guarantee appears weaker than that of \citet{yao2023asymmetric}. 
In particular, we show that $\forall \epsilon > 0$, there exist constants $C$ and $N$ such that for all $n \geq N$,
\[\sP(P_0^*(\{\widehat\phi_{\widehat{\lambda}}(X)\neq 0\})\leq \alpha + Cn^{-1/2}) > 1-\epsilon.\]
Comparing this with~\eqref{eq:np_control}, the distinction lies in the quantities being controlled.
Our analysis bounds the deviation of the empirical type I error from the target level, which introduces a stochastic fluctuation of order $n^{-1/2}$. 
In contrast, the NP umbrella framework provides a stricter guarantee, ensuring with high probability that the population type I error never exceeds the target level.
Importantly, our EL-based estimator is flexible and can be paired with different error control methods: when combined with the NP umbrella algorithm, it yields the same guarantee as in \citet{yao2023asymmetric}. 
Moreover, our method does not require prior knowledge of $\mM^*$. 
Further discussion and empirical results are provided in Appendix~\ref{app:npumbrella}.
\end{remark}

\subsection{Simulation study}
\label{sec:binary_exp}
In this section, we evaluate the performance of our proposed classifier~\eqref{eq:binary_el_classifier} on simulated data and compare it with existing methods. 
Following the experimental settings of \citet{yao2023asymmetric}, we consider three distinct distributions for the feature variables $X^{y}$:
\begin{itemize}[leftmargin=*]
\item \textbf{Case A: Gaussian distribution}. 
Let $X^0 \sim N(\mu_0^*, \Sigma^*)$ and $X^1 \sim N(\mu_1^*, \Sigma^*)$, where $\mu_0^* = (0,0,0)^\top$, $\mu_1^* = (1,1,1)^\top$, and $\Sigma^*$ is a Toeplitz matrix with diagonal elements being $2$, first off-diagonal elements being $-1$, and remaining elements $0$.

\item \textbf{Case B: Uniform distribution within circles}. 
Let $X^0$ and $X^1$ be uniformly distributed within unit circles centered at $(0,0)^\top$ and $(1,1)^\top$, respectively. 

\item \textbf{Case C: $t$-distribution}. 
Let $X^0$ and $X^1$ follow multivariate non-central $t$-distributions with location parameters $\mu_0^*$ and $\mu_1^*$, shape matrix $\Sigma^*$, and 15 degrees of freedom. The values of $\mu_0^*$, $\mu_1^*$, and $\Sigma^*$ match those in Case A. 
\end{itemize}
Using the relationship $\widetilde{P}_0^* = m_0^* P_0^* + (1-m_0^*) P_1^*$ and $\widetilde{P}_1^* = m_1^* P_0^* + (1-m_1^*) P_1^*$,
where $m_0^* = \sP(Y=0|\widetilde{Y}=0)$ and $m_1^* = \sP(Y=0|\widetilde{Y}=1)$ represent the label corruption probabilities, we generate samples with corrupted labels. 
For each setting, we create a training sample of $2n$ observations, equally split between the corrupted class $0$ and corrupted class $1$. 
We vary the sample size $n$ from $1{,}000$ to $5{,}000$ to examine the effect of sample size on performance. 
To accurately approximate the true type I and type II errors, we generate a large evaluation set consisting of $20{,}000$ true class $0$ observations and $20{,}000$ true class $1$ observations. 
For each combination of distribution and sample size, we repeat the experiment $R=500$ times and report the average type I and type II errors across all repetitions.

\begin{table}[!ht]
\centering
\caption{Type I and II error for binary classification based on $500$ repetitions. 
NPC$^{*}$ represents the oracle method with known $m_0^*$ and $m_1^*$. For each setting, the method achieving type I error closest to the nominal level is highlighted in boldface.}
\label{tab:binary}
\resizebox{\textwidth}{!}{
\begin{tabular}{ccccc>{\bfseries}c@{\hspace{2em}}ccc>{\bfseries}c@{\hspace{2em}}ccc>{\bfseries}c@{\hspace{2em}}ccc>{\bfseries}c}
\toprule
\multirow{3}{*}{Case} &\multirow{3}{*}{$n$} & \multicolumn{4}{c}{$m_0^*=0.95,m_1^*=0.05$} & \multicolumn{4}{c}{$m_0^*=0.95,m_1^*=0.05$} & \multicolumn{4}{c}{$m_0^*=0.9,m_1^*=0.1$} & \multicolumn{4}{c}{$m_0^*=0.9,m_1^*=0.1$} \\
&& \multicolumn{4}{c}{$\alpha=0.05,\delta=0.05$} & \multicolumn{4}{c}{$\alpha=0.1,\delta=0.1$} & \multicolumn{4}{c}{$\alpha=0.05,\delta=0.05$} & \multicolumn{4}{c}{$\alpha=0.1,\delta=0.1$} \\ 
\cmidrule(lr{1.5em}){3-6}
\cmidrule(lr{1.5em}){7-10}
\cmidrule(lr{1.5em}){11-14}
\cmidrule(lr){15-18}
&&Vanilla & NPC & NPC$^{*}$ & Ours & Vanilla& NPC & NPC$^{*}$ & Ours & Vanilla & NPC & NPC$^{*}$ & Ours & Vanilla & NPC & NPC$^{*}$ & Ours\\ 
\midrule
&\multicolumn{16}{c}{Type I error}\\
\cmidrule(lr){3-18} 
\multirow{3}{*}{A} & 1000 & 0.011 & 0.010 & 0.027 & \textbf{0.051} & 0.049 & 0.045 & 0.076 & \textbf{0.101} & 0.004 & 0.003 & 0.024 & \textbf{0.052} & 0.025 & 0.022 & 0.073 & \textbf{0.102} \\
 & 2000 & 0.014 & 0.012 & 0.033 & \textbf{0.050} & 0.054 & 0.052 & 0.083 & \textbf{0.100} & 0.005 & 0.004 & 0.031 & \textbf{0.050} & 0.028 & 0.027 & 0.081 & \textbf{0.100} \\
 & 5000 & 0.017 & 0.016 & 0.039 & \textbf{0.050} & 0.058 & 0.056 & 0.089 & \textbf{0.100} & 0.006 & 0.005 & 0.038 & \textbf{0.050} & 0.031 & 0.030 & 0.088 & \textbf{0.100} \\
\midrule
\multirow{3}{*}{B} & 1000 & 0.001 & 0.001 & 0.026 & \textbf{0.051} & 0.042 & 0.038 & 0.077 & \textbf{0.105} & 0.000 & 0.000 & 0.021 & \textbf{0.048} & 0.007 & 0.006 & 0.074 & \textbf{0.101} \\
 & 2000 & 0.002 & 0.002 & 0.033 & \textbf{0.051} & 0.047 & 0.046 & 0.083 & \textbf{0.104} & 0.000 & 0.000 & 0.029 & \textbf{0.047} & 0.010 & 0.009 & 0.081 & \textbf{0.101} \\
 & 5000 & 0.004 & 0.004 & 0.039 & \textbf{0.051} & 0.052 & 0.051 & 0.090 & \textbf{0.104} & 0.000 & 0.000 & 0.037 & \textbf{0.047} & 0.013 & 0.013 & 0.088 & \textbf{0.101} \\
\midrule
\multirow{3}{*}{C} & 1000 & 0.014 & 0.012 & 0.028 & \textbf{0.062} & 0.051 & 0.048 & 0.077 & \textbf{0.112} & 0.007 & 0.006 & 0.026 & \textbf{0.063} & 0.029 & 0.027 & 0.075 & \textbf{0.114} \\
 & 2000 & 0.017 & 0.015 & 0.034 & \textbf{0.061} & 0.055 & 0.053 & 0.083 & \textbf{0.111} & 0.008 & 0.007 & 0.032 & \textbf{0.061} & 0.032 & 0.031 & 0.081 & \textbf{0.111} \\
 & 5000 & 0.019 & 0.018 & 0.039 & \textbf{0.060} & 0.059 & 0.058 & 0.089 & \textbf{0.111} & 0.009 & 0.009 & 0.038 & \textbf{0.061} & 0.035 & 0.034 & 0.088 & \textbf{0.111} \\
\midrule
&\multicolumn{16}{c}{Type II error}\\
\cmidrule(lr){3-18} 
\multirow{3}{*}{A} & 1000 & 0.534 & 0.567 & 0.398 & \textbf{0.278} & 0.287 & 0.305 & 0.218 & \textbf{0.171} & 0.689 & 0.730 & 0.427 & \textbf{0.277} & 0.403 & 0.427 & 0.227 & \textbf{0.171} \\
 & 2000 & 0.489 & 0.515 & 0.352 & \textbf{0.279} & 0.269 & 0.276 & 0.201 & \textbf{0.171} & 0.653 & 0.674 & 0.368 & \textbf{0.279} & 0.376 & 0.388 & 0.205 & \textbf{0.172} \\
 & 5000 & 0.459 & 0.470 & 0.319 & \textbf{0.278} & 0.254 & 0.260 & 0.188 & \textbf{0.171} & 0.622 & 0.630 & 0.326 & \textbf{0.279} & 0.357 & 0.363 & 0.190 & \textbf{0.171} \\
\midrule
\multirow{3}{*}{B} & 1000 & 0.352 & 0.418 & 0.180 & \textbf{0.137} & 0.151 & 0.158 & 0.108 & \textbf{0.078} & 0.661 & 0.700 & 0.205 & \textbf{0.142} & 0.244 & 0.274 & 0.113 & \textbf{0.082} \\
 & 2000 & 0.268 & 0.310 & 0.164 & \textbf{0.137} & 0.143 & 0.145 & 0.100 & \textbf{0.079} & 0.610 & 0.637 & 0.172 & \textbf{0.143} & 0.212 & 0.221 & 0.103 & \textbf{0.082} \\
 & 5000 & 0.230 & 0.240 & 0.154 & \textbf{0.137} & 0.136 & 0.138 & 0.092 & \textbf{0.078} & 0.571 & 0.584 & 0.157 & \textbf{0.143} & 0.200 & 0.203 & 0.094 & \textbf{0.082} \\
\midrule
\multirow{3}{*}{C} & 1000 & 0.588 & 0.629 & 0.461 & \textbf{0.282} & 0.321 & 0.341 & 0.247 & \textbf{0.176} & 0.722 & 0.755 & 0.487 & \textbf{0.280} & 0.440 & 0.462 & 0.255 & \textbf{0.176} \\
 & 2000 & 0.545 & 0.571 & 0.406 & \textbf{0.283} & 0.301 & 0.310 & 0.227 & \textbf{0.177} & 0.684 & 0.704 & 0.421 & \textbf{0.284} & 0.413 & 0.422 & 0.233 & \textbf{0.178} \\
 & 5000 & 0.514 & 0.525 & 0.369 & \textbf{0.283} & 0.286 & 0.291 & 0.213 & \textbf{0.176} & 0.656 & 0.666 & 0.378 & \textbf{0.282} & 0.392 & 0.399 & 0.214 & \textbf{0.176} \\
\bottomrule
\end{tabular}}
\end{table}
For our method, we use linear basis functions in all experiments. 
These three cases are designed to evaluate both efficiency and robustness: the DRM assumption holds in Case A, whereas the model is misspecified in Cases B and C.
We compare our method with the Vanilla and NPC algorithms proposed in~\citet{yao2023asymmetric}. 
\begin{itemize}
    \item \textbf{Vanilla}: Ignores label noise and directly applies the NP paradigm to the noisy labeled data.
    \item \textbf{NPC/NPC$^*$}: Both methods account for label noise. NPC$^*$ \textbf{knows} the true $m_0^*$ and $m_1^*$ values, while NPC assumes that the bounds $m_0^{\#} \geq m_0^*$ and $m_1^{\#} \leq m_1^*$ are known. 
\end{itemize}
Following the original method, NPC includes a parameter $\delta$ such that the type I error is controlled with probability at least $1-\delta$.
We set $m_0^{\#} = \max\{m_0^* + \delta/3,1\}$ and $m_1^{\#} = \min\{m_1^* - \delta/3,0\}$, and examine various combinations of $(m_0^*, m_1^*, \delta, \alpha)$. 
Throughout all experiments, we fix $g(x) = x$ for our algorithm.
For the Vanilla and NPC algorithms, we employ different base classifiers to estimate $\widetilde\pi(x)^*$: LDA for Case A and logistic regression for Cases B and C.
The complete experimental results are presented in Table~\ref{tab:binary}.

The results demonstrate three key findings.
First, the Vanilla method proves overly conservative in controlling type I error, resulting in substantially inflated type II error rates.
Second, while both NPC and NPC$^*$ successfully adjust for label noise, they remain more conservative than our proposed approach.
Notably, NPC$^*$--which leverages the known values of $m_0$ and $m_1$-- outperforms NPC that relies on their bounds.
In fact, when only bounds are available, the NPC method can be as conservative as, or even more so than, the Vanilla approach.
Third, our method achieves near-optimal performance: though exhibiting slightly elevated type I error (well within acceptable limits), it most closely approximates the true error rates among all compared approaches, even under cases B and C where the DRM assumption is violated.

%% file: sections/multiclass.tex
\section{Neyman-Pearson multiclass classification}
\label{sec:multiclass}
In this section, we study the more general NPMC problem.
Extending NP binary classification to the multiclass setting is nontrivial for two reasons~\citep{tian2024neyman}: (i) the NPMC problem may be infeasible, and (ii) solving it requires a more involved dual formulation.
We adopt their NPMC solver and combine it with our EL-based estimators of $\{w_k^*, \pi_k^*(x)\}_{k=0}^{K-1}$ to handle noisy labels.
When feasible, we show that the resulting classifiers control class-specific type I errors up to a vanishing term.
Empirically, our method matches the clean-label benchmark with sufficiently large samples.
For completeness, we first review their solver under clean data, then introduce our estimator with theoretical guarantees, and finally present empirical results.

\subsection{Revisiting NPMC with clean labels}
\label{sec:NPMC_review}
\citet{tian2024neyman} solves~\eqref{eq:NPMC} via its dual formulation.
The Lagrangian associated with~\eqref{eq:NPMC} is
\begin{align}
L(\phi, \blambda) =&~\sum_{k\in [K]}\rho_k P_k^*(\{\phi(X)\neq k\}) -\sum_{k\in \gS}\lambda_k \{\alpha_k - P_k^*(\{\phi(X)\neq k\})\}\label{eq:lagrangian}
\\
=&~\sum_{k\in [K]} \rho_k + \sum_{k\in \gS}\lambda_k(1-\alpha_k)-\sum_{k\not\in \gS}\rho_k P_k^*(\{\phi(X)=k\}) - \sum_{k\in \gS}(\rho_k +\lambda_k) P_k^*(\{\phi(X)=k\}), \nonumber    
\end{align}
where $\blambda = \{\lambda_k\}_{k\in\gS}\in \sR_{+}^{|\gS|}$ are the Lagrange multipliers.
For any fixed $\blambda$, the optimal classifier that minimizes the Lagrangian has the closed form~\citep[Lemma 2]{tian2024neyman}:
\begin{equation}
\label{eq:optimal_classifier}
\phi^*_{\blambda}(x) = \argmax_{k}\{c_k(\blambda, \vw^*)\pi^*_{k}(x)\},
\end{equation}
with $c_k(\blambda, \vw) = \{\rho_k + \lambda_k \mathbbm{1}(k\in \gS)\}/w_k$.

If strong duality holds, the optimal $\blambda$ can be obtained by solving the dual problem:
\be
\label{eq:population_optimal_lambda}
\blambda^* = \argmax_{\blambda\in \sR_{+}^{\text{Card}(\gS)}} G(\blambda)
\ee
where the Lagrange dual function admits equivalent expressions:
\begin{equation}
\label{eq:cx_pyx}
\begin{split}
G(\blambda)=&\min_{\phi}L(\phi, \blambda) = L(\phi^*_{\blambda}, \blambda)\nonumber\\
=&-\sE_{X}\left\{c_{\phi^*_{\blambda}(X)}(\blambda, \vw^*)\sP(Y=\phi^*_{\blambda}(X)| X)\right\} + \sum_{k\in [K]} \rho_k + \sum_{k\in \gS}\lambda_k(1-\alpha_k).    
\end{split}
\end{equation}
The key advantage of the dual formulation is that the dual objective $G(\blambda)$ is concave, regardless of whether the primal problem is convex. 
This allows us to solve for $\blambda^*$ via standard convex optimization, and then substitute it into~\eqref{eq:optimal_classifier} to obtain the optimal classifier $\phi^*_{\blambda^*}$.

Given a finite training set $\gD^{\text{train}} = \{(X_i, Y_i)\}_{i=1}^{n}$,~\citet{tian2024neyman} proposes the \textbf{NPMC-CX} (ConveX) estimator for the dual objective $G(\blambda)$:
\be
\label{eq:obj-cx}
\widehat{G}(\blambda) = -\frac{1}{n}\sum_{i=1}^{n}c_{\widehat\phi_{\blambda}(x_i)}(\blambda, \widehat{\vw})\widehat{\sP}(Y=\widehat\phi_{\blambda}(X_i)| X=X_i) + \sum_{k\in[K]} \rho_k + \sum_{k\in \gS}\lambda_k(1-\alpha_k),
\ee
where 
\be
\label{eq:classifier_sample}
\widehat\phi_{\blambda}(x) = \argmax_{k}\{c_k(\blambda, \widehat\vw)\widehat\pi_{k}(x)\},
\ee
$\widehat{\vw}$ and $\{\widehat{\pi}_k(x)\}$ are estimated using any suitable classifier trained on $\gD^{\text{train}}$.
Let $\widehat{\blambda} = \argmax_{\lambda} \widehat{G}(\blambda)$; then, plugging $\widehat{\blambda}$ into~\eqref{eq:classifier_sample} yields the final classifier $\widehat\phi_{\widehat\blambda}$.

Under assumptions such as the Rademacher complexity of the function class ${\pi_k(x): \pi_k \in \gF}$ vanishing as $n \to \infty$, they show that $\widehat\phi_{\widehat\blambda}$ satisfies the NP oracle inequalities when the primal problem is feasible. 
When the primal problem is infeasible, they further show that, with high probability, the dual optimal value is smaller than $1$.

\begin{remark}[Dual problem, weak and strong duality]
For readers less familiar with optimization, we briefly explain the motivation for solving the dual problem and its connection to the primal problem. 
The primal problem~\eqref{eq:NPMC} involves constraints on class-specific error for certain classes, which can make direct optimization challenging. 
By introducing Lagrange multipliers $\blambda$, we form the \emph{Lagrangian}~\eqref{eq:lagrangian}, and define the \emph{dual function} $G(\blambda)$ as the minimum of the Lagrangian over all classifiers. 
These two optimization problems are connected via the \textbf{weak duality} property, which states that for any feasible $\blambda\ge 0$, the dual objective $G(\blambda)$ provides a lower bound on the primal objective. 
This ensures that maximizing the dual gives a value that never exceeds the optimal value of the primal.  
If \textbf{strong duality} holds, meaning that the maximum of the dual objective equals the minimum of the primal, solving the dual problem yields the optimal multipliers $\blambda^*$, which can be substituted into~\eqref{eq:optimal_classifier} to recover the primal optimal classifier.  
The key advantage of working with the dual is that $G(\blambda)$ is concave regardless of the convexity of the primal, making the problem easier to solve with standard convex optimization techniques.
We refer interested readers to~\citet[Chapter 5]{boyd2004convex} for more details on duality.
\end{remark}

\subsection{NPMC with noisy labels with statistical guarantees}
\label{sec:NPMC_noisy_label}
Since our EL-based procedure allows us to directly estimate $\widehat{\vw}$ and $\widehat{\pi}_k$ from noisy labeled observations, we can plug in these estimators—just as in the clean-label setting—to construct a classifier with statistical guarantees.

Let $\widehat{\vw}$, $\widehat{\bgamma}$, and $\widehat{\bbeta}$ denote the MELE estimators in~\eqref{eq:mele_estimator} for $\vw^*$, $\bgamma^*$, and $\bbeta^*$, respectively.
Define
\[\widehat{\pi}_k(x) = \frac{\exp(\widehat{\gamma}_k^{\dagger} + \widehat{\beta}_k^{\top}g(x))}{\sum_{k'}\exp(\widehat{\gamma}_{k'}^{\dagger} + \widehat{\beta}_{k'}^{\top}g(x))}.\]
Let $\widehat{\blambda}$ be the minimizer of~\eqref{eq:obj-cx} with $\widehat{\pi}_k$ and $\widehat{\vw}$ plugged in as above.
Our final \textbf{EL-CX} classifier is then defined as
\be
\label{eq:el-cx}
\widehat{\phi}_{\widehat{\blambda}}(x) = \argmax_{k}\{c_k(\widehat\blambda, \widehat\vw)\widehat\pi_{k}(x)\}.
\ee

\begin{remark}[Numerical maximization of the dual problem]
To construct the final EL-CX classifier, in addition to computing $\widehat\vw$ and $\widehat\pi_k(x)$ via the EM algorithm, one also needs to maximize $\widehat{G}(\blambda)$ to obtain $\widehat{\blambda}$.
We perform this maximization using the Hooke–Jeeves algorithm proposed in~\citet{tian2024neyman} (implemented in \texttt{pymoo}~\citealt{blank2020pymoo}).
We adopt the default initialization strategy in \texttt{pymoo} to generate the initial search points.
The optimization is constrained to the hypercube $[0,200]^{|\gS|}$.
\end{remark}

We now formally establish the theoretical properties of our EL-CX classifier below.
The following assumptions are required for our theoretical results.
\begin{assumption}[Compactness]
\label{assump:compactness_multiclass}
The parameter space $\Theta=\{(\bgamma, \bbeta): \bgamma\in \sR^{K-1}, \bbeta\in \sR^{d(K-1)}\}$ is compact, and $\min_{k} w_k>0$.
\end{assumption}

\begin{assumption}[Local strong concavity]
\label{assump:local_strong_concavity}
Let $\blambda^*$ and $G(\blambda)$ as defined in~\eqref{eq:population_optimal_lambda} and~\eqref{eq:cx_pyx} respectively. 
We assume that $G(\blambda)$ is twice continuously differentiable and $\nabla^2 G(\blambda^*) \prec 0$.
\end{assumption}

\begin{theorem}
\label{thm:multiclass_error}
For any classifier $\phi$, denote the $k$th class-specific misclassification probability as $R_k(\phi) = \sP(\phi(X) \neq k | Y=k)$.
Let $\widehat{\phi}_{\widehat{\blambda}}$ be our EL-CX based classifier in~\eqref{eq:el-cx}. When the NPMC problem~\eqref{eq:NPMC} is feasible, under Assumptions~\ref{assumption:bounded_second_moment},~\ref{assump:compactness_multiclass}, and~\ref{assump:local_strong_concavity}, we have that for any $k\in \gS$: $R_k(\widehat\phi_{\widehat{\blambda}})\leq \alpha_k + O_p(\varepsilon(n))$, where $\varepsilon(n)\to 0$ as $n\to\infty$.
\end{theorem}
The proof is in Appendix~\ref{app:multiclass_proof}.
This theorem establishes that when the primal problem is feasible, our classifier asymptotically controls the class-specific errors of the target classes at the target level. 
Both Assumptions~\ref{assump:compactness_multiclass} and~\ref{assump:local_strong_concavity} are the same as that in~\citet{tian2024neyman}.

\begin{remark}[NPMC-ER approach]
\citet{tian2024neyman} also proposes an alternative estimator for $G(\blambda)$ based on a sample-splitting strategy.
While this estimator relaxes the assumption on the Rademacher complexity of the function class $\{\pi_k(x)\}$, it has lower efficiency due to sample splitting.
As shown in the proof of Theorem~\ref{thm:multiclass_error}, our model satisfies the Rademacher complexity condition, so we do not adopt this alternative estimator in our work.
\end{remark}

\subsection{Simulation study}
\label{sec:simulation_multiclass}
In this section, we conduct empirical experiments to compare our proposed method with other approaches.
Following~\citet{tian2024neyman}, we consider three typical settings to generate data:
\begin{enumerate}[label=(\alph*),leftmargin=*]
\item \textbf{Independent covariates}.  
We consider three-class independent Gaussian conditional distributions $X^{k} \sim N(\mu_k^*, I_5)$ where $\mu_0^* = (-1, 2, 1, 1, 1)^\top$, $\mu_1^* = (0, 1, 0, 1, 0)^\top$, $\mu_2 = (1, 1, -1, 0, 1)^\top$, and $I_5$ is the $5$-dimensional identity matrix. 
The marginal distribution of $ Y $ is $\sP(Y=0) = \sP(Y=1) = 0.3$ and $\sP(Y=2) = 0.4$.

\item \textbf{Dependent covariates}.  
In contrast to (a), where all five variables are independent Gaussian, we consider a four-class correlated Gaussian conditional distribution: $X^{k} \sim N(\eta_k, \Sigma)$, where $\eta_0^* = (1, -2, 0, -1, 1)^\top$, $\eta_1^* = (-1, 1, -2, -1, 1)^\top$, $\eta_2^* = (2, 0, -1, 1, -1)^\top$, $\eta_3^* = (1, 0, 1, 2, -2)^\top$, and $\Sigma_{ij}^* = 0.1^{|i-j|}$. The marginal distribution of $ Y $ is $\sP(Y = k) = 0.1(k+1)$ for $ k = 0, \ldots, 3$.

\item \textbf{Unequal covariance}. In contrast to previous cases where the covariance matrices are the same, we consider the following distribution $X^{k} \sim N(\mu_k^*, (k+1)I_5)$, where $\mu_k^*$s and the marginal distribution of $Y$ are the same as in (a). 
\end{enumerate}
Under each case, we generate a training set of size $n$ and a test set of fixed size $ m = 20{,}000 $. 
Since we study the NPMC problem under noisy labeled data, we \emph{generate noisy labels for the training set} using the following transition matrix: 
$T_{lk}^* = \eta/K$ if $k \neq l$ and $1 - (K-1)\eta/K$ otherwise.
We vary $n$ from $5000$ to $25,000$ with a step size of $5,000$, and consider $\eta\in \{0.05,0.1,0.15,0.2\}$. 
Each setting is repeated $R=500$ times.

For each data setting, we consider the following NPMC problems:
\begin{enumerate}[label=(\alph*),leftmargin=*]
\item \textbf{Independent covariates}. We aim to minimize $P_2^*(\{\phi(X) \neq 2\})$ subject to $P_0^*(\{\phi(X) \neq 0\}) \leq 0.1$ and $P_1^*(\{\phi(X) \neq 1\}) \leq 0.15$.

\item \textbf{Dependent covariates}. The goal is to minimize the overall misclassification error subject to the constraints:
$P_0^*(\{\phi(X) \neq 0\}) \leq 0.1$ and $P_3^*(\{\phi(X) \neq 3\}) \leq 0.05$.

\item \textbf{Unequal covariance}. We aim to minimize $\sum_{k=0}^{2}P_k^*(\{\phi(X) \neq k\})$ subject to $P_0^*(\{\phi(X) \neq 0\}) \leq 0.1$ and $P_1^*(\{\phi(X) \neq 1\}) \leq 0.12$.
\end{enumerate}
We compare EL-CX with the following baselines: (1) \textbf{Oracle}, which applies NPMC-CX~\citep{tian2024neyman} to \emph{clean} training labels; and (2) \textbf{Vanilla}, which applies NPMC-CX directly to noisy labels \emph{without accounting for label noise}.

As shown in Appendix~\ref{app:DRM}, the DRM assumption~\eqref{eq:drm} holds in all cases with $g(x)=x$ for the first two cases and $g(x)=(x^\top,x^\top x)^\top$ for the third case. 
We use the same covariates in the multinomial logistic regression model for all methods. 
The MSE of our MELE estimator converges at the root-$n$ rate, visualized in Figure~\ref{fig:parameter_estimation} in Appendix~\ref{app:model_fitting}.

We evaluate the classifiers using two metrics on the test set: class-specific errors and the NPMC objective in~\eqref{eq:NPMC}.
To assess error control, we define the excessive risk for each class as the difference between its empirical misclassification rate and the target level.
Table~\ref{tab:npmc} reports the mean and standard deviation of this quantity at $\eta=0.1$ over $R$ repetitions.
Our method achieves objective values comparable to the oracle, with class-specific excessive risks nearly identical and improving as the sample size increases.
In contrast, the Vanilla method, which ignores label noise, is markedly conservative, leading to larger excessive risks (especially for class 0 in cases (b) and (c)) and consequently suboptimal objective values that worsen with increasing sample size.

\begin{table}[!htbp]
\centering
\caption{Mean excessive risk (with standard deviation in parentheses) across target classes and objective function values, grouped by sample size and method, for the case $\eta = 0.1$.}
\label{tab:npmc}
\resizebox{\textwidth}{!}{
\begin{tabular}{@{}cc@{\hspace{2em}}cc>{\columncolor{RoyalBlue!5}}c@{\hspace{2em}}cc>{\columncolor{RoyalBlue!5}}c@{\hspace{2em}}cc>{\columncolor{RoyalBlue!5}}c@{}}
\toprule
\multirow{2}{*}{$n$} & \multirow{2}{*}{Method} & \multicolumn{3}{c}{(a)}&\multicolumn{3}{c}{(b)}&\multicolumn{3}{c}{(c)}\\ 
\cmidrule(lr{2em}){3-5}
\cmidrule(lr{2em}){6-8}
\cmidrule(lr{-0.5em}){9-11}
&&Class 0 & Class 1 & Obj.&Class 0 & Class 3 & Obj. &Class 0 & Class 1 & Obj. \\
\midrule
\multirow{4}{*}{5K} 
& Vanilla & \negred{-0.023}\,\bracket{0.005} & \negred{-0.050}\,\bracket{0.006} & 0.982\,\bracket{0.004} & \negred{-0.097}\,\bracket{0.001} & \negred{-0.019}\,\bracket{0.005} & 0.498\,\bracket{0.027} & \negred{-0.028}\,\bracket{0.007} & 0.040\,\bracket{0.006} & 1.226\,\bracket{0.008} \\
& \textbf{Ours}   & 0.001\,\bracket{0.010} & 0.002\,\bracket{0.012} & 0.036\,\bracket{0.007} & \negred{-0.032}\,\bracket{0.008} & 0.001\,\bracket{0.007} & 0.086\,\bracket{0.003} & 0.005\,\bracket{0.013} & \negred{-0.012}\,\bracket{0.007} & 0.830\,\bracket{0.147} \\
& Oracle  & 0.001\,\bracket{0.008} & 0.002\,\bracket{0.010} & 0.035\,\bracket{0.005} & \negred{-0.033}\,\bracket{0.008} & 0.001\,\bracket{0.005} & 0.085\,\bracket{0.002} & 0.004\,\bracket{0.009} & \negred{-0.013}\,\bracket{0.006} & 0.800\,\bracket{0.078} \\
\midrule
\multirow{4}{*}{10K} 
& Vanilla & \negred{-0.023}\,\bracket{0.005} & \negred{-0.050}\,\bracket{0.005} & 0.982\,\bracket{0.003} & \negred{-0.099}\,\bracket{0.001} & \negred{-0.018}\,\bracket{0.003} & 0.512\,\bracket{0.004} & \negred{-0.029}\,\bracket{0.006} & 0.040\,\bracket{0.006} & 1.225\,\bracket{0.007} \\
& \textbf{Ours}   & 0.000\,\bracket{0.008} & 0.001\,\bracket{0.010} & 0.036\,\bracket{0.005} & \negred{-0.033}\,\bracket{0.007} & 0.000\,\bracket{0.005} & 0.086\,\bracket{0.002} & 0.003\,\bracket{0.009} & \negred{-0.013}\,\bracket{0.006} & 0.814\,\bracket{0.114} \\& Oracle  & 0.000\,\bracket{0.006} & 0.001\,\bracket{0.008} & 0.035\,\bracket{0.004} & \negred{-0.034}\,\bracket{0.007} & 0.000\,\bracket{0.004} & 0.086\,\bracket{0.002} & 0.003\,\bracket{0.007} & \negred{-0.013}\,\bracket{0.005} & 0.791\,\bracket{0.049} \\\midrule
\multirow{4}{*}{20K} 
& Vanilla & \negred{-0.023}\,\bracket{0.004} & \negred{-0.050}\,\bracket{0.005} & 0.984\,\bracket{0.002} & \negred{-0.099}\,\bracket{0.001} & \negred{-0.018}\,\bracket{0.003} & 0.512\,\bracket{0.004} & \negred{-0.029}\,\bracket{0.005} & 0.040\,\bracket{0.006} & 1.227\,\bracket{0.007} \\
& \textbf{Ours}   & 0.001\,\bracket{0.006} & 0.001\,\bracket{0.008} & 0.035\,\bracket{0.004} & \negred{-0.034}\,\bracket{0.006} & 0.000\,\bracket{0.004} & 0.086\,\bracket{0.002} & 0.004\,\bracket{0.007} & \negred{-0.013}\,\bracket{0.005} & 0.797\,\bracket{0.064} \\
& Oracle  & 0.000\,\bracket{0.005} & 0.001\,\bracket{0.006} & 0.035\,\bracket{0.003} & \negred{-0.034}\,\bracket{0.006} & 0.000\,\bracket{0.003} & 0.086\,\bracket{0.002} & 0.004\,\bracket{0.006} & \negred{-0.013}\,\bracket{0.005} & 0.789\,\bracket{0.031} \\
\midrule
\multirow{4}{*}{25K} 
& Vanilla & \negred{-0.023}\,\bracket{0.004} & \negred{-0.050}\,\bracket{0.004} & 0.984\,\bracket{0.002} & \negred{-0.099}\,\bracket{0.001} & \negred{-0.018}\,\bracket{0.003} & 0.512\,\bracket{0.004} & \negred{-0.029}\,\bracket{0.004} & 0.040\,\bracket{0.005} & 1.226\,\bracket{0.006} \\
& \textbf{Ours}   & \negred{-0.000}\,\bracket{0.006} & 0.000\,\bracket{0.007} & 0.035\,\bracket{0.004} & \negred{-0.034}\,\bracket{0.006} & 0.000\,\bracket{0.004} & 0.086\,\bracket{0.002} & 0.003\,\bracket{0.007} & \negred{-0.013}\,\bracket{0.005} & 0.800\,\bracket{0.062} \\
& Oracle  & 0.000\,\bracket{0.005} & \negred{-0.000}\,\bracket{0.006} & 0.035\,\bracket{0.003} & \negred{-0.034}\,\bracket{0.006} & 0.000\,\bracket{0.003} & 0.086\,\bracket{0.002} & 0.003\,\bracket{0.005} & \negred{-0.014}\,\bracket{0.004} & 0.790\,\bracket{0.031} \\
\bottomrule
\end{tabular}}
\end{table}

We also provide violin plots of the class-specific error and the objective value under different noise levels $\eta$ for case (a) in Figure~\ref{fig:violin_plots}. 
The corresponding plots for the other two cases are included in Appendix~\ref{app:violin_plots}. 
\begin{figure}[!ht]
\centering
\includegraphics[width=0.8\linewidth]{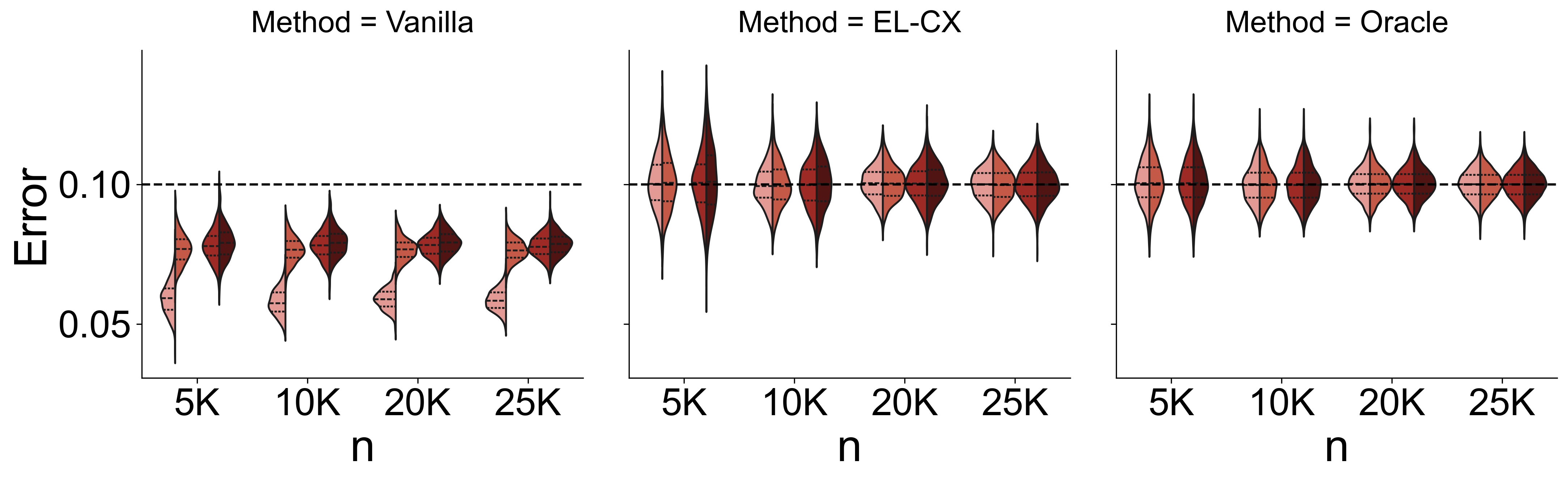}
\includegraphics[width=0.8\linewidth]{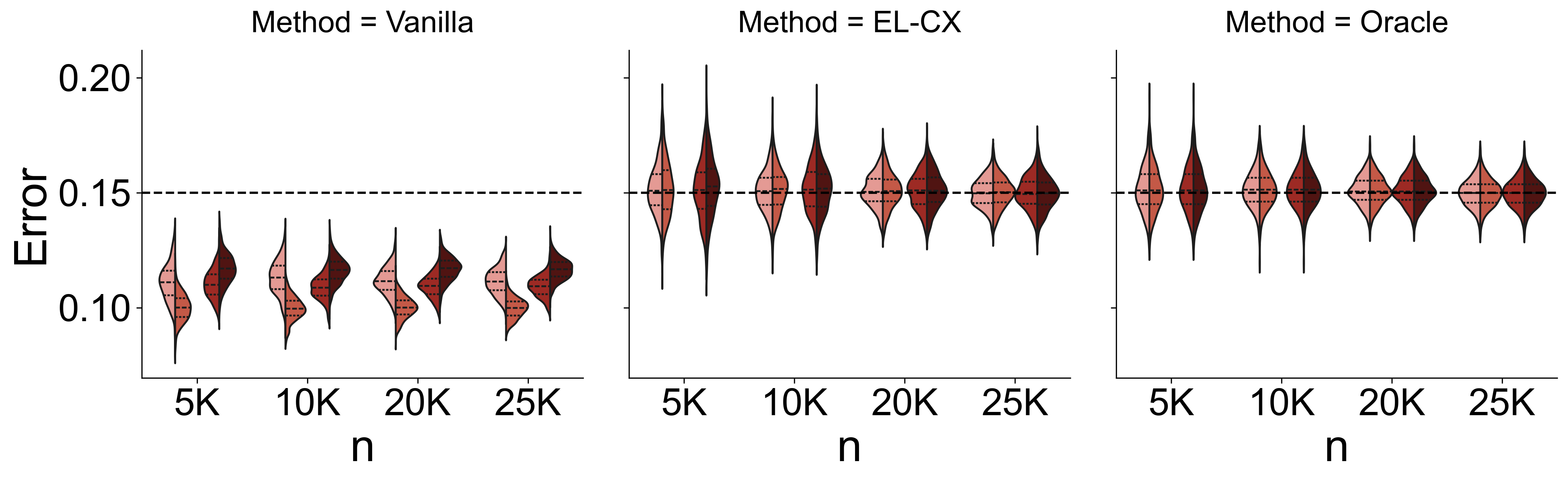}
\includegraphics[width=0.8\linewidth]{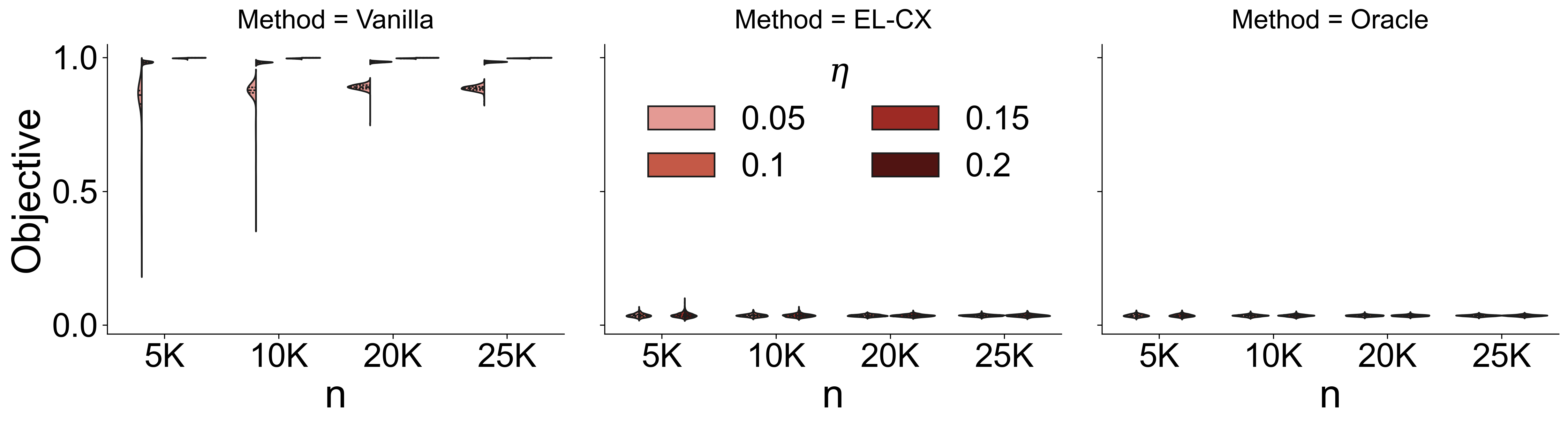}
\caption{Violin plots (top to bottom) show the misclassification errors for Class 0, Class 1, and the objective function value under Case (a), computed over $R = 500$ repetitions. Results compare different methods across varying sample sizes ($n$) and noise levels ($\eta$). The dashed black lines in the first two rows mark the target misclassification errors.}
\label{fig:violin_plots}
\end{figure}
These visualizations further demonstrate that our EL-based method performs nearly as well as the oracle estimator when the sample size $n$ is large, regardless of $\eta$, which confirms the effectiveness of the proposed method.

%% file: sections/real_data.tex
\section{Real data experiments}
\label{sec:real_data}
Following~\citet{tian2024neyman}, we conduct experiments on two real datasets: the Dry Bean dataset and the Landsat Satellite dataset. 
Both datasets are used for multi-class classification. 
Due to space constraints, we present the results for the Landsat Satellite dataset here, and defer the results for the Dry Bean dataset to Appendix~\ref{app:dry_bean}.

The dataset used in this experiment is the Statlog (Landsat Satellite) dataset obtained from the UCI Machine Learning Repository at this~\href{https://archive.ics.uci.edu/dataset/146/statlog+landsat+satellite}{link}. 
It contains multi-spectral satellite imagery captured by NASA’s Landsat satellite, with $6,435$ instances divided into a training set ($4,435$ instances) and a test set ($2,000$ instances). 
Each instance consists of $36$ numerical attributes representing pixel values from four spectral bands over a $3\times 3$ pixel neighborhood. 
The goal is to predict the class of the central pixel, which belongs to one of six land cover types: red soil ($1,533$), cotton crop ($703$), grey soil ($1,358$), damp grey soil ($626$), soil with vegetation stubble ($707$), and very damp grey soil ($1,508$). 
This dataset is commonly used for benchmarking classification methods and is known to be challenging due to the similarity between some land cover classes.

We recode the six classes into labels 0 through 5, and consider minimizing the sum of class-specific errors subject to $P_3^{*}(\phi(X)\neq 3) \leq 0.2$ and $P_4^{*}(\phi(X)\neq 4) \leq 0.1$. 
We use the train-test split provided in the UCI repository and generate noisy labels for the training set using the contamination model described in Section~\ref{sec:simulation_multiclass}. 
This procedure is repeated 100 times. 
We fit a multinomial logistic regression model with $g(x)=x$ in all cases, using $L_2$ regularization with penalty strength $10$ wherever multinomial or weighted multinomial logistic regression applies.

\begin{table}[!ht]
\centering
\caption{Mean excessive risk across target classes and objective function values, grouped by noise level $\eta$ and method.}
\label{tab:satlog}
\resizebox{0.95\textwidth}{!}{
\begin{tabular}{@{}c@{\hspace{2em}}cc>{\columncolor{RoyalBlue!5}}c@{\hspace{2em}}cc>{\columncolor{RoyalBlue!5}}c@{\hspace{2em}}cc>{\columncolor{RoyalBlue!5}}c@{\hspace{2em}}cc>{\columncolor{RoyalBlue!5}}c@{}}
\toprule
\multirow{2}{*}{Method} & \multicolumn{3}{c}{$\eta=0.05$}&\multicolumn{3}{c}{$\eta=0.10$}&\multicolumn{3}{c}{$\eta=0.15$}&\multicolumn{3}{c}{$\eta=0.20$}\\ 
\cmidrule(lr{2em}){2-4}
\cmidrule(lr{2em}){5-7}
\cmidrule(lr{2em}){8-10}
\cmidrule(lr{-0.5em}){11-13}
&Class 3 & Class 4 & Obj. &Class 3 & Class 4 & Obj. &Class 3 & Class 4 & Obj. &Class 3 & Class 4 & Obj.\\
\midrule
Na\"ive& 0.620&0.282&0.243&0.669&0.317&0.254&0.711&0.351&0.265&0.730&0.388&0.274\\
Vanilla&\negred{-0.133}&\negred{-0.031}&0.684&\negred{-0.167}&0.013&0.691&\negred{-0.162}&0.030&0.695&\negred{-0.147}&0.034&0.698\\
\textbf{Ours}&\negred{-0.089}&\negred{-0.019}&0.249&\negred{-0.097}&\negred{-0.007}&0.268&\negred{-0.107}&0.001&0.284&\negred{-0.114}&0.000&0.310\\
Oracle&\negred{-0.096}&\negred{-0.016}&0.233&\negred{-0.096}&\negred{-0.016}&0.233&\negred{-0.096}&\negred{-0.016}&0.233&\negred{-0.096}&\negred{-0.016}&0.233\\
\bottomrule
\end{tabular}}
\end{table}
The performance of the methods, based on 100 repetitions, is summarized in Table~\ref{tab:satlog}. 
We use the same performance metrics as in Section~\ref{sec:simulation_multiclass}: the mean excessive risk and objective value. 
We also include a new baseline \textbf{Na\"ive}, which uses standard multinomial logistic regression without any error control. 
The results show the necessity of error control, as the naive approach fails to meet the target errors.
Among the remaining methods, even though the multinomial logistic regression model with $g(x)=x$ can be misspecified on this real dataset, our method stays closer to the target error levels compared to the vanilla approach. 
While our method becomes slightly worse when $\eta$ increases, it still outperforms the vanilla method. When the noise is small (\ie, $\eta=0.05$), our method is comparable to the oracle approach.

%% file: sections/conclusion.tex
\section{Discussion and concluding remarks}
\label{sec:conclusion}
This paper addresses a critical challenge in safety-critical machine learning applications: maintaining reliable error control guarantees when the training data contain label noise. In domains like medical diagnosis or autonomous systems, label noise can systematically distort NPMC, leading to violated safety constraints (e.g., elevated false negatives in disease screening) with potentially severe consequences. 

By extending the Neyman–Pearson classification framework to accommodate noisy labels, we have developed an empirical likelihood–based approach that preserves class-specific error guarantees without requiring prior knowledge of the noise transition matrix.
Theoretically, our method satisfies NP oracle inequalities under mild conditions. Empirically, it achieves near-optimal performance on both synthetic and real-world datasets, closely approximating the behavior of classifiers trained on clean labels.

Incorporating data-adaptive selection of the basis function $g$ as in~\citet{zhang2022density}, or even integrating it with deep neural networks, remains an important direction for future work.
In addition, extending our framework to handle more complex instance-dependent label noise~\citep{xia2020part} is another promising yet challenging avenue.

%% file: sections/appendix.tex
\input{sections/appendices/EM}
\input{sections/appendices/more_exp}

\section{Theoretical proofs}
\input{sections/appendices/el}

\subsection{Identifiability}
\label{app:identifiability}
\input{sections/appendices/identifiability}

\subsection{Preliminaries on Rademacher complexity}
In this section, we introduce key definitions on Rademacher complexity and essential lemmas for our class-specific error control analysis.
\input{sections/appendices/def_lemma}

\subsection{Guaranteed type I error control under binary classification}
\label{app:binary_proof}
This section is organized as follows: we first prove Theorem~\ref{thm:binary_type1}, then present the supporting lemmas at the end.
\input{sections/appendices/error_control_binary}

\subsection{Guaranteed class specific error control under multiclass classification}
\label{app:multiclass_proof}
\input{sections/appendices/error_control_multiclass}

%% file: sections/appendices/EM.tex
\section{EM algorithm details and convergence analysis}
We present the complete derivation of the EM algorithm for the EL based inference under the general multiclass case with noisy labels. 
\subsection{EM algorithm implementation}
\label{app:em_detail}
Consider the complete data $\{(X_i,Y_i,\widetilde{Y}_i)\}_{i=1}^n$ with the same notation as Section~\ref{sec:EM_algorithm}. 
The complete data log empirical likelihood (EL) is:
\[
\ell_{n}^{c}(\vp, \btheta) = \sum_{i=1}^{n} \sum_{l\in [K]} \sum_{k\in [K]} \mathbbm{1}(\widetilde{Y}_i=l, Y_i=k) \log \sP(X = X_i, \widetilde{Y}_i=l, Y_i=k).
\]
Using the chain rule and the instance-independent assumption in~\eqref{assumption:instance_independent_noise}, we decompose the joint probability into:
\[
\begin{split}
\sP(X = X_i, \widetilde{Y}_i=l, Y_i=k) &= \sP(X = X_i | \widetilde{Y}_i = l, Y_i = k) \sP(\widetilde{Y}_i = l | Y_i = k) \sP(Y_i = k) \\
&= \sP(X = X_i | Y_i = k) T_{lk}^* w_k^* \\
&= P_0^*(\{X_i\}) \exp\{\gamma_k + \beta_k^{\top} g(X_i)\} T_{lk}^* w_k^*,
\end{split}
\]
where the last equality follows from the DRM in~\eqref{eq:drm}. 
This yields the complete data profile log-EL:
\begin{equation*}
\ell_{n}^{c}(\vp,\btheta) =\sum_{i=1}^{n} \sum_{l\in [K]} \sum_{k\in [K]}\mathbbm{1}(\widetilde{Y}_i=l)\mathbbm{1}(Y_i=k)\left[\log p_{i} + \{\gamma_k + \beta_k^{\top}g(X_i)\} + \log T_{lk} + \log w_k\right].
\end{equation*}
We define the profile complete data log-EL as 
\[p\ell_{n}^{c}(\btheta) = \sup_{\vp} \ell_{n}^{c}(\vp, \btheta)\]
where the maximization is subject to the constraints in~\eqref{eq:constraints}. 
Using the same derivation from Appendix~\ref{app:profile_loglik}, the profile complete data log-EL becomes
\begin{equation*}
p\ell_{n}^{c}(\btheta) =\sum_{i=1}^{n} \sum_{l\in [K]} \sum_{k\in [K]}\mathbbm{1}(\widetilde{Y}_i=l)\mathbbm{1}(Y_i=k)\left(\gamma_k + \beta_k^{\top}g(X_i) + \log T_{lk} + \log w_k\right) + \sum_{i=1}^{n}\log p_i(\btheta)
\end{equation*}
where $p_i(\btheta)$ is given in~\eqref{eq:vp_lagrange}, and the $\nu_k$'s are the solutions to~\eqref{eq:lagrange-multiplier-solution}.

The EM algorithm is an iterative procedure that starts with an initial guess $\btheta^{(0)}$. 
At each iteration $t$, given the current parameter estimate $\btheta^{(t)}$, the algorithm proceeds in two steps:
\begin{itemize}[leftmargin=*]
\item \textbf{E-step (Expectation step)}: The missing data is ``filled in'' by computing an educated guess based on the current parameter estimate $\btheta^{(t)}$. 
This involves calculating the posterior distribution of the missing data given the observed data and the current parameters.
\begin{equation}
\label{eq:e-step}
\begin{split}
\omega_{ik}^{(t)} 
=&~\sE(\mathbbm{1}(Y_i=k)| \btheta^{(t)};\widetilde{Y}_i, X_i) =\sP(Y_{i}=k| \btheta^{(t)};\widetilde{Y}_i, X_i)\\
=&~\frac{\sP(Y_i=k,\widetilde{Y}=\widetilde{Y}_i, X=X_i;\btheta^{(t)})}{\sum_{k'}\sP(Y_i=k',\widetilde{Y}=\widetilde{Y}_i, X=X_i;\btheta^{(t)})}\\
=&~\frac{\sP(X=X_i| Y_i=k,\widetilde{Y}=\widetilde{Y}_i;\btheta^{(t)})\sP(\widetilde{Y}=\widetilde{Y}_i| Y_i=k;\btheta^{(t)})\sP(Y_i=k;\btheta^{(t)})}{\sum_{k'}\sP(X=X_i| Y_i=k',\widetilde{Y}=\widetilde{Y}_i;\btheta^{(t)})\sP(\widetilde{Y}=\widetilde{Y}_i| Y_i=k';\btheta^{(t)})\sP(Y_i=k';\btheta^{(t)})}\\
=&~\frac{\sP(X=X_i| Y_i=k;\btheta^{(t)})T^{(t)}_{\widetilde{Y}_i,k}w_k^{(t)}}{\sum_{k'}\sP(X=X_i| Y_i=k';\btheta^{(t)})T^{(t)}_{\widetilde{Y}_i,k'}w_{k'}^{(t)}}=\frac{\exp(\gamma_k^{(t)}+\langle \beta_k^{(t)}, g(X_i)\rangle)T_{\widetilde{Y}_i,k}^{(t)}w_k^{(t)}}{\sum_{k'=1}^{K}\exp(\gamma_{k'}^{(t)}+\langle \beta_{k'}^{(t)}, g(X_i)\rangle)T_{\widetilde{Y}_i,k'}^{(t)}w_{k'}^{(t)}}.
\end{split}
\end{equation}
Then the expected complete data profile empirical log-EL at iteration $t$ is
\begin{equation*}
\begin{split}
Q(\btheta;\btheta^{(t)})=&~\sE\{p\ell_n^{(c)}(\btheta)| \btheta^{(t)};\gD\}\\
=&~\sum_{i=1}^{n}\sum_{l\in[K]}\mathbbm{1}(\widetilde{Y}_i=l)\sum_{k\in [K]}\omega_{ik}^{(t)}(\log w_k + \gamma_k + \beta_k^{\top}g(X_i) + \log T_{lk})+\sum_{i=1}^{n} \log p_i(\btheta)    
\end{split}
\end{equation*}

\item \textbf{M-step (Maximization step)}: Instead of maximizing log-EL, the M-step maximizes $Q$ with respect to $\btheta$ instead, \ie
\[\btheta^{(t+1)} = \argmax Q(\btheta, \btheta^{(t)}).\]
Note that $Q$ is separable in $\vw$, $\mT$, and $(\bgamma, \bbeta)$. 
Since it can be written as an additive sum:
\[Q(\btheta;\btheta^{(t)}) = Q_1(\vw;\btheta^{(t)})+Q_2(\mT;\btheta^{(t)})+Q_3(\bgamma,\bbeta;\btheta^{(t)})\]
where 
\[
\begin{split}
Q_1(\vw;\btheta^{(t)}) =&~\sum_{i=1}^{n}\sum_{k\in [K]}\omega_{ik}^{(t)}\log w_k\\
Q_2(\mT;\btheta^{(t)}) =&~\sum_{i=1}^{n}\sum_{l\in [K]}\mathbbm{1}(\widetilde{Y}_i=l)\sum_{k\in [K]}\omega_{ik}^{(t)} \log T_{lk}\\
Q_3(\bgamma,\bbeta;\btheta^{(t)}) =&~\sum_{i=1}^{n}\sum_{l\in [K]}\mathbbm{1}(\widetilde{Y}_i=l)\sum_{k\in [K]}\omega_{ik}^{(t)}(\gamma_k + \beta_k^{\top}g(X_i))\\
&-\sum_{i=1}^{n} \log \left\{1+\sum_{k=1}^{K-1}\nu_k\left(\exp(\gamma_k+\beta_k^{\top}g(X_i))-1\right)\right\}
\end{split}
\]
We can therefore optimize each set of parameters separately as follows:
\begin{itemize}[leftmargin=*]
\item \textbf{Maximize over $\vw$.} 
From the decomposition of $Q$, we have
\[\vw^{(t+1)} = \argmax_{\vw} Q_1(\vw;\btheta^{(t)})\]
subject to the constraint that $w_k \in [0, 1]$ and $\sum_{k\in [K]} w_k = 1$. 
Using the method of Lagrange multipliers, the Lagrangian is
\[\gL_{\vw} = \sum_{i=1}^{n}\sum_{k\in [K]}\omega_{ik}^{(t)}\log w_k - \zeta \left(\sum_{k\in [K]}w_k -1\right).\]
At the optimum, we require
\[\frac{\partial \gL_{\vw}}{\partial \zeta}=0, \frac{\partial \gL_{\vw}}{\partial w_k}=0,~k\in [K].\]
Solving these $K+1$ equations gives
\[w_{k}^{(t+1)} =\frac{1}{n}\sum_{i=1}^{n}\omega_{ik}^{(t)}.\]

\item \textbf{Maximize over $\mT$.}  
From the decomposition of $Q$, we have
\[\mT^{(t+1)} = \argmax_{\mT} Q_2(\mT;\btheta^{(t)})\]
subject to the constraint that $\sum_{l\in [K]} T_{lk} = 1$ for all $k \in [K]$. 
Using the method of Lagrange multipliers, the Lagrangian is
\[\gL_{\mT} = \sum_{i=1}^{n}\sum_{l\in [K]}\mathbbm{1}(\widetilde{Y}_i=l)\sum_{k\in [K]}\omega_{ik}^{(t)}\log T_{lk} - \sum_{j\in [K]}\zeta_j \left(\sum_{l\in [K]} T_{lk} -1\right).\]
At the optimum, we require
\[
\begin{split}
0 =&~\frac{\partial \gL_{\mT}}{\partial \zeta_k} = \sum_{l\in [K]} T_{lk} -1,~k\in [K] \\
0 =&~\frac{\partial \gL_{\mT}}{\partial T_{lk}} = \frac{\sum_{i=1}^{n}\mathbbm{1}(\widetilde{Y}_i=l)\omega_{ik}^{(t)}}{T_{lk}} - \zeta_k,~l\in[K],~k\in [K].
\end{split}
\]
Solving these $K^2 + K$ equations yields
\[T_{lk}^{(t+1)} =\frac{\sum_{i=1}^{n}\mathbbm{1}(\widetilde{Y}_i=l)\omega_{ik}^{(t)}}{\sum_{i=1}^{n}\omega_{ik}^{(t)}}.\]

\item \textbf{Maximize over $(\bgamma,\bbeta)$.} 
Note that $(\bgamma^{(t+1)}, \bbeta^{(t+1)})$ maximizes the function
{\small{
\[
Q_3(\bgamma,\bbeta;\btheta^{(t)}) =\sum_{i=1}^{n}\sum_{k\in [K]}\omega_{ik}^{(t)}(\gamma_k + \beta_k^{\top}g(X_i)) -\sum_{i=1}^{n} \log\left\{1+\sum_{k=1}^{K-1}\nu_k \left(\exp(\gamma_k +\beta_k^{\top}g(X_i))-1\right)\right\}
\]}}
up to a constant that does not depend on $(\bgamma, \bbeta)$. 
The maximizers $\bgamma^{(t+1)}$ and $\bbeta^{(t+1)}$ are the solutions to the following system:
\[
\frac{\partial Q_3}{\partial \gamma_k} = \sum_{i=1}^{n}\omega_{ik}^{(t)} -\sum_{i=1}^{n}\frac{\nu_k \exp(\gamma_k +\beta_k^{\top}g(X_i)) + (\partial\nu_k/\partial \gamma_k )\left\{\exp(\gamma_k +\beta_k^{\top}g(X_i))-1\right\}}{1+\sum_{k'=1}^{K-1}\nu_{k'} \left\{\exp(\gamma_{k'} +\beta_{k'}^{\top}g(X_i))-1\right\}} =0. 
\]
Since the Lagrange multipliers $\nu_k$ satisfy the equation in~\eqref{eq:lagrange-multiplier-solution}, this simplifies to
\[\sum_{i=1}^{n}\omega_{ik}^{(t)} = \nu_k \sum_{i=1}^{n}\frac{\exp(\gamma_k +\beta_k^{\top}g(X_i))}{\sum_{k\in[K]}\nu_k \exp(\gamma_k +\beta_k^{\top}g(X_i))} = n \nu_k\]
and at the optimum, we have $\nu_k = n^{-1} \sum_{i=1}^{n} \omega_{ik}^{(t)} := \omega_{\cdot k}^{(t)}$. 
Thus, the stationary point of $Q_3(\bgamma, \bbeta;\btheta^{(t)})$ coincides with the stationary point of
\[
\widetilde{Q}(\bgamma,\bbeta;\btheta^{(t)}) =\sum_{i=1}^{n}\sum_{k\in [K]}\omega_{ik}^{(t)}(\gamma_k + \beta_k^{\top}g(X_i)) -\sum_{i=1}^{n} \log\left\{\sum_{k\in[K]}w_{\cdot k}\exp(\gamma_k +\beta_k^{\top}g(X_i))\right\}.
\]
Let $\bar\gamma_k = \gamma_k + \log (\omega_{\cdot k}^{(t)}/\omega_{\cdot 0}^{(t)})$, then the function $\widetilde{Q}$ becomes
\[
\widetilde{Q}(\bar\bgamma,\bbeta;\btheta^{(t)}) =\sum_{i=1}^{n}\sum_{k\in [K]}\omega_{ik}^{(t)}(\bar\gamma_k + \beta_k^{\top}g(X_i)) -\sum_{i=1}^{n} \log\left\{\sum_{k\in[K]}\exp(\bar\gamma_k +\beta_k^{\top}g(X_i))\right\}.
\]
Thus, we have 
\[\bar\bgamma^{(t+1)}, \bbeta^{(t+1)} = \argmax_{\bar\bgamma,\bbeta}\widetilde{Q}(\bar\bgamma,\bbeta;\btheta^{(t)}).\]
Now, we show that $(\bar{\bgamma}^{(t+1)}, \bbeta^{(t+1)})$ is the maximum weighted log-likelihood estimator under the multinomial logistic regression model based on the dataset in Table~\ref{tab:weighted_glm}.
\begin{table}[hbt]
\centering
\caption{The dataset for weighted multinomial logistic regression in the M step.}
\label{tab:weighted_glm}
\begin{tabular}{ccc}
\toprule
Response & Covariates & weight\\
\midrule
$0$ & $(1, g^{\top}(X_1))^{\top}$ & $\omega_{10}^{(t)}$\\
\vdots & \vdots & \vdots \\
$0$ & $(1,g^{\top}(X_{n}))^{\top}$ & $\omega_{n0}^{(t)}$\\
$1$ & $(1,g^{\top}(X_1))^{\top}$ & $\omega_{11}^{(t)}$\\
\vdots & \vdots & \vdots \\
$1$ & $(1,g^{\top}(X_{n}))^{\top}$ & $\omega_{n1}^{(t)}$\\
\vdots & \vdots & \vdots \\
$K-1$ & $(1,g^{\top}(X_1))^{\top}$ & $\omega_{1,K-1}^{(t)}$\\
\vdots & \vdots & \vdots \\
$K-1$ & $(1,g^{\top}(X_{n}))^{\top}$ & $\omega_{n,K-1}^{(t)}$\\
\bottomrule
\end{tabular}
\end{table}
To see this, consider the weighted multinomial logistic regression. 
Given the multinomial logistic regression model
\[\sP(Y=k| X=x) = \frac{\exp(x^{\top}\beta_k)}{\sum_{j=0}^{K-1}\exp(x^{\top}\beta_j)},~k\in [K]\]
and the dataset $\{(x_i, y_i)\}_{i=1}^{n}$, the weighted log-likelihood function with weights $\{c_i\}_{i=1}^{n}$ becomes
\be
\label{eq:loss_weighted_multinom}
\begin{split}
\ell(\bbeta) =&~\sum_{i=1}^{n} c_i\log \left(\frac{\exp(x_i^{\top}\beta_{y_i})}{\sum_{j=1}^{K}\exp(x_i^{\top}\beta_j)} \right)\\
=&\sum_{i=1}^{n}\sum_{k\in [K]}c_ix_i^{\top}\beta_k\mathbbm{1}(y_i=k) - \sum_{i=1}^{n}c_i\log \left(\sum_{j=1}^{K}\exp(x_i^{\top}\beta_j)\right)
\end{split}
\ee
By substituting the values from Table~\ref{tab:weighted_glm}, the weighted log-likelihood function becomes $\widetilde{Q}$.
Thus, $(\bar{\bgamma}^{(t+1)}, \bbeta^{(t+1)})$ maximizes the weighted log-likelihood.
We can use any existing packages in \texttt{R} or \texttt{Python} to find the numerical value.
Finally, we compute 
\be
\gamma_k^{(t+1)} = \bar\gamma_k^{(t+1)} - \log (\omega_{\cdot k}^{(t)}/\omega_{\cdot 0}^{(t)}),~k\in[K].
\ee
\end{itemize}
The weights are 
\[
p_i^{(t+1)} = p_i(\btheta^{(t+1)}) = n^{-1}\left\{\sum_{k\in [K]} \omega_{\cdot k}^{(t)} \exp(\gamma_k^{(t+1)}+\langle \beta_k^{(t+1)}, g(X_i)\rangle)\right\}^{-1}
\]
since $\nu_k = \omega_{\cdot k}^{(t)}$ at the optimal point.
\end{itemize}
The E-step and M-step are repeated iteratively until the change in the profile log-EL falls below a predefined threshold. 
This stopping criterion is guaranteed to be met because the EM algorithm produces a sequence of estimates that monotonically increase the true objective--the profile log-EL. 
We provide a proof of this convergence property in the next subsection.


\subsection{Proof of the convergence of EM algorithm}
\label{app:em_convergence}
We establish the monotonic increase of the profile log-EL at each EM iteration.
\begin{proof}
Consider the difference in profile log-EL between consecutive iterations:
\[
\begin{split}
&p\ell_n(\btheta^{(t+1)})-p\ell_n(\btheta^{(t)})\\
=& \sum_{i=1}^{n} \log\left(\frac{p_i(\btheta^{(t+1)})}{p_i(\btheta^{(t)})}\right) +  \sum_{i=1}^{n}\log\left[\frac{\sum_{k\in [K]} \left\{w_k^{(t+1)}T^{(t+1)}_{\widetilde{Y}_i,k}\exp(\gamma^{(t+1)}_k + \langle\beta^{(t+1)}_k,g(X_i)\rangle)\right\}}{\sum_{k\in [K]} \left\{w_l^{(t)}T^{(t)}_{\widetilde{Y}_i,k}\exp(\gamma^{(t)}_k + \langle\beta^{(t)}_k,g(X_i)\rangle)\right\}} \right].
\end{split}
\]
Using the weights $\omega_{ik}^{(t)}$ defined in~\eqref{eq:e-step}, we can rewrite this as:
\[
\begin{split}
&p\ell_n(\btheta^{(t+1)})-p\ell_n(\btheta^{(t)})\\
=& \sum_{i=1}^{n} \log\left(\frac{p_i(\btheta^{(t+1)})}{p_i(\btheta^{(t)})}\right) +  \sum_{i=1}^{n}\log\left\{\sum_{k\in [K]}\omega_{ik}^{(t)}\frac{w_k^{(t+1)}T^{(t+1)}_{\widetilde{Y}_i,k}\exp(\gamma^{(t+1)}_k + \langle\beta^{(t+1)}_k,g(X_i)\rangle)}{w_k^{(t)}T^{(t)}_{\widetilde{Y}_i,k}\exp(\gamma^{(t)}_k + \langle\beta^{(t)}_k,g(X_i)\rangle)} \right\}.
\end{split}
\]
By Jensen's inequality and the concavity of $\log(\cdot)$, we obtain:
\[
\begin{split}
&p\ell_n(\btheta^{(t+1)})-p\ell_n(\btheta^{(t)})\\
\geq& \sum_{i=1}^{n} \log\left(\frac{p_i(\btheta^{(t+1)})}{p_i(\btheta^{(t)})}\right) +  \sum_{i=1}^{n}\sum_{k\in [K]}\omega_{ik}^{(t)}\log\left(\frac{w_k^{(t+1)}T^{(t+1)}_{\widetilde{Y}_i,k}\exp(\gamma^{(t+1)}_k + \langle\beta^{(t+1)}_k,g(X_i)\rangle)}{w_k^{(t)}T^{(t)}_{\widetilde{Y}_i,k}\exp(\gamma^{(t)}_k + \langle\beta^{(t)}_k,g(X_i)\rangle)} \right)\\
=&Q(\btheta^{(t+1)};\btheta^{(t)})-Q(\btheta^{(t)};\btheta^{(t)}).
\end{split}
\]
Since $\btheta^{(t+1)}$ maximizes $Q(\btheta;\btheta^{(t)})$ by construction, we have:
\[Q(\btheta^{(t+1)};\btheta^{(t)})-Q(\btheta^{(t)};\btheta^{(t)})\geq 0.\]
Consequently:
\[p\ell_n(\btheta^{(t+1)})-p\ell_n(\btheta^{(t)})\geq 0,\] 
which completes the proof of monotonic convergence.    
\end{proof}

\subsection{Incorporating prior knowledge on transition matrix}
\label{app:em_prior_knowledge_transition_matrix}
We present modifications to the M-step when prior knowledge about the transition matrix $\mT^*$ is available. 
Specifically, consider known lower bounds for the diagonal elements:
\[
T_{kk}^* \geq \xi_k, \quad \xi_k \in [0,1], \quad k = 1,\ldots,K
\]
where typical values might be $\xi_k = 0.8$. 
The updates for $\vw$ and $(\bgamma, \bbeta)$ remain unchanged, while the update for $\mT$ becomes a constrained optimization problem:
\begin{eqnarray*}
&\min_{\mT, \vs} &\left\{-\sum_{i=1}^{n}\sum_{l\in [K]}\mathbbm{1}(\widetilde{Y}_i=l)\sum_{k\in [K]}\omega_{ik}^{(t)}\log T_{lk}\right\}\\
&\text{subject to } & \sum_{l\in [K]} T_{lk} = 1,~k\in[K]\\
&& T_{kk} - \xi_k -s_k^2 = 0,~k\in [K]
\end{eqnarray*}
where $\vs = \{s_k\}_{k=1}^K$ are slack variables. 
The Lagrangian is:
{\small{\[\gL_{\mT, \vs} = -\sum_{i=1}^{n}\sum_{l\in [K]}\mathbbm{1}(\widetilde{Y}_i=l)\sum_{k\in [K]}\omega_{ik}^{(t)}\log T_{lk} - \sum_{k\in [K]}\zeta_k \left(\sum_{l\in [K]} T_{lk} -1\right) - \sum_{k\in [K]}\kappa_k \left(T_{kk} - \xi_k-s_k^2\right).\]}}
The KKT conditions yield:
\begin{align}
0 =&~\frac{\partial \gL_{\mT,\vs}}{\partial \zeta_k} = \sum_{l\in [K]} T_{lk} -1,~k\in[K] \label{eq:lambda_grad}\\
0 =&~\frac{\partial \gL_{\mT,\vs}}{\partial \kappa_k} = T_{kk} -\xi_k - s_k^2,~k\in [K] \label{eq:eta_grad}\\
0 =&~\frac{\partial \gL_{\mT,\vs}}{\partial s_k} = 2\kappa_k s_k,~k\in [K] \label{eq:supplementary_slackness}\\
0 =&~\frac{\partial \gL_{\mT,\vs}}{\partial T_{lk}} = -\frac{\sum_{i=1}^{n}\mathbbm{1}(\widetilde{Y}_i=l)\omega_{ik}^{(t)}}{T_{lk}} - \zeta_k-\kappa_k \mathbbm{1}(k=l),~l\in[K],~k\in [K].\label{eq:T_grad}
\end{align}
We analyze two cases based on condition~\eqref{eq:supplementary_slackness}:
\begin{itemize}[leftmargin=*]
\item Case I: Unconstrained solution ($\kappa_k=0$). 
In this case,~\eqref{eq:T_grad} implies that 
$T_{lk} = -\sum_{i=1}^{n}\mathbbm{1}(\widetilde{Y}_i=l)\omega_{ik}^{(t)}/\zeta_k$.
Summing over $k$ and using \eqref{eq:lambda_grad}, we obtain
\[T_{lk} = \frac{\sum_{i=1}^{n}\mathbbm{1}(\widetilde{Y}_i=l)\omega_{ik}^{(t)}}{\sum_{i=1}^{n}\omega_{ik}^{(t)}}.\]
From~\eqref{eq:eta_grad}, we require that
\[T_{kk} = -\frac{\sum_{i=1}^{n}\mathbbm{1}(\widetilde{Y}_i=k)\omega_{ik}^{(t)}}{\sum_{i=1}^{n}\omega_{ik}^{(t)}} \geq \xi_k.\]
If this condition is not met, this case is infeasible.

\item Case II: Active constraint ($s_k=0$).
When $s_k=0$,~\eqref{eq:eta_grad} implies that $T_{kk}=\xi_k$ and~\eqref{eq:T_grad} implies that 
\[-\sum_{i=1}^{n}\mathbbm{1}(\widetilde{Y}_i=l)\omega_{ik}^{(t)} = \zeta_k T_{lk}  + \kappa_k \mathbbm{1}(k=l) T_{lk}\]
Summing over $l$ on both sides of this equation, we get
\[-\sum_{i=1}^{n}\omega_{ik}^{(t)} = \zeta_k  + \kappa_k T_{kk} = \zeta_k + \kappa_k \xi_k\]
Plugging this into~\eqref{eq:T_grad}, we obtain
\be
\label{eq:T_penalty}
T_{lk} = \frac{\sum_{i=1}^{n}\mathbbm{1}(\widetilde{Y}_i=l)\omega_{ik}^{(t)}}{\sum_{i=1}^{n}\omega_{ik}^{(t)} + \kappa_k \xi_k -\kappa_k \mathbbm{1}(k=l)}.
\ee
For the case when $l = k$, we have
\[\xi_k = T_{kk} = \frac{\sum_{i=1}^{n}\mathbbm{1}(\widetilde{Y}_i=k)\omega_{ik}^{(t)}}{\sum_{i=1}^{n}\omega_{ik}^{(t)} + \kappa_k (\xi_k -1)}\]
which implies
\[\kappa_k = \frac{\sum_{i=1}^{n}\{\mathbbm{1}(\widetilde{Y}_i=k)-\xi_k\}\omega_{ik}^{(t)}}{  \xi_k(\xi_k -1)}.\]
Substituting this expression for $\kappa_k$ into the previous equation gives the closed-form update for $\mT$.

\end{itemize}
To summarize, the complete update combines both cases:
\[
T_{lk}^{(t+1)}=
\begin{cases}
\frac{\sum_{i=1}^{n}\mathbbm{1}(\widetilde{Y}_i=l)\omega_{ik}^{(t)}}{\sum_{i=1}^{n}\omega_{ik}^{(t)}} & \text{if } \frac{\sum_{i=1}^{n}\mathbbm{1}(\widetilde{Y}_i=k)\omega_{ik}^{(t)}}{\sum_{i=1}^{n}\omega_{ik}^{(t)}}\geq \xi_k\\    
\frac{\sum_{i=1}^{n}\mathbbm{1}(\widetilde{Y}_i=l)\omega_{ik}^{(t)}}{\sum_{i=1}^{n}\omega_{ik}^{(t)} + \kappa_k^{(t+1)} \{\xi_k - \mathbbm{1}(k=l)\}} &\text{otherwise}
\end{cases}
\]
where 
\[\kappa_k^{(t+1)} = \frac{\sum_{i=1}^{n}\{\mathbbm{1}(\widetilde{Y}_i=k)-\xi_k\}\omega_{ik}^{(t)}}{\xi_k(\xi_k -1)}.\]

\subsection{Penalized EM algorithm}
\label{app:emp_penalization}
To enhance numerical stability and accelerate convergence, we introduce penalties on the diagonal elements $\{T_{kk}\}_{k=0}^{K-1}$ of the transition matrix. 
Let $\eta_k \geq 0$ denote prespecified penalty parameters. 
The penalized profile log empirical likelihood becomes:
\[
p\ell_{n}(\btheta) =\sum_{i=1}^{n}\sum_{l\in [K]}\mathbbm{1}(\widetilde{Y}_i=l)\log \sum_{k\in [K]} \left\{w_k T_{lk}\exp(\gamma_k + \beta_k^{\top}g(x))\right\} + \sum_{i=1}^{n}\log p_i(\btheta) + \sum_{k\in [K]}\eta_k \log T_{kk}.
\]
This penalty structure prevents the diagonal elements from approaching zero while maintaining the validity of the EM algorithm framework. 
Notably, the updates for $\vw$ and $(\bgamma, \bbeta)$ remain unchanged from the unpenalized version. 

For the transition matrix $\mT$, we solve the constrained optimization problem:
\begin{eqnarray*}
&\max_{\mT} &\left\{\sum_{i=1}^{n}\sum_{l\in [K]}\mathbbm{1}(\widetilde{Y}_i=l)\sum_{k\in [K]}\omega_{ik}^{(t)}\log T_{lk} + \sum_{k\in [K]}\gamma_k \log T_{kk}\right\}\\
&\text{subject to } & \sum_{l\in [K]} T_{lk} = 1,~k\in [K].
\end{eqnarray*}
The corresponding Lagrangian function is:
\[\gL_{\mT} = \sum_{i=1}^{n}\sum_{l\in [K]}\mathbbm{1}(\widetilde{Y}_i=l)\sum_{k\in [K]}\omega_{ik}^{(t)}\log T_{lk} - \sum_{k\in [K]}\zeta_k \left(\sum_{l\in [K]} T_{lk} -1\right) + \sum_{k\in [K]}\eta_k \log T_{kk}.\]
The optimality conditions yield the following system of equations:
\begin{align*}
0 =&~\frac{\partial \gL_{\mT}}{\partial \zeta_k} = \sum_{l\in [K]} T_{lk} -1,~k\in [K]  \\
0 =&~\frac{\partial \gL_{\mT}}{\partial T_{lk}} = \frac{\sum_{i=1}^{n}\mathbbm{1}(\widetilde{Y}_i=l)\omega_{ik}^{(t)}}{T_{lk}} - \zeta_k+\frac{\eta_k \mathbbm{1}(k=l)}{
T_{kk}},~k\in[K],~l\in [K].
\end{align*}
Solving for the equalities above, we then have
\[T_{lk}^{(t+1)} = \frac{\sum_{i=1}^{n}\mathbbm{1}(\widetilde{Y}_i=l)\omega_{ik}^{(t)} + \eta_k \mathbbm{1}(k=l)}{\sum_{i=1}^{n}\omega_{ik}^{(t)} + \eta_k}.\]

%% file: sections/appendices/more_exp.tex
\section{Experimental details and more results}
\subsection{Density ratio model under normal model}
\label{app:DRM}
In two of our simulation settings, we generate the training set from the following model:
\[X| Y=k\sim N(\mu_k^*, \Sigma^*).\]
Then the density ratios of $X| Y=k$ and $X| Y=0$ are:
\[
\begin{split}
\frac{dP_k^*}{dP_0^*}(x) =&~\frac{\exp(-(x-\mu_k^*)^{\top}(\Sigma^*)^{-1}(x-\mu_k^*)/2)}{\exp(-(x-\mu_0^*)^{\top}(\Sigma^*)^{-1}(x-\mu_0^*)/2)}\\
=&~\exp\left((\mu_k^*-\mu_0^*)^{\top}(\Sigma^*)^{-1}x - \frac{1}{2}\left\{(\mu_k^*)^{\top}(\Sigma^*)^{-1}\mu_k^* - (\mu_0^*)^{\top}(\Sigma^*)^{-1}\mu_0^*\right\}\right)
\end{split}
\]
Therefore, we have $g(x) = (1,x)^{\top}$ with 
\[\gamma_k^* = - \frac{1}{2}\left\{(\mu_k^*)^{\top}(\Sigma^*)^{-1}\mu_k^* - (\mu_0^*)^{\top}(\Sigma^*)^{-1}\mu_0^*\right\},~\beta_k^* = (\mu_k^*-\mu_0^*)^{\top}(\Sigma^*)^{-1}\]
and $(\gamma_k^*)^{\dagger} = \gamma_k^* + \log\left(\pi_k^*/\pi_0^*\right)$.
In the special case where $\Sigma^*= (\sigma^*)^2I_d$, the parameters in the DRM reduce to
\[\gamma_k^* = -\frac{(\mu_k^*)^{\top}\mu_k^*-(\mu_0^*)^{\top}\mu_0^*}{2(\sigma^*)^2},~\beta_k^* = \frac{\mu_k^*-\mu_0^*}{(\sigma^*)^2}.\]

When the covariance matrices are different, 
\[X| Y=k\sim N(\mu_k^*, (\sigma_k^*)^2 I_{d})\]
We then have the density ratios of $X| Y=k$ and $X| Y=1$ as:
\[
\begin{split}
\frac{dP_k^*}{dP_0^*}(x) =&~\frac{(\sigma_k^*)^{-d}\exp(-(x-\mu_k^*)^{\top}(\Sigma_k^*)^{-1}(x-\mu_k^*)/2)}{(\sigma_0^*)^{-d}\exp(-(x-\mu_0^*)^{\top}(\Sigma_0^*)^{-1}(x-\mu_0^*)/2)}\\
=&~\exp\left(\frac{x^{\top}\{(\Sigma_0^*)^{-1}-(\Sigma_k^*)^{-1}\}x}{2} + \{(\Sigma_k^*)^{-1}\mu_k^*-(\Sigma_0^*)^{-1}\mu_0^*\}^{\top}x\right) \\
&\times \exp\left(\frac{(\mu_0^*)^{\top}(\Sigma_0^*)^{-1}\mu_0^* - (\mu_k^*)^{\top}(\Sigma_k^*)^{-1}\mu_k^*}{2} + d\log \frac{\sigma_0^*}{\sigma_k^*}\right)
\end{split}
\]
Under the special case where $\Sigma_k^* = (\sigma_k^*)^2 I_{d}$, we have $g(x) = (x^{\top}, x^{\top}x)^{\top}$ with $\gamma_k^* = \|\mu_0^*\|^2/\{2(\sigma_0^*)^2\}-\|\mu_k^*\|^2/\{2(\sigma_k^*)^2\} + d\log(\sigma_0^*/\sigma_k^*)$ and $\beta_k^* = ((\mu_k^*/(\sigma_k^*)^2-\mu_0^*/(\sigma_0^*)^2)^{\top},((\sigma_0^*)^{-2}-(\sigma_k^*)^{-2})/2)^{\top}$.

\subsection{More experiment results}
\label{app:more_exp}

\subsubsection{Model fitting under noisy labeled data}
\label{app:model_fitting}
To show the effectiveness of our method in estimating the true conditional probability of $Y|X=x$ with noisy labeled data, we compare our EL-based estimator in~\eqref{eq:conditional_distribution_estimate} with two baselines:
\begin{itemize}
    \item The \textbf{Vanilla} estimator, which naively fits a multinomial logistic regression model to the noisy labeled training data;
    \item The \textbf{Oracle} estimator, which fits the same model to the clean labeled dataset.
\end{itemize}
The data generating approach, and the multinomial logistic regression model are the same as that described in Section~\ref{sec:simulation_multiclass}. 
We evaluate their performance by computing the average MSE of the estimated regression coefficients in~\eqref{eq:posterior} in each setting.
Figure~\ref{fig:parameter_estimation} summarizes the results across different sample sizes and noise levels $\eta$.
\begin{figure}[htbp]
\centering
\includegraphics[width=0.3\textwidth]{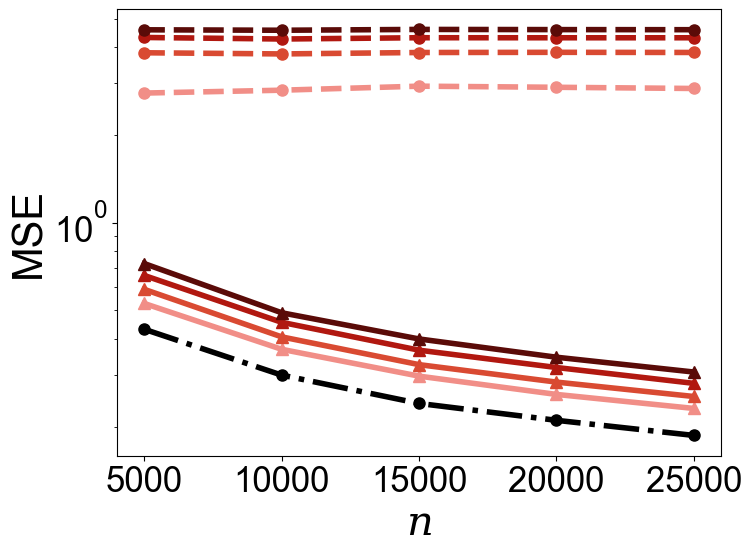}
\includegraphics[width=0.3\textwidth]{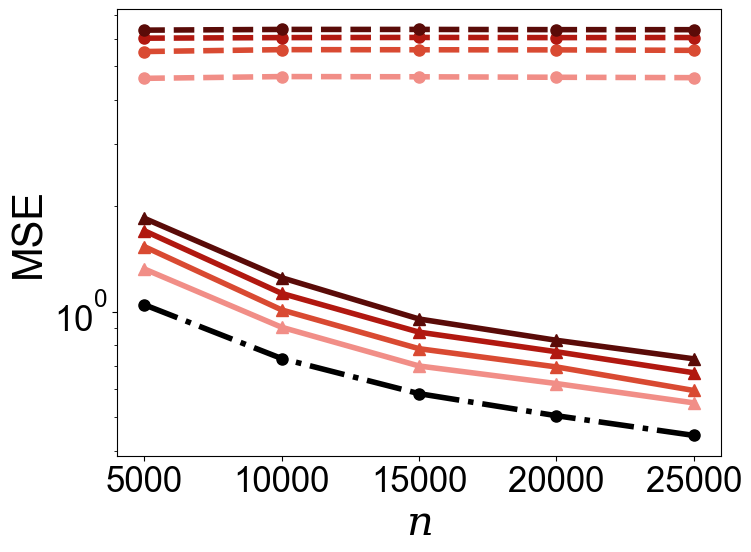}
\includegraphics[width=0.3\textwidth]{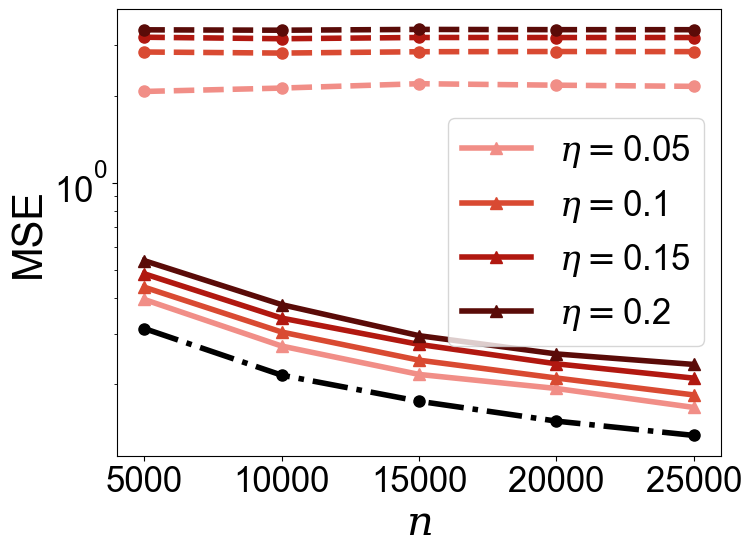}
\caption{Mean squared error (MSE) of regression coefficients for different estimators: Oracle (dash-dotted line), our method (solid line), and Vanilla (dashed line). 
Results are shown under varying noise levels (colors) and sample sizes $n$, with cases (a)–(c) displayed from left to right.}
\label{fig:parameter_estimation}
\end{figure}
The MSE of our estimator decreases with sample size $n$ at a rate of order $n^{-1}$, consistent with theory.
As the noise level $\eta$ increases, both our method and the vanilla estimator exhibit larger MSE.
However, our estimator remains close to the oracle, while the vanilla estimator deteriorates significantly.
This highlights that the vanilla estimator is not consistent under label noise, as its error does not decrease with $n$.

\subsubsection{Violin plots for case (b) and (c)}
\label{app:violin_plots}
The violin plots of the class-specific errors and the objective value under different noise levels $\eta$ for cases (b) and (c) are shown in Figure~\ref{fig:caseb} and Figure~\ref{fig:casec}, respectively. 
In both cases, the mean class-specific error of our method is almost as close as that of the oracle, with only a slightly larger variance in some settings. 
The gap between our method and the oracle becomes smaller as $n$ increases. 
In contrast, the vanilla approach is either too conservative or too aggressive, leading to unreliable class-specific error control and a substantially larger objective value.

\begin{figure}[!ht]
    \centering
    \includegraphics[width=0.9\linewidth]{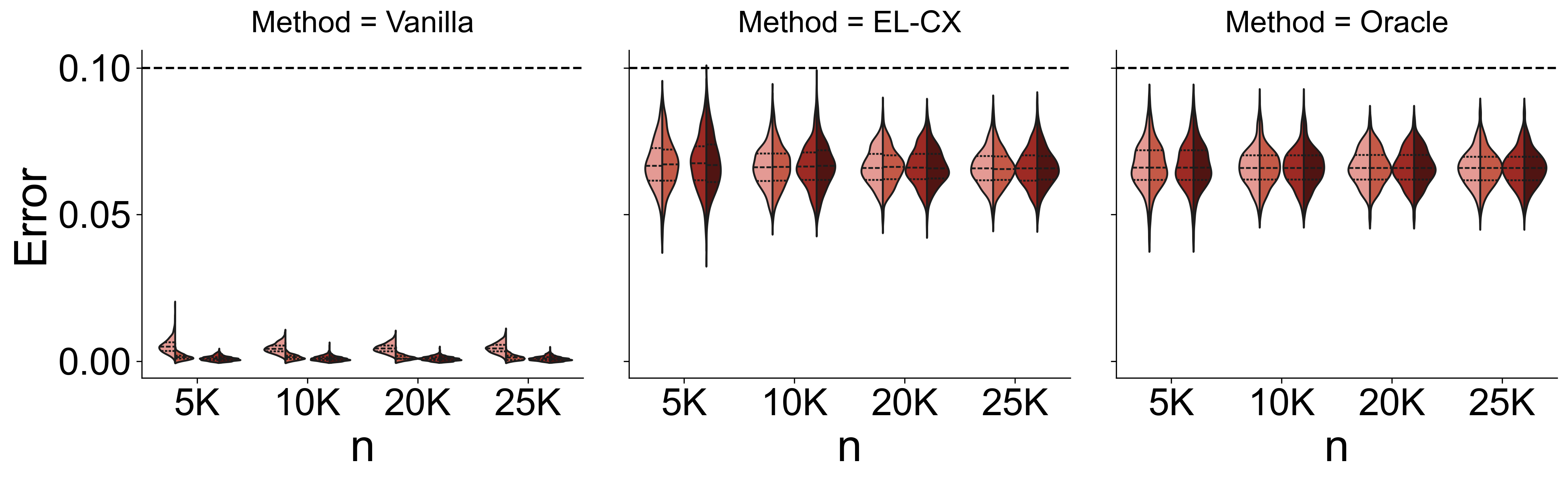}
    \includegraphics[width=0.9\linewidth]{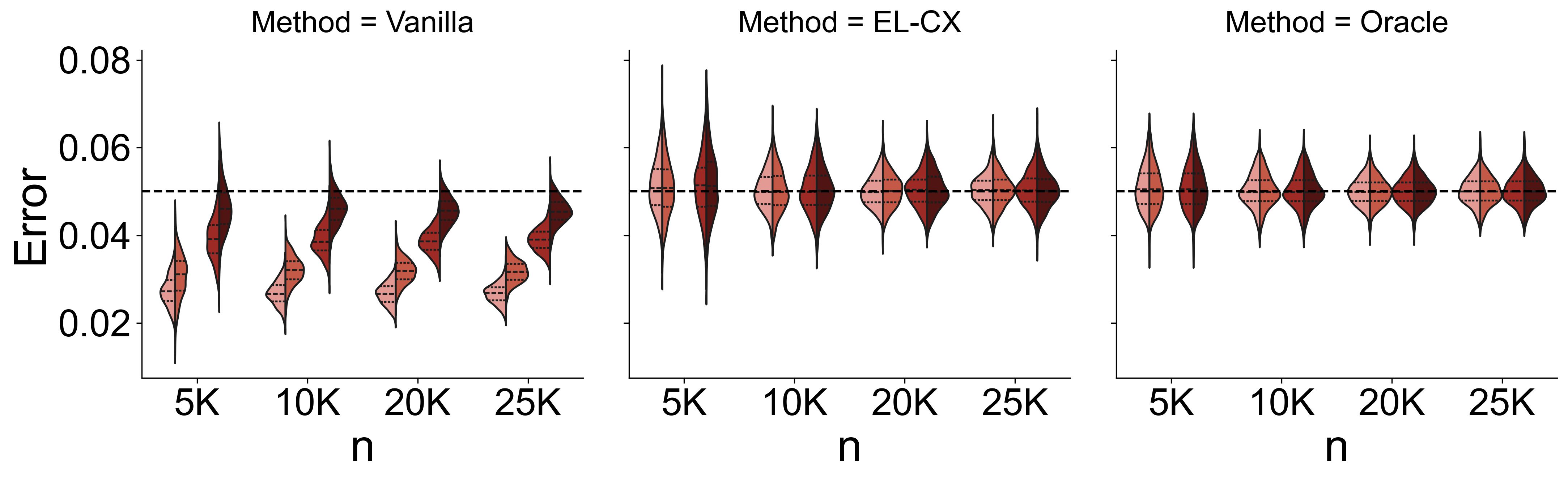}
    \includegraphics[width=0.9\linewidth]{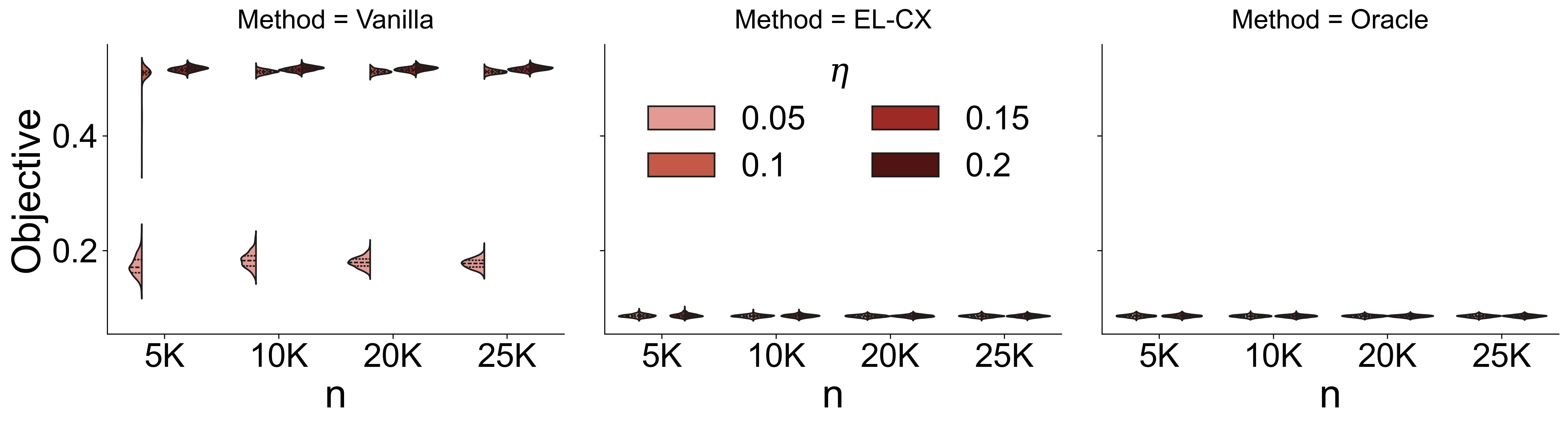}
    \caption{Violin plots (top to bottom) show the misclassification errors for Class 0, Class 3, and the objective function value under Case (b), computed over $R = 500$ repetitions. Results compare different methods across varying sample sizes ($n$) and noise levels ($\eta$). The dashed black lines in the first two rows mark the target misclassification errors.}
    \label{fig:caseb}
\end{figure}

\begin{figure}[!ht]
    \centering
    \includegraphics[width=0.9\linewidth]{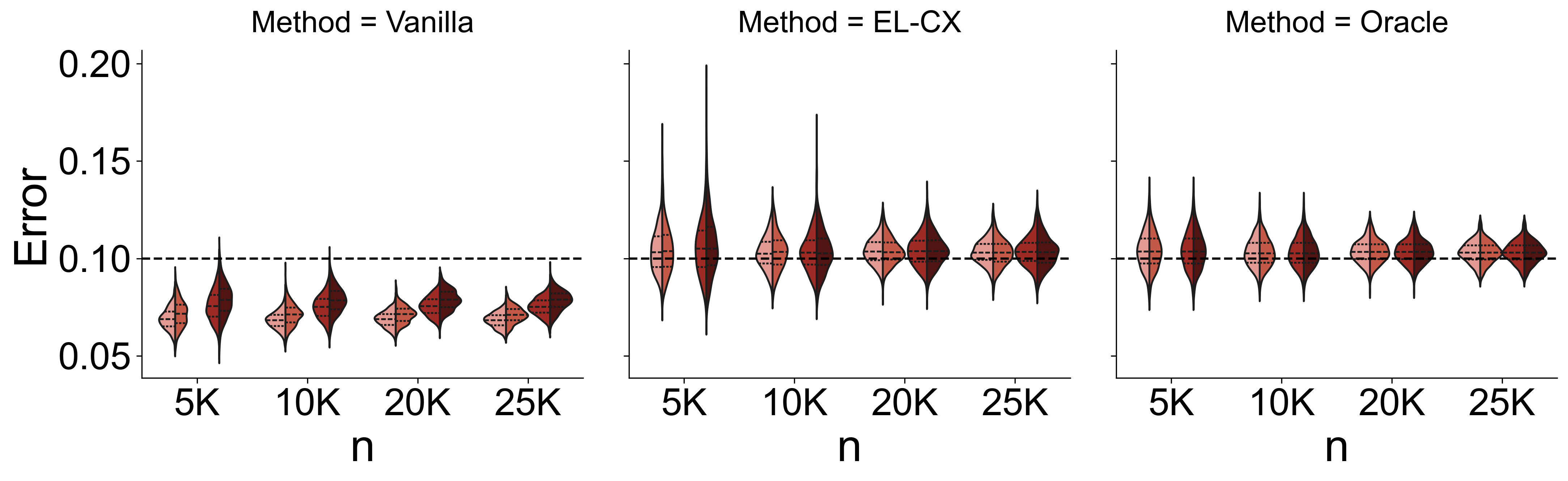}
    \includegraphics[width=0.9\linewidth]{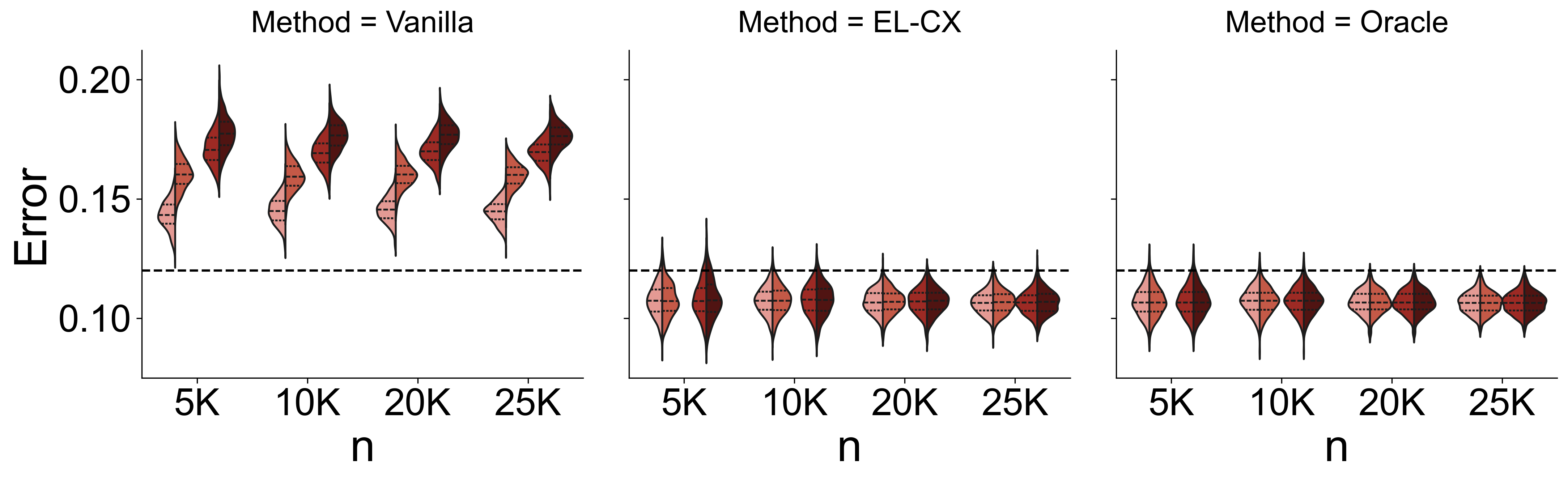}
    \includegraphics[width=0.9\linewidth]{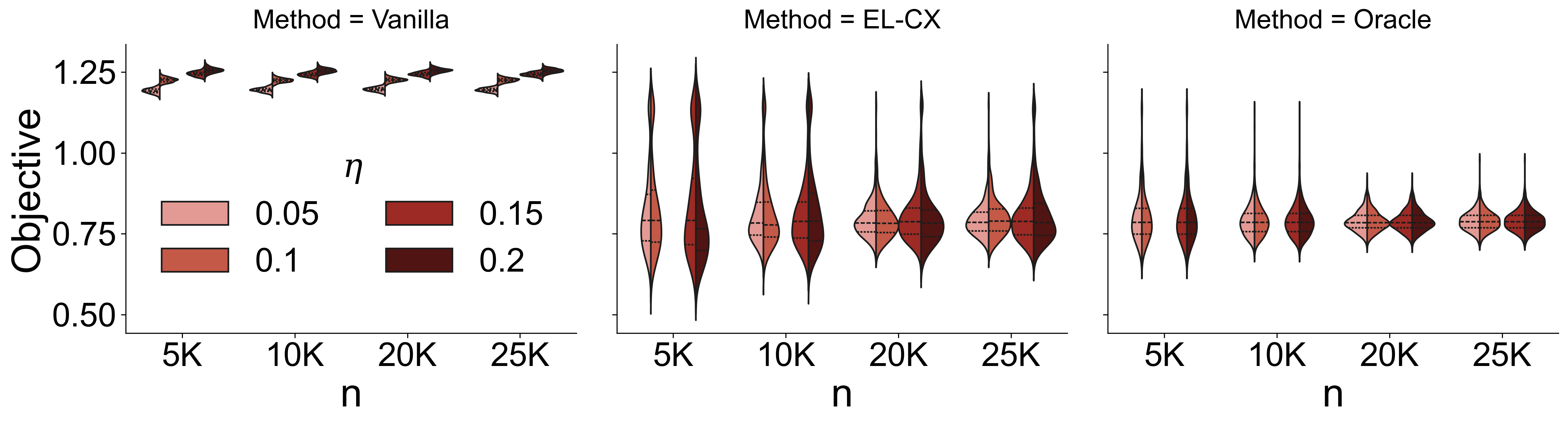}
    \caption{Violin plots (top to bottom) show the misclassification errors for Class 0, Class 1, and the objective function value under Case (c), computed over $R = 500$ repetitions. Results compare different methods across varying sample sizes ($n$) and noise levels ($\eta$). The dashed black lines in the first two rows mark the target misclassification errors.}
    \label{fig:casec}
\end{figure}

\subsubsection{NP umbrella type control}
\label{app:npumbrella}
In this section, we discuss how to integrate our EL-based estimator under noisy labels with the framework of~\citet{tong2018neyman} to achieve error control in binary classification.

The NP umbrella algorithm, proposed by~\citet{tong2018neyman}, provides a flexible procedure that accommodates various score functions from binary classifiers. 
It produces a classifier that controls the type I error at the target level with high probability while minimizing the type II error. 
The procedure is as follows:
\begin{enumerate}
\item Split the class-$0$ data into two disjoint subsets: $S_0^{\text{train}}$ and $ S_0^{\text{cal}}$.

\item Train a base scoring function $f: \gX \to \sR$ using $S_0^{\text{train}}$ together with all class-$1$ samples, where larger values of $f(x)$ indicate a higher likelihood that $x$ belongs to class $1$.

\item Evaluate $f$ on $S_0^{\text{cal}}$, obtaining scores $\{T_j = f(X_j): X_j \in S_0^{\text{cal}}\}$. 

\item Sort these scores in increasing order: $T_{(1)} \leq T_{(2)} \leq \cdots \leq T_{(m)}$, where $m = |S_0^{\text{cal}}|$.

\item Define the NP umbrella classifier by choosing an index $k^*$ and setting
\[\phi(x) = \mathbb{I}\{ f(x) > T_{(k^*)} \}.\]
Since $S_0^{\text{cal}}$ contains only class-$0$ samples, any misclassification corresponds to a type I error. 
Thus, the number of misclassified calibration samples follows a $\text{Binomial}(m, \alpha)$ distribution. 
The index $k^*$ is chosen as the smallest integer satisfying
\[
\sP\big( \text{Binomial}(m, \alpha) \geq  k^*  \big)= \delta,
\]
where $\delta$ is a user-specified tolerance level (e.g., $\delta = 0.05$).
By construction, with probability at least $1 - \delta$ over the calibration set, the classifier $\phi$ satisfies $P_0^*(\phi \neq 0) \leq \alpha$.
\end{enumerate}

Although effective for controlling type I error,~\citet{yao2023asymmetric} showed that the NP umbrella algorithm becomes too conservative when the training data are contaminated, i.e., when labels are corrupted as in Assumption~\ref{assumption:instance_independent_noise}. 
This is because ignoring label noise in the NP umbrella algorithm only guarantees $\widetilde{P}_0^*(\phi \neq 0)\leq \alpha$ instead of $P_0^*(\phi \neq 0) \leq \alpha$.
To address this issue, they extended the NP umbrella algorithm to handle the noisy-label setting as follows:
\begin{enumerate}
    \item Let $\widetilde{S}_0$ and $\widetilde{S}_1$ respectively be the corrupted class $0$ and $1$ training set.
    \item Partition $\widetilde{S}_0$ into three random, disjoint, nonempty subsets: $\widetilde{S}_0^{\text{train}}, \widetilde{S}_0^{\text{est}},$ and $\widetilde{S}_0^{\text{cal}}$.
    Similarly, partition $\widetilde{S}_1$ into two random, disjoint, nonempty subsets: $\widetilde{S}_1^{\text{train}}$ and $\widetilde{S}_1^{\text{est}}$.

    \item  Train a base scoring function $f: \gX \to \sR$ using $\widetilde{S}_0^{\text{train}}\cap \widetilde{S}_1^{\text{train}}$, where larger values of $f(x)$ indicate a higher likelihood that $x$ belongs to corrupted class $1$.
    
    \item Evaluate $f$ on $\widetilde{S}_0^{\text{cal}}$, obtaining scores $\{T_j = f(X_j): X_j \in \widetilde{S}_0^{\text{cal}}\}$. 
    
    \item Sort these scores in increasing order: $T_{(1)} \leq T_{(2)} \leq \cdots \leq T_{(m)}$, where $m = |\widetilde{S}_0^{\text{cal}}|$. 
    These serve as the candidate thresholds, just as in the original NP umbrella algorithm. 
    
    \item Define the NP umbrella classifier by choosing an index $k^*$ and setting
    \[\phi(x) = \mathbb{I}\{ f(x) > T_{(k^*)} \}.\]
    To account for label noise, the noise-adjusted NP umbrella algorithm (with known corruption levels) selects $k^*$ as  
    \[k^* =  \min\{k\in\{1, \ldots, n\}:\alpha_{k,\delta} - \widehat{D}^+(T_{(k)}) \leq \alpha \}\,\]
    where $\alpha_{k,\delta}$ satisfies $\sP\big(\text{Binomial}(m, \alpha_{k,\delta}) \geq  k^*  \big)= \delta$,
    $\widehat{D}^+(\cdot) =\widehat D(\cdot)\vee 0:= \max(\widehat D(\cdot), 0)$, and 
    \be
    \label{eq:typeIerrordiff}
    \widehat D(\cdot ) = \frac{1-m_0^*}{m_0^*-m_1^*}\left(\widetilde{P}_0^{f}(\cdot) - \widetilde{P}_1^{f}(\cdot)\right).
    \ee
    Here, $\widehat{D}(t)$ estimates the gap between type I errors computed on clean versus corrupted data, with $\widetilde{P}_0^{f}(t)$ and $\widetilde{P}_1^{f}(t)$ denoting empirical estimates of $\sP(f(X)\leq t \mid \widetilde{Y}=0)$ and $\sP(f(X)\leq t \mid \widetilde{Y}=1)$, based on $\widetilde{S}_0^{\text{est}}$ and $\widetilde{S}_1^{\text{est}}$, respectively.  
\end{enumerate}
\begin{table}[!ht]
\centering
\caption{Type I error violation rate and II error for binary classification based on $500$ repetitions. 
NPC$^{*}$ represents the oracle method with known $m_0^*$ and $m_1^*$. }
\label{tab:binary_npc_comparison}
\resizebox{\textwidth}{!}{
\begin{tabular}{ccccc@{\hspace{2em}}ccc@{\hspace{2em}}ccc@{\hspace{2em}}ccc}
\toprule
\multirow{3}{*}{Case} &\multirow{3}{*}{$n$} & \multicolumn{3}{c}{$m_0^*=0.95,m_1^*=0.05$} & \multicolumn{3}{c}{$m_0^*=0.95,m_1^*=0.05$} & \multicolumn{3}{c}{$m_0^*=0.9,m_1^*=0.1$} & \multicolumn{3}{c}{$m_0^*=0.9,m_1^*=0.1$} \\
&& \multicolumn{3}{c}{$\alpha=0.05,\delta=0.05$} & \multicolumn{3}{c}{$\alpha=0.1,\delta=0.1$} & \multicolumn{3}{c}{$\alpha=0.05,\delta=0.05$} & \multicolumn{3}{c}{$\alpha=0.1,\delta=0.1$} \\ 
\cmidrule(lr{1.5em}){3-5}
\cmidrule(lr{1.5em}){6-8}
\cmidrule(lr{1.5em}){9-11}
\cmidrule(lr){12-14}
&&NPC$^{*}$ & NPC & NPC+ & NPC$^{*}$ & NPC & NPC+ & NPC$^{*}$ & NPC & NPC+ & NPC$^{*}$ & NPC & NPC+\\ 
\midrule
&\multicolumn{12}{c}{Type I error violation rate}\\
\midrule
\multirow{3}{*}{A} & 1000 & 0.038 & 0.000 & 0.048 & 0.080 & 0.000 & 0.100 & 0.044 & 0.000 & 0.044 & 0.090 & 0.000 & 0.132 \\
                   & 2000 & 0.042 & 0.000 & 0.042 & 0.094 & 0.000 & 0.122 & 0.048 & 0.000 & 0.042 & 0.110 & 0.000 & 0.116 \\
                   & 5000 & 0.074 & 0.000 & 0.058 & 0.106 & 0.000 & 0.100 & 0.054 & 0.000 & 0.058 & 0.106 & 0.000 & 0.116 \\
\midrule
\multirow{3}{*}{B} & 1000 & 0.036 & 0.000 & 0.016 & 0.086 & 0.000 & 0.044 & 0.046 & 0.000 & 0.006 & 0.102 & 0.000 & 0.044 \\
                   & 2000 & 0.052 & 0.000 & 0.008 & 0.074 & 0.000 & 0.040 & 0.048 & 0.000 & 0.006 & 0.090 & 0.000 & 0.028 \\
                   & 5000 & 0.060 & 0.000 & 0.008 & 0.118 & 0.000 & 0.042 & 0.046 & 0.000 & 0.002 & 0.104 & 0.000 & 0.016 \\
\midrule
\multirow{3}{*}{C} & 1000 & 0.046 & 0.000 & 0.132 & 0.106 & 0.002 & 0.248 & 0.054 & 0.000 & 0.124 & 0.108 & 0.000 & 0.246 \\
                   & 2000 & 0.040 & 0.000 & 0.204 & 0.100 & 0.000 & 0.306 & 0.052 & 0.000 & 0.182 & 0.084 & 0.000 & 0.290 \\
                   & 5000 & 0.058 & 0.000 & 0.394 & 0.084 & 0.000 & 0.468 & 0.056 & 0.000 & 0.378 & 0.112 & 0.000 & 0.456 \\
\midrule
&\multicolumn{12}{c}{Type II error}\\
\midrule
\multirow{3}{*}{A} & 1000 & 0.398 & 0.567 & 0.394 & 0.218 & 0.305 & 0.218 & 0.427 & 0.730 & 0.422 & 0.227 & 0.427 & 0.226 \\
                   & 2000 & 0.352 & 0.515 & 0.352 & 0.201 & 0.276 & 0.201 & 0.368 & 0.674 & 0.368 & 0.205 & 0.388 & 0.206 \\
                   & 5000 & 0.319 & 0.470 & 0.320 & 0.188 & 0.260 & 0.188 & 0.326 & 0.630 & 0.327 & 0.190 & 0.363 & 0.190 \\
\midrule
\multirow{3}{*}{B} & 1000 & 0.180 & 0.418 & 0.189 & 0.108 & 0.158 & 0.114 & 0.205 & 0.700 & 0.218 & 0.113 & 0.274 & 0.121 \\
                   & 2000 & 0.164 & 0.310 & 0.172 & 0.100 & 0.145 & 0.105 & 0.172 & 0.637 & 0.184 & 0.103 & 0.221 & 0.111 \\
                   & 5000 & 0.154 & 0.240 & 0.162 & 0.092 & 0.138 & 0.098 & 0.157 & 0.584 & 0.169 & 0.094 & 0.203 & 0.102 \\
\midrule
\multirow{3}{*}{C} & 1000 & 0.461 & 0.629 & 0.411 & 0.247 & 0.341 & 0.224 & 0.487 & 0.755 & 0.425 & 0.255 & 0.462 & 0.230 \\
                   & 2000 & 0.406 & 0.571 & 0.362 & 0.227 & 0.310 & 0.208 & 0.421 & 0.704 & 0.373 & 0.233 & 0.422 & 0.212 \\
                   & 5000 & 0.369 & 0.525 & 0.328 & 0.213 & 0.291 & 0.194 & 0.378 & 0.666 & 0.333 & 0.214 & 0.399 & 0.195 \\
\bottomrule
\end{tabular}}
\end{table}

This approach requires knowledge of the parameters $m_0^* = \sP(Y=0|\widetilde{Y}=0)$ and $m_1^* = \sP(Y=0 |\widetilde{Y}=1)$ in the transition matrix; we refer to this method as \textbf{NPC$^*$}.
When these quantities are unknown,~\citet{yao2023asymmetric} further extended the method to the case where only a lower bound on $m_1^*$ and an upper bound on $m_0^*$ are available. 
In this approach, they replace $m_0^*$ and $m_1^*$ with their bounds in the last step.
We refer to this version as \textbf{NPC}.

As shown in Table~\ref{tab:binary} in the main paper, relaxing the assumption of knowing $m_0^*$ and $m_1^*$ can lead to performance as poor as the vanilla method that ignores label noise. 
Thus, accurate estimates of $m_0^*$ and $m_1^*$ are crucial. 
Our proposed method natually provides an effective way to obtain these estimates. 
Specifically, we use the full dataset to compute the MELE estimators for $m_0^*$ and $m_1^*$ and then replace the unknown values with these estimates in the final step. ]
We refer to this version as \textbf{NPC+}.

We then compare NPC$^*$, NPC, and NPC+ following the same experiment setting in Section~\ref{sec:binary_exp}.
The results are reported in Table~\ref{tab:binary_npc_comparison}.
We report the proportion of repetitions that violate the type I error and the corresponding estimated type II error across $500$ repetitions.
The results show that NPC is more conservative, as it never violates the type I error but exhibits substantially larger type II errors in all cases.
Our NPC+ method performs comparably to NPC$^*$ under cases A and B, where the models are correctly specified.
Under model misspecification in case C, however, NPC+ has a higher chance of violating the type I error compared to NPC$^*$.

\subsubsection{Real data experiment: dry bean dataset}
\label{app:dry_bean}
The Dry Bean dataset contains $13,611$ instances of dry beans along with their corresponding types~\citep{koklu2020multiclass}. 
Each instance is described by $16$ features capturing the beans' physical properties, derived from shape, size, and color measurements. 
Specifically, the features include $11$ geometric, $1$ shape-related, and $2$ color-related attributes. 
The dataset contains seven types of dry beans with the following sample sizes: Barbunya ($1,322$), Bombay ($522$), Cali ($1,630$), Dermosan ($3,546$), Horoz ($1,928$), Seker ($2,027$), and Sira ($2,636$). 
For convenience, we recode the types into classes $0$ through $6$. 
The prediction task is to classify the type of dry bean based on its $16$ physical features. 
The dataset is publicly available from the UCI Machine Learning Repository at this~\href{https://archive.ics.uci.edu/ml/datasets/Dry+Bean+Dataset}{link}.

Since the features are on different scales, we standardize each feature by subtracting its sample mean and scaling to unit variance. 
We then randomly split the data into $80$\% training and $20$\% test sets while preserving class proportions. 
Noisy labels for the training set are generated according to the contamination model described in Section~\ref{sec:simulation_multiclass}, and this procedure is repeated $100$ times.

We consider the following NPMC problem from~\citet{tian2024neyman}:
\begin{equation*}
\begin{split}
\min_{\phi} &\quad \frac{1}{4}\left[P_2^*(\{\phi(X)\neq 2\}) + P_4^*(\{\phi(X)\neq 4\}) + P_5^*(\{\phi(X)\neq 5\}) + P_6^*(\{\phi(X)\neq 6\})\right] \\
\text{s.t.} &\quad P_0^*(\{\phi(X)\neq 0\}) \leq 0.05,  \quad P_1^*(\{\phi(X)\neq 1\}) \leq 0.01, \quad P_3^*(\{\phi(X)\neq 3\}) \leq 0.03.
\end{split}
\end{equation*}

We evaluate our proposed EL-CX method along with other baselines described in Section~\ref{sec:simulation_multiclass}, using multinomial logistic regression as the base classifier with the same feature mapping $g(x)$ and $L_2$ regularization. 
Performance, based on 100 repetitions, is summarized in Table~\ref{tab:drybean}.
\begin{table}[!htbp]
\centering
\caption{Mean excessive risk across target classes and objective function values, grouped by noise level $\eta$ and method.}
\label{tab:drybean}
\resizebox{\textwidth}{!}{
\begin{tabular}{@{}c@{\hspace{2em}}ccc>{\columncolor{RoyalBlue!5}}c@{\hspace{2em}}ccc>{\columncolor{RoyalBlue!5}}c@{\hspace{2em}}ccc>{\columncolor{RoyalBlue!5}}c@{\hspace{2em}}ccc>{\columncolor{RoyalBlue!5}}c@{}}
\toprule
\multirow{2}{*}{Method} & \multicolumn{4}{c}{$\eta=0.05$}&\multicolumn{4}{c}{$\eta=0.10$}&\multicolumn{4}{c}{$\eta=0.15$}&\multicolumn{4}{c}{$\eta=0.20$}\\ 
\cmidrule(lr{2em}){2-5}
\cmidrule(lr{2em}){6-9}
\cmidrule(lr{2em}){10-13}
\cmidrule(lr{-0.5em}){14-17}
&Class 0 & Class 1 & Class 3 & Obj.&Class 0 & Class 1 & Class 3 & Obj.&Class 0 & Class 1 & Class 3 & Obj&Class 0 & Class 1 & Class 3 & Obj\\
\midrule
Na\"ive&0.059&\negred{-0.001}&0.083&0.070&0.071&0.002&0.094&0.073&0.080&0.003&0.100&0.076&0.087&0.007&0.105&0.078\\
Vanilla&\negred{-0.046}&\negred{-0.009}&\negred{-0.024}&0.961&\negred{-0.046}&\negred{-0.010}&\negred{-0.022}&0.979&\negred{-0.046}&\negred{-0.010}&\negred{-0.022}&0.980&\negred{-0.046}&\negred{-0.010}&\negred{-0.021}&0.980\\
\textbf{Ours}&\negred{-0.046}&\negred{-0.007}&\negred{-0.011}&0.322&\negred{-0.045}&\negred{-0.008}&\negred{-0.013}&0.378&\negred{-0.046}&\negred{-0.009}&\negred{-0.014}&0.432&\negred{-0.045}&\negred{-0.009}&\negred{-0.015}&0.469\\
Oracle&\negred{-0.045}&\negred{-0.002}&\negred{-0.010}&0.285&\negred{-0.045}&\negred{-0.002}&\negred{-0.010}&0.285&\negred{-0.045}&\negred{-0.002}&\negred{-0.010}&0.285&\negred{-0.045}&\negred{-0.002}&\negred{-0.010}&0.285\\
\bottomrule
\end{tabular}}
\end{table}
Consistent with the results on the Landsat Satellite dataset, our method is the closest to the oracle in all scenarios. 
The Na\"ive method fails to control errors by a wide margin, and its class-specific error violations worsen as $\eta$ increases. 
Although the Vanilla method maintains error control, it is more conservative than ours for class $3$. 
Consequently, the Vanilla approach results in a substantially larger objective value compared to our method.

%% file: sections/appendices/el.tex
\subsection{Proof of proposition~\ref{proposition:x_conditional_relationships} }
\label{app:proposition_proof}
\begin{proof}[Proof of proposition~\ref{proposition:x_conditional_relationships}]
We first show~\eqref{eq:marginal_link}, recall that $\widetilde{w}_l^* = \sP(\widetilde{Y}=l)$ and $ w_k^* = \sP(Y=k)$, $M_{lk}^* = \sP(Y=k| \widetilde{Y}=l)$ and $T_{lk}^* = \sP(\widetilde{Y}=l| Y=k)$ for all $k,l\in [K]$.
Starting with $ w_k^*$, we have:
\[
w_k^* = \sum_{l}\sP(Y=k,\widetilde{Y}=l) = \sum_{l}\sP(Y=k| \widetilde{Y}=l)\sP(\widetilde{Y}=l) = \sum_{l\in [K]}M_{lk}^*\widetilde{w}_l^*.
\]
Similarly, for $\widetilde{w}_l^*$, we obtain:
\[
\widetilde{w}_l^*= \sum_{k}\sP(Y=k,\widetilde{Y}=l) = \sum_{k}\sP(\widetilde{Y}=l| Y=k)\sP(Y=k) = \sum_{k\in[K]}T_{lk}^* w_k^*,
\]
which completes the proof of~\eqref{eq:marginal_link}.

Next, we prove~\eqref{eq:conditional_distribution_link}.
Let $A$ be an arbitrary event. 
Using the law of total probability, we write: 
\[
\widetilde{P}_l^*(A)=\sum_{k}\sP(X\in A, Y=k| \widetilde{Y}=l)=\sum_{k}\sP(X\in A| Y=k,\widetilde{Y}=l)\sP(Y=k| \widetilde{Y}=l)\]
By Assumption~\ref{assumption:instance_independent_noise}, we can make use of the independence of $X$ and $\widetilde{Y}$ conditioned on $Y$, so this becomes:
\[
\widetilde{P}_l^*(A)=\sum_{k}\sP(X\in A| Y=k)\sP(Y=k| \widetilde{Y}=l)=\sum_{k\in [K]}M_{lk}^*P_k^*(A),
\]
Similarly,
\[
\begin{split}
P_k^*(A) =&~\sP(X\in A| Y=k) = \sum_{l}\sP(X\in A, \widetilde{Y}=l| Y=k)\\
=&~\sum_{l}\sP(X\in A| Y=k,\widetilde{Y}=l)\sP(\widetilde{Y}=l| Y=k)=\sum_{l\in [K]}T_{lk}^*\widetilde{P}_l^*(A),
\end{split}
\]
which completes the proof of~\eqref{eq:conditional_distribution_link}. 

Finally, we prove~\eqref{eq:posterior_link}.
We begin by using the total law of probability,
\[
\sP(\widetilde{Y}=l| X=x) =\sum_{k}\sP(\widetilde{Y}=l,Y=k| X=x).\]
Based on chain rule and Assumption~\ref{assumption:instance_independent_noise}, this simplifies into
\[
\begin{split}
\widetilde{\pi}_{l}^*(x)=\sP(\widetilde{Y}=l| X=x)=&~\sum_{k} \sP(\widetilde{Y}=l,| Y=k,X=x)\sP(Y=k| X=x) \\
=&~\sum_{k} T_{lk}^*\sP(Y=k| X=x) = \sum_{k} T_{lk}^*\pi^*_k(x),    
\end{split}
\]
which completes the proof.
\end{proof}

\subsection{Derivation of the profile log empirical likelihood}
\label{app:profile_loglik}
\begin{proof}
The empirical log-likelihood function as a function of $\vp$ becomes
\[
\ell_{n}(\vp) = \sum_{i=1}^{n} \log p_{i}  + \text{constant}
\]
where the constant depends only on other parameters such as $\{\vw, \bgamma, \bbeta, \mT\}$ and does not depend on $\vp$.
We now maximize the empirical log-likelihood function
with respect to $\vp$ under the constraint~\eqref{eq:constraints} using the Lagrange multiplier
method. 
Let 
\[
\gL=\sum_{i}\log p_{i} - \nu \left(\sum_{i=1}^{n} p_i - 1\right)-n\sum_{k=1}^{K-1}\nu_k\left[\sum_{i}p_{i}\left\{\exp(\gamma_k +\beta_k^{\top}g(X_i))-1\right\}\right].
\]
Setting 
\[
\frac{\partial \gL}{\partial p_{i}} =\frac{1}{p_{i}}-\nu - n\sum_{k=1}^{K-1}\nu_k\left\{\exp(\gamma_k +\beta_k^{\top}g(X_i))-1\right\}=0
\]
gives
\[
p_{i}=\left[\nu + n\sum_{k=1}^{K-1}\nu_k \left\{\exp(\gamma_k +\beta_k^{\top}g(X_i))-1\right\}\right]^{-1}.
\]
Summing $p_{i}\frac{\partial \gL}{\partial p_{i}}$ over $i$, we find that $\nu = n$.
Hence,
\[p_{i}=n^{-1}\left[1 + \sum_{k=1}^{K-1}\nu_k \left\{\exp(\gamma_k +\beta_k^{\top}g(X_i))-1\right\}\right]^{-1}\]
where $\nu_1,\ldots,\nu_{K-1}$ are the solutions to
\[\sum_{i=1}^{n} \frac{\exp(\gamma_k +\beta_k^{\top}g(X_i)) -1}{1+\sum_{k'=1}^{K-1}
\nu_{k'} \bigl\{\exp(\gamma_{k'} +\beta_{k'}^{\top}g(X_i))- 1 \bigr\}} = 0.\]
The proof is completed.    
\end{proof}

\subsection{Asymptotic normality of EL based parameter estimation}
\label{app:asymptotic_normality}
\subsubsection{Notation and preparation}
Let $\bnu = (\nu_1,\ldots, \nu_{K-1})^{\top}$ and $\btheta = \{\vw,\bgamma,\bbeta, \mT\}$ where $\vw=(w_1,\ldots,w_{K-1})^{\top}$ and $\mT = (T_{10},\ldots T_{K-1,0},T_{11},\ldots, T_{K-1,1},\ldots, T_{1,K-1},\ldots,T_{K-1,K-1})^{\top}$.
For simplicity of notation, we view $\btheta$ as vectorized parameters.
We have $w_0 = 1-\sum_{k=1}^{K-1}w_k$ and $T_{0k} = 1-\sum_{l=1}^{K-1}T_{lk}$.
Define
\begin{equation}
\label{eq:obj_h}
\begin{split}
h(\bnu,\btheta) =&~\sum_{i,l}\mathbbm{1}(\widetilde{Y}_i=l)\log\left[\sum_{k} \{w_k T_{lk}\exp(\gamma_k + \beta_k^{\top}g(X_i))\}\right]\\
&-\sum_{i=1}^{n}\log \left[1+\sum_{k=1}^{K-1}\nu_k \left\{\exp(\gamma_k + \beta_k^{\top}g(X_i))-1\right\}\right].
\end{split}    
\end{equation}
For the convenience of presentation, we also define auxiliary functions
\begin{align*}
&\xi(x;\gamma,\beta)=\exp(\gamma + \beta^{\top}g(x)),\\
&\eta_{l}(x;\btheta) = \sum_{k=0}^{K-1}w_k T_{lk}\xi(x;\gamma_k, \beta_k),\\
&\delta(x;\bnu,\btheta) = 1+ \sum_{k=1}^{K-1}\nu_k \left[\exp(\gamma_k + \beta_k^{\top}g(x))-1\right].
\end{align*}

\noindent
\textbf{Gradient of $h$}.
For $j=1,\ldots,K-1$, the gradients of the function $h(\boldsymbol{\nu}, \boldsymbol{\theta})$ are:
\[
\begin{split}
\frac{\partial h}{\partial \nu_j} =&~-\sum_{i=1}^{n}\frac{\xi(X_i;\gamma_j,\beta_j)-1}{\delta(X_i;\bnu,\btheta)}\\
\frac{\partial h}{\partial w_j} =&~\sum_{i=1}^n \sum_{l=0}^{K-1} \mathbbm{1}(\widetilde{Y}_i = l) \frac{T_{lj}\xi(X_i;\gamma_j,\beta_j)-T_{l0}}{\eta_{l}(X_i;\btheta)}\\
\frac{\partial h}{\partial \gamma_j} =&~\sum_{i=1}^n \sum_{l=0}^{K-1} \mathbbm{1}(\widetilde{Y}_i = l) \frac{w_j T_{lj}\xi(X_i;\gamma_j,\beta_j)}{\eta_{l}(X_i;\btheta)}-\sum_{i=1}^{n}\frac{\nu_j \xi(X_i;\gamma_j,\beta_j)}{\delta(X_i;\bnu,\btheta)}\\
\frac{\partial h}{\partial \beta_j} =&~\sum_{i=1}^n \sum_{l=0}^{K-1} \mathbbm{1}(\widetilde{Y}_i = l) \frac{w_j T_{lj}\xi(X_i;\gamma_j,\beta_j)g(X_i)}{\eta_{l}(X_i;\btheta)}-\sum_{i=1}^{n}\frac{\nu_j\xi(X_i;\gamma_j,\beta_j)g(X_i)}{\delta(X_i;\bnu,\btheta)}.
\end{split}
\]
For $l=1,\ldots, K-1$ and $k\in [K]$, we have
\[
\frac{\partial h}{\partial T_{lk}} =\sum_{i=1}^n \mathbbm{1}(\widetilde{Y}_i = l) \frac{w_k \xi(X_i;\gamma_k,\beta_k)}{\eta_{l}(X_i;\btheta)}-\sum_{i=1}^n \mathbbm{1}(\widetilde{Y}_i = 0) \frac{w_k \xi(X_i;\gamma_k,\beta_k)}{\eta_{0}(X_i;\btheta)}.
\]

\subsubsection{Expectation of score function at truth}
In this section, we present the proof of our main technical lemma.
\begin{lemma}
\label{lemma:score_function_exp_var}
Let $\bnu^* = (w_1^*,\ldots, w_{K-1}^*)^{\top}$ and $\btheta^*$ be the true value of $\btheta$.
Let 
\[
\mS_{n} = (\mS_{n1}^{\top},\mS_{n2}^{\top},\mS_{n3}^{\top},\mS_{n4}^{\top},\mS_{n5}^{\top})^{\top}
\]
where $\mS_{n1} = \frac{\partial h(\bnu^*,\btheta^*)}{\partial \bnu}$, $\mS_{n2} = \frac{\partial h(\bnu^*,\btheta^*)}{\partial \bgamma}$, $\mS_{n3} = \frac{\partial h(\bnu^*,\btheta^*)}{\partial \bbeta}$, $\mS_{n4} = \frac{\partial h(\bnu^*,\btheta^*)}{\partial \widetilde{\vw}}$ where $\widetilde{\vw}=(\widetilde{w}_1,\ldots, \widetilde{w}_{K-1})^{\top}$, and $\mS_{n5} = \frac{\partial h(\bnu^*,\btheta^*)}{\partial \mT}$. Then $\mS_n$ satisfies $\sE(\mS_n)=0$.
\end{lemma}

The proof relies on the following key result about expectation transformations:
\begin{lemma}
\label{lemma:expect_change_of_var}
Under the same notation as before, for any measurable function $g$, we have:
\[
\begin{split}
\sE_{X\sim \widetilde{P}_l^*}\left\{\frac{g(X;\btheta)}{\eta_l(X;\btheta)}\right\} =&~\{1/\widetilde{w}_l^*\}\sE_{X\sim P_0^*}\{g(X;\btheta)\}\\
\sE_{*}\left\{\frac{g(X;\btheta)}{\delta(X;\bnu^*,\btheta^*)}\right\} =&~\sE_{X\sim P_{0}^*}\left\{g(X;\btheta)\right\},    
\end{split}
\]
where $\sE_{*}$ denotes expectation under the true marginal distribution of $X$.
\end{lemma}
\begin{proof}
The first equality follows immediately from the relationship that
\[\frac{d\widetilde{P}_l^*}{dP_0^*}(x) =\frac{\eta_{l}(x;\btheta)}{\widetilde{w}_l^*}.\]
For the second conclusion, recall that:
\[\delta(X;\bnu^*,\btheta^*) = \sum_{k=0}^{K-1}w_k^*\xi(X;\gamma_k^*,\beta_k^*)\]
and the true marginal distribution of $X$ is $\sum_{l=0}^{K-1}\widetilde{w}_l^*\widetilde{P}_l^*$.
Therefore, 
\[
\begin{split}
\sE_{*}\left\{\frac{g(X;\btheta)}{\delta(X;\bnu^*,\btheta^*)}\right\} =&~\sum_{l=0}^{K-1}\widetilde{w}_l^* \int \frac{g(x;\btheta)}{\sum_{m=0}^{K-1}w_m^*\xi(x;\gamma_m^*,\beta_m^*)} \,P_{l}^*(dx)\\ 
=&~\sum_{l=0}^{K-1}\int \frac{\eta_l(X;\btheta^*)g(x;\btheta)}{\sum_{m=0}^{K-1}w_m^*\xi(x;\gamma_m^*,\beta_m^*)} \,P_{0}^*(dx)\\ 
=&~\sum_{l=0}^{K-1}\int \frac{\sum_{k=0}^{K-1}w_k^*T_{lk}^*\xi(x;\gamma_k^*,\beta_k^*)g(x;\btheta)}{\sum_{m=0}^{K-1}w_m^*\xi(x;\gamma_m^*,\beta_m^*)} \,P_{0}^*(dx).
\end{split}
\]
By interchanging the order of summation and using the property $\sum_{l=0}^{K-1}T_{lj}^*=1$ for any $j\in[K]$, we obtain:
\[
\sE_{*}\left\{\frac{g(X;\btheta)}{\delta(X;\bnu^*,\btheta^*)}\right\} =\sE_{X\sim P_{0}^*}\left\{g(X;\btheta)\right\},
\]
which completes the proof.
\end{proof}

We now proceed with the proof of Lemma~\ref{lemma:score_function_exp_var}.

\begin{proof}[Proof of Lemma~\ref{lemma:score_function_exp_var}]
We first establish $\sE(\mS_n) = 0$ by showing:
\[\sE\left\{\frac{\partial h(\bnu^*,\btheta^*)}{\partial \nu_j}\right\}=0\quad \text{and}\quad\sE\left\{\frac{\partial h(\bnu^*,\btheta^*)}{\partial w_j}\right\}=0\]
for all $j$.
The remaining equalities follow similarly from Lemma~\ref{lemma:expect_change_of_var}.

For $\partial h/\partial w_j$, applying the law of total expectation conditional on $\widetilde{Y}_i$ yields:
\[
\begin{split}
\sE\left\{\frac{\partial h(\bnu^*,\btheta^*)}{\partial w_j}\right\}=&~\sum_{i=1}^{n}\sum_{l=0}^{K-1}\widetilde{w}_l^*\sE_{\widetilde{P}_l^*}\left\{\frac{T_{lj}^*\xi(X_i;\gamma_j^*,\beta_j^*)-T_{l0}^*}{\eta_{l}(X_i;\btheta^*)}\right\}\\
=&~\sum_{i=1}^{n}\sum_{l=0}^{K-1}\sE_{P_0^*}\left\{T_{lj}^*\xi(X_i;\gamma_j^*,\beta_j^*)-T_{l0}^*\right\},
\end{split}
\]
where the last equality follows from Lemma~\ref{lemma:expect_change_of_var}.
From the exponential tilting relationship in~\eqref{eq:drm}, we have:
\[\sE_{P_0^*}\{\xi(X_i;\gamma_j^*,\beta_j^*)\} = \sE_{P_j^*}(1)=1.\]
Hence, we get
\[\sE\left\{\frac{\partial h(\bnu^*,\btheta^*)}{\partial w_j}\right\} = \sum_{i=1}^{n}\sum_{l=0}^{K-1}\{T_{lj}^* - T_{l0}^*\} = 0,\]
where the last follows from $\sum_{l=0}^{K-1}T_{lj}^*=1$ for any $j\in[K]$.

For $\partial h/\partial \nu_j$, applying Lemma~\ref{lemma:expect_change_of_var} gives:
\[
\sE\left\{\frac{\partial h(\bnu^*,\btheta^*)}{\partial \nu_j}\right\}=\sum_{i=1}^{n}\sE_{P_0^*}\left\{\xi(X_i;\gamma_j^*,\beta_j^*)-1\right\}
=\sum_{i=1}^{n}\left\{\sE_{P_j^*}(1)-\sE_{P_0^*}(1)\right\}=0,
\]
This completes the proof that $\sE(\mS_n) = 0$.
\end{proof}

\subsubsection{Asymptotic normality of the estimator}
\begin{proof}[Proof of Theorem~\ref{thm:estimator_asymptotic_normality}]
Let $\bnu^* = (w_1^*,\ldots, w_{K-1}^*)^{\top}$ and $\btheta^*$ be the true value of $\btheta$.
For the ease of notation, we denote 
\[\vv = (\bnu^{\top}, \btheta^{\top})^{\top}\quad\text{and}\quad\vv^* = ((\bnu^*)^{\top}, (\btheta^*)^{\top})^{\top}\quad\text{and}\quad\widehat\vv = (\widehat\bnu^{\top}, \widehat\btheta^{\top})^{\top}.\]
We now establish the asymptotic properties of $\widehat\vv$, which can be obtained via the second-order Taylor expansion on $h(\vv)$.

Following the proofs of Lemma 1 and Theorem 1 of~\cite{qin1994empirical} and the proof of Theorem 2 in~\citet{qin2015using}, we obtain 
\[\widehat\vv = \vv^* + O_p(n^{-1/2}).\]
Note that the MELE $\widehat\btheta$ of $\btheta$ and the corresponding Lagrange multipliers $\widehat\bnu$ must satisfy
\[\frac{\partial h(\widehat\vv)}{\partial \vv} = 0.\]
Applying a first-order Taylor expansion to $\partial h(\widehat\vv)/\partial \vv$ gives
\be
\label{eq:helper1}
0 = \frac{\partial h(\widehat\vv)}{\partial \vv} = \frac{\partial h(\vv^*)}{\partial \vv} + \frac{\partial^2 h(\vv^*)}{\partial \vv \partial \vv^{\top}}(\widehat\vv - \vv^*) + o_p(n^{1/2}).
\ee
Using the law of large numbers, we have
\be
\label{eq:helper2}
\frac{1}{n}\frac{\partial^2 h(\vv^*)}{\partial \vv\partial \vv^{\top}} = -\mW + o_p(1)
\ee
where $\mW = \sE\left\{n^{-1}\partial^2 h(\vv^*)/\partial \vv\partial \vv^*\right\}$,
Combining~\eqref{eq:helper1} and~\eqref{eq:helper2}, we get
\[\widehat\vv - \vv^* = n^{-1}\mW^{-1}\mS_n + o_p(n^{-1/2}).\]
Using the central limit theorem, Lemma~\ref{lemma:score_function_exp_var}, and Slutsky’s theorem, we have
\[\sqrt{n}(\widehat\vv - \vv^*) \to N(0, \Sigma' = \mW^{-1}\mV\mW^{-1})\]
in distribution.
The marginal asymptotic distribution of $\btheta$ is therefore also a normal distribution, which completes the proof.
\end{proof}

%% file: sections/appendices/identifiability.tex
In this section, we establish the theoretical guarantees for the identifiability of the unknown parameters in our proposed model. We observe data from $K$ distinct distributions $\{\widetilde{P}_l^*\}_{l=1}^K$, which are assumed to be identifiable. 
Each distribution $\widetilde{P}_l^*$ is modeled as a mixture over $K$ latent components relative to an unknown reference distribution $P_0^*$. 

Specifically, the density (or Radon-Nikodym derivative) of $\widetilde{P}_l^*$ with respect to $P_0^*$ for an observation $x \in \gX$ is given by:
\begin{equation} \label{eq:model_setup}
    \frac{d\widetilde{P}_l^*}{dP_0^*}(x) = \sum_{k=1}^K M_{lk}^* \exp\big(\gamma_k^* + \langle \beta_k^*, g(x) \rangle \big), \quad \forall l \in [K],
\end{equation}
where $\mM^* \in \mathbb{R}^{K \times K}$ is the mixing matrix with entries $M_{lk}^*$, $\gamma_k^* \in \mathbb{R}$ are the component-specific intercepts, $\beta_k^* \in \mathbb{R}^d$ are the feature coefficients, and $g: \gX \to \mathbb{R}^d$ is a known feature mapping.

\subsubsection{Assumptions}
To ensure that the parameters $(\mM^*, \bgamma^*, \bbeta^*)$ and $P_0^*$ are uniquely identifiable from the distributions $\{\widetilde{P}_l^*\}_{l=1}^K$, we require the following regularity conditions.

\begin{assumption}[Distinct features] \label{assump:distinct}
    The coefficients $\beta_k^*$ are pairwise distinct for all $k \in [K]$.
\end{assumption}

\begin{assumption}[Richness of feature space] \label{assump:richness}
The image of the feature mapping, $\gZ = \{g(x) : x \in \gX\} \subseteq \mathbb{R}^d$, contains a non-empty open set in $\mathbb{R}^d$.
\end{assumption}

\begin{assumption}[Invertibility and recoverability of confusion matrix] \label{assump:recoverability}
The $K \times K$ mixing matrix $\mM^*$ has full rank, and its rows are normalized such that $\sum_{k=1}^K M_{lk}^* = 1$ for all $l \in [K]$.
Furthermore, the parameter space restricts $\mM^*$ such that its inverse has strictly positive diagonal entries and nonpositive off-diagonal entries.
\end{assumption}

\subsubsection{Supporting Lemma}
Assumption \ref{assump:richness} is crucial to ensure that the exponential functions in our mixture do not collapse into linear dependence over the data support.

\begin{lemma} 
\label{lemma:linear_indep}
    Under Assumption \ref{assump:richness}, for any finite set of pairwise distinct vectors $\beta_1, \dots, \beta_K \in \mathbb{R}^d$, the functions $f_k(x) = \exp(\langle \beta_k, g(x) \rangle)$ are linearly independent over $\gX$. That is, if 
    \[\sum_{k=1}^K c_k \exp(\langle \beta_k, g(x) \rangle) = 0 \text{ for all  }x \in \gX, \]
    then $c_1 = \dots = c_K = 0$.
\end{lemma}

\begin{proof}
    Let $z = g(x)$. By hypothesis, $F(z) = \sum_{k=1}^K c_k \exp(\langle \beta_k, z \rangle) = 0$ holds for all $z \in \gZ$, where $\gZ = \{g(x) : x \in \gX\} \subseteq \mathbb{R}^d$ is the image of $\gX$. Because $\gZ$ contains a non-empty open set (Assumption \ref{assump:richness}) and $F(z)$ is real-analytic on $\mathbb{R}^d$, the identity theorem for real-analytic functions dictates that $F(z) = 0$ for all $z \in \mathbb{R}^d$.

    Since the vectors $\{\beta_k\}_{k=1}^K$ are pairwise distinct, the finite set of difference vectors $\{\beta_i - \beta_j : i \neq j\}$ does not contain the zero vector. Consequently, the union of the orthogonal hyperplanes $\{v \in \mathbb{R}^d : \langle \beta_i - \beta_j, v \rangle = 0\}$ has Lebesgue measure zero in $\mathbb{R}^d$. We can therefore choose a direction vector $v \in \mathbb{R}^d$ that avoids all such hyperplanes. For this $v$, the inner products $\lambda_k = \langle \beta_k, v \rangle$ are strictly distinct scalars. 

    Evaluating $F(z)$ along the line $z = tv$ for $t \in \mathbb{R}$ yields:
    \begin{equation*}
        F(tv) = \sum_{k=1}^K c_k \exp(t \langle \beta_k, v \rangle) = \sum_{k=1}^K c_k \exp(\lambda_k t) = 0, \quad \forall t \in \mathbb{R}.
    \end{equation*}
    Without loss of generality, assume $\lambda_1 > \lambda_2 > \dots > \lambda_K$. Dividing the equation by $\exp(\lambda_1 t)$ gives:
    \begin{equation*}
        c_1 + \sum_{k=2}^K c_k \exp\big((\lambda_k - \lambda_1) t\big) = 0.
    \end{equation*}
    Taking the limit as $t \to \infty$, the terms $\exp((\lambda_k - \lambda_1)t)$ vanish because $\lambda_k - \lambda_1 < 0$ for $k \ge 2$. This implies $c_1 = 0$. Applying this argument recursively to the remaining terms yields $c_k = 0$ for all $k \in [K]$, concluding the proof.
\end{proof}

\subsubsection{Proof of identifiability}
We now present the proof for the main theorem~\ref{thm:identifiability}.

\begin{proof}
    Suppose there are two sets of parameters $(M, \gamma, \beta)$ with reference density $p_0(x)$, and $(\widetilde{M}, \widetilde{\gamma}, \widetilde{\beta})$ with reference density $\widetilde{p}_0(x)$, both satisfying Assumptions \ref{assump:distinct}--\ref{assump:recoverability}, that produce identical observed distributions $p_l(x)$ for all $l \in [K]$. 
    
    Define the exponential components as $f_k(x) = \exp(\gamma_k + \langle \beta_k, g(x) \rangle)$ and $\widetilde{f}_j(x) = \exp(\widetilde{\gamma}_j + \langle \widetilde{\beta}_j, g(x) \rangle)$. Equating the two parameterizations yields:
    \begin{equation*}
        \sum_{k=1}^K M_{lk} f_k(x) p_0(x) = \sum_{j=1}^K \widetilde{M}_{lj} \widetilde{f}_j(x) \widetilde{p}_0(x), \quad \forall l \in [K].
    \end{equation*}
    Let $R(x) = \widetilde{p}_0(x) / p_0(x)$ be the density ratio. Let $f(x) = [f_1(x), \dots, f_K(x)]^\top$ and $\widetilde{f}(x) = [\widetilde{f}_1(x), \dots, \widetilde{f}_K(x)]^\top$. In matrix notation, $\mM f(x) = \widetilde{\mM} \widetilde{f}(x) R(x)$. 
    
    By Assumption \ref{assump:recoverability}, $\mM$ is invertible. Left-multiplying by $\mM^{-1}$ and defining $A = \mM^{-1} \widetilde{\mM}$, we obtain $f(x) = A \widetilde{f}(x) R(x)$. The $k$-th component is:
    \begin{equation} \label{eq:matrix_A_elements}
        f_k(x) = \sum_{j=1}^K A_{kj} \widetilde{f}_j(x) R(x).
    \end{equation}

    \textbf{Step 1: Isolating the Density Ratio $R(x)$.} \\
    Evaluate \eqref{eq:matrix_A_elements} for the anchor $k=1$. 
    By assumption $\beta_1 = \mathbf{0}$ and $\gamma_1 = 0$, so $f_1(x) = 1$. Thus,
    \begin{equation*}
        1 = \left( \sum_{j=1}^K A_{1j} \widetilde{f}_j(x) \right) R(x) \implies R(x) = \left( \sum_{j=1}^K A_{1j} \widetilde{f}_j(x) \right)^{-1}.
    \end{equation*}
    Substitute $R(x)$ back into the general equation \eqref{eq:matrix_A_elements} for any $k$:
    \begin{equation*}
        f_k(x) \sum_{j=1}^K A_{1j} \widetilde{f}_j(x) = \sum_{m=1}^K A_{km} \widetilde{f}_m(x).
    \end{equation*}
    Expanding the definitions of $f_k$ and $\widetilde{f}_j$ yields:
    \begin{equation} \label{eq:poly_exp}
        \sum_{j=1}^K A_{1j} \exp\big(\gamma_k + \widetilde{\gamma}_j + \langle \beta_k + \widetilde{\beta}_j, g(x) \rangle\big) = \sum_{m=1}^K A_{km} \exp\big(\widetilde{\gamma}_m + \langle \widetilde{\beta}_m, g(x) \rangle\big).
    \end{equation}

    \textbf{Step 2: Uniqueness of the Anchor Mapping.} \\
    Note that~\eqref{eq:poly_exp} must hold for every fixed $k \in [K]$. To see how linear independence applies, we bring all terms to one side:
    \begin{align*}
        &\sum_{j=1}^K \Big[ A_{1j} \exp(\gamma_k + \widetilde{\gamma}_j) \Big] \exp\big(\langle \beta_k + \widetilde{\beta}_j, g(x) \rangle\big) \\
        &\qquad - \sum_{m=1}^K \Big[ A_{km} \exp(\widetilde{\gamma}_m) \Big] \exp\big(\langle \widetilde{\beta}_m, g(x) \rangle\big) = 0.
    \end{align*}
    Before invoking linear independence, we must group terms with identical exponent vectors. Let $V$ be the set of all strictly unique exponent vectors present across both sums. We can rewrite the equation as:
    \begin{equation*}
        \sum_{v \in V} W_v \exp\big(\langle v, g(x) \rangle\big) = 0,
    \end{equation*}
    where $W_v$ is the net scalar coefficient for the unique vector $v$. Because the elements of $V$ are now pairwise distinct by definition, Lemma \ref{lemma:linear_indep} directly applies, dictating that $W_v = 0$ for all $v \in V$.

    Suppose there is an ``active'' index $j$ on the left side where $A_{1j} \neq 0$. The exponent vector for this term is $v^* = \beta_k + \widetilde{\beta}_j$. Because the original vectors $\{\widetilde{\beta}_j\}_{j=1}^K$ are pairwise distinct (Assumption \ref{assump:distinct}), the shifted vectors $\{\beta_k + \widetilde{\beta}_j\}_{j=1}^K$ are also strictly pairwise distinct. This means no other term in the first sum shares the exponent vector $v^*$. 

    Therefore, the net coefficient for $v^*$ takes the form:
    \begin{equation*}
        W_{v^*} = A_{1j} \exp(\gamma_k + \widetilde{\gamma}_j) - D_{v^*},
    \end{equation*}
    where $D_{v^*} = A_{km} \exp(\widetilde{\gamma}_m)$ if $v^*$ matches some $\widetilde{\beta}_m$ in the second sum, and $D_{v^*} = 0$ if it does not. 
    
    Since $A_{1j} \neq 0$ and the exponential function is strictly positive, the first term $A_{1j} \exp(\gamma_k + \widetilde{\gamma}_j)$ is non-zero. However, we established that $W_{v^*} = 0$. The only way this is mathematically possible is if $D_{v^*} \neq 0$. This forces the conclusion that $v^*$ must identically match one of the available vectors in the second sum.
    
    Thus, for any active index $j$, the vector $\beta_k + \widetilde{\beta}_j$ must perfectly equal $\widetilde{\beta}_m$ for some $m \in [K]$.

    Let $S = \{\widetilde{\beta}_1, \dots, \widetilde{\beta}_K\}$ be the set of available exponent vectors on the right-hand side. For a fixed active index $j$ (where $A_{1j} \neq 0$), as we vary $k$ from $1$ to $K$, the left-hand side generates a set of $K$ exponent vectors: $B_j = \{\beta_1 + \widetilde{\beta}_j, \dots, \beta_K + \widetilde{\beta}_j\}$. Because every vector generated on the left must be cancelled by a vector on the right, every vector in $B_j$ must belong to $S$. Since $B_j \subseteq S$ and both sets have exactly $K$ elements, they must be identical: $B_j = S$.

    Now, suppose there are two distinct active indices, $i$ and $j$, where $A_{1i} \neq 0$ and $A_{1j} \neq 0$. By the exact same cancellation logic, we must have $B_i = S$ and $B_j = S$, which implies $B_i = B_j$. This set equivalence means that $\{\beta_k + \widetilde{\beta}_i\}_{k=1}^K = \{\beta_k + \widetilde{\beta}_j\}_{k=1}^K$. Letting $\mathcal{B} = \{\beta_1, \dots, \beta_K\}$, we can write this as:
    \begin{equation*}
        \mathcal{B} + \widetilde{\beta}_i = \mathcal{B} + \widetilde{\beta}_j \implies \mathcal{B} = \mathcal{B} + (\widetilde{\beta}_i - \widetilde{\beta}_j).
    \end{equation*}
    This states that the finite set of vectors $\mathcal{B}$ is invariant under translation by the vector $\Delta = \widetilde{\beta}_i - \widetilde{\beta}_j$. However, a finite set of real vectors cannot be invariant under translation by a non-zero vector.

    Therefore, the translation vector must be zero, meaning $\widetilde{\beta}_i = \widetilde{\beta}_j$. Since the coefficient vectors are pairwise distinct, we must have $i = j$. Conclusively, there is exactly one index $j^*$ such that $A_{1j^*} \neq 0$, and $A_{1j} = 0$ for all $j \neq j^*$.

    \textbf{Step 3: Establishing the Permutation.} \\
    Since only $A_{1j^*} \neq 0$, Equation \eqref{eq:poly_exp} collapses to a single term on the left-hand side:
    \begin{equation*}
        A_{1j^*} \exp\big(\gamma_k + \widetilde{\gamma}_{j^*} + \langle \beta_k + \widetilde{\beta}_{j^*}, g(x) \rangle\big) = \sum_{m=1}^K A_{km} \exp\big(\widetilde{\gamma}_m + \langle \widetilde{\beta}_m, g(x) \rangle\big).
    \end{equation*}
    By linear independence, this single term must match exactly one term on the right-hand side. For each $k$, there is exactly one index $m = \sigma(k)$ where $A_{k, \sigma(k)} \neq 0$. This immediately gives $\beta_k + \widetilde{\beta}_{j^*} = \widetilde{\beta}_{\sigma(k)}$.
    
    Evaluate this for the anchor $k=1$. Since $\beta_1 = \mathbf{0}$, we have $\widetilde{\beta}_{j^*} = \widetilde{\beta}_{\sigma(1)}$. Because the vectors are distinct, it follows that $\sigma(1) = j^*$.

    \textbf{Step 4: Fixing the Permutation via Recoverability.} \\
    We have established that $A$ contains exactly one non-zero entry per row at $(k, \sigma(k))$. Thus, $A$ is a permutation matrix scaled by row-specific factors. From $\widetilde{\mM} = \mM A$, we take the inverse: $(\widetilde{\mM})^{-1} = A^{-1} \mM^{-1}$. 
    
    The inverse matrix $A^{-1}$ has its non-zero entries at $(\sigma(k), k)$. Because $R(x) = 1 / (A_{1j^*} \widetilde{f}_{j^*}(x))$ must be strictly positive (as a ratio of two densities), we must have $A_{1j^*} > 0$. The equality of coefficients from Step 3 ensures $A_{k, \sigma(k)} = A_{1j^*} \exp(\gamma_k + \widetilde{\gamma}_{j^*} - \widetilde{\gamma}_{\sigma(k)}) > 0$. Since all non-zero entries of $A$ are positive, all non-zero entries of $A^{-1}$ are also strictly positive.
    
    The diagonal entries of the new inverse matrix are given by:
    \begin{equation*}
        [(\widetilde{\mM})^{-1}]_{k,k} = [A^{-1}]_{k, \sigma^{-1}(k)} [\mM^{-1}]_{\sigma^{-1}(k), k}.
    \end{equation*}
    Suppose the permutation is not the identity, meaning there is some $k$ where $\sigma^{-1}(k) \neq k$. Then $[\mM^{-1}]_{\sigma^{-1}(k), k}$ is an off-diagonal entry. By Assumption \ref{assump:recoverability}, all off-diagonal entries of $\mM^{-1}$ are nonpositive ($\le 0$). This implies that $[(\widetilde{\mM})^{-1}]_{k,k} \le 0$, which explicitly contradicts the requirement in Assumption \ref{assump:recoverability} that $(\widetilde{\mM})^{-1}$ must have strictly positive diagonal entries. 
    
    Therefore, we must have $\sigma(k) = k$ for all $k \in [K]$. This forces $j^* = \sigma(1) = 1$.

    \textbf{Step 5: Global Identifiability of $P_0$ and Parameters.} \\
    Because $j^* = 1$ and $\sigma(k) = k$, we have $A_{11} \neq 0$ and $A_{1j} = 0$ for all $j \neq 1$. Applying the anchor assumption $\widetilde{\beta}_1 = \mathbf{0}$ and $\widetilde{\gamma}_1 = 0$, we substitute this back into our ratio expression from Step 1:
    \begin{equation*}
        R(x) = \frac{1}{A_{11} \exp(0)} = \frac{1}{A_{11}}.
    \end{equation*}
    The density ratio $R(x)$ is therefore a constant. Because $p_0(x)$ and $\widetilde{p}_0(x)$ are both probability densities (due to the normalization of the rows of $\mM$), they must integrate to 1 over $\gX$. Thus:
    \begin{equation*}
        \int \widetilde{p}_0(x) d\mu(x) = \int \frac{1}{A_{11}} p_0(x) d\mu(x) \implies 1 = \frac{1}{A_{11}} \cdot 1 \implies A_{11} = 1.
    \end{equation*}
    Therefore, $R(x) = 1$, which proves that $\widetilde{p}_0(x) = p_0(x)$ for all $x$, uniquely identifying the reference distribution $P_0$. 
    
    With $A_{11} = 1$ and $\sigma(k) = k$, the relation $\beta_k + \widetilde{\beta}_1 = \widetilde{\beta}_k$ becomes $\beta_k = \widetilde{\beta}_k$. Equating the coefficients $A_{k, k} = A_{11} \exp(\gamma_k + \widetilde{\gamma}_1 - \widetilde{\gamma}_k)$ yields $1 = \exp(\gamma_k - \widetilde{\gamma}_k)$, meaning $\gamma_k = \widetilde{\gamma}_k$. 
    
    Finally, since $A$ is the identity matrix $I$, the relation $\widetilde{\mM} = \mM A$ reduces to $\widetilde{\mM} = \mM$. We conclude that the parameters $(\mM, \gamma, \beta)$ and the reference distribution $P_0$ are strictly identifiable.
\end{proof}

%% file: sections/appendices/def_lemma.tex
\subsubsection{Definitions}
\begin{definition}[Rademacher complexity]
\label{def:rademacher_complexity}
Let $Z_1,\ldots, Z_n$ be independent Rademacher random variables (\ie $\sP(Z_i=-1)=\sP(Z_i=1)=0.5$).
Given a function class $\gF$ and IID samples $X_1,\ldots, X_n$, the empirical Rademacher complexity of $\gF$ is defined to be:
\[\widehat{\mathfrak{R}}_n(\gF) = \sE_{Z_1,\ldots,Z_n}\left(\sup_{f\in\gF} \left|\frac{1}{n}\sum_{i=1}^{n}Z_i f(X_i)\right|\right).\]
The Rademacher complexity of $\gF$ is defined as:
$\mathfrak{R}_n(\gF) = \sE_{X_1,\ldots,X_n}\{\widehat{\mathfrak{R}}_n(\gF)\}$.
\end{definition}

\begin{definition}[$\epsilon$-cover and covering number]
Let $\gF$ be a class of functions and $D$ a metric on $\gF$. 
\begin{itemize}[leftmargin=*]
\item An \emph{$\epsilon$-cover} of $\gF$ with respect to $D$ is a finite set $\{f_1, \ldots, f_l\} \subset \gF$ such that for every $f \in \gF$, there exists $f_i$ satisfying $D(f, f_i) \leq \epsilon$.
\item The \emph{$\epsilon$-covering number}, denoted $\gN(\epsilon, \gF, D)$, is the smallest cardinality of any $\epsilon$-cover of $\gF$ with respect to $D$.
\end{itemize}
\end{definition}

\subsubsection{Lemmas}
\begin{lemma}[Bounds for Vapnik–Chervonenkis classes]
\label{lemma:dudley}   
Let $X_1, \ldots, X_n$ be fixed sample points with the empirical $L_2$ norm defined as
\be
\label{eq:l2_norm}
\|f\|_{n}^2 := \frac{1}{n}\sum_{i=1}^n f^2(X_i).
\ee
Then we have 
\[\mathfrak{R}_n(\gF) \leq \frac{24}{\sqrt{n}}\int_{0}^{B} \sqrt{\log\gN(t,\gF, \|\|_n)}\,dt\]
where $\sup_{f,g\in \gF} \|f-g\|_n \leq B$.
\end{lemma}
\begin{proof}[Proof of Lemma~\ref{lemma:dudley}]
See~\citet[Example 5.24]{wainwright2019high}.
\end{proof}

\begin{lemma}[Uniform law via Rademacher complexity]
\label{lemma:ULLN_rademacher}
For a function class $\gF$ that is uniformly bounded by $M$ (\ie $\|f\|_{\infty}\leq M$ for all $f\in \gF$), the following inequality holds with probability at least $1 - \delta$:
\[\sup_{f\in\gF} \left|\frac{1}{n}\sum_{i=1}^{n} f(X_i)-\sE[f(X)]\right|\leq 2\mathfrak{R}_n(\gF) + M\sqrt{\frac{2\log(1/\delta)}{n}}\]
where $\mathfrak{R}_n(\gF)$ is the Rademacher complexity of $\gF$ defined in Definition~\ref{def:rademacher_complexity}.
\end{lemma}
\begin{proof}[Proof of Lemma~\ref{lemma:ULLN_rademacher}]
See~\citet[Theorem 4.10]{wainwright2019high}.
\end{proof}

\begin{lemma}[Covering number of product of uniformly bounded functions]
\label{lemma:covering_number_product}
Let $\gH$ and $\gK$ be two classes of functions that are uniformly bounded by $M_{H}$ and $M_{K}$.
Define the product class
\[
\gF = \{f(x) = h(x)k(x) : h \in \gH, \, k \in \gK\}.
\]
Then, for any $\epsilon > 0$, the covering number of $\gF$ w.r.t. the empirical $L_2$ norm in~\eqref{eq:l2_norm} satisfies
\[
\gN(\epsilon, \gF, \|\cdot\|_{n}) 
\le \gN\Big(\frac{\epsilon}{M_{H}+M_{K}}, \gH, \|\cdot\|_{n}\Big) 
\cdot \gN\Big(\frac{\epsilon}{M_{H}+M_{K}}, \gK, \|\cdot\|_{n}\Big).
\]
\end{lemma}

\begin{proof}[Proof of Lemma~\ref{lemma:covering_number_product}]
For any $\epsilon > 0$, let $\delta = \frac{\epsilon}{M_{H}+M_{K}}$.  
Let $\gH_\delta$ be a $\delta$-cover of $\gH$ under $\|\cdot\|_{n}$, and let $\gK_\delta$ be a $\delta$-cover of $\gK$.  

For any $f = hk \in \gF$, select approximations $\tilde{h} \in \gH_\delta$ and $\tilde{k} \in \gK_\delta$ such that
\[
\|h - \tilde{h}\|_{n} \le \delta, 
\quad \|k - \tilde{k}\|_{n} \le \delta.
\]
Define $\tilde{f} := \tilde{h}\tilde{k}$.
Then, $f - \tilde{f} 
= hk - \tilde{h}\tilde{k}
= h(k - \tilde{k}) + (h - \tilde{h})\tilde{k}$.
We have
\[
\|f - \tilde{f}\|_{n} 
\le \|h(k - \tilde{k})\|_{n} + \|(h - \tilde{h})\tilde{k}\|_{n}.
\]
Next,
\[
\|h(k - \tilde{k})\|_{n} 
= \sqrt{\frac{1}{n}\sum_{i=1}^n h^2(X_i)\,(k(X_i)-\tilde{k}(X_i))^2}
\le \max_{1 \le i \le n} |h(X_i)| \cdot \|k - \tilde{k}\|_{n}
\le M_{H} \delta.
\]
Similarly,
\[
\|(h - \tilde{h})\tilde{k}\|_{n}
\le \max_{1 \le i \le n} |\tilde{k}(X_i)| \cdot \|h - \tilde{h}\|_{n}
\le M_{K} \delta.
\]
Combining, we get
\[
\|f - \tilde{f}\|_{n} 
\le (M_{H}+M_{K})\delta = \epsilon.
\]
In summary, the collection $\{\tilde{h}\tilde{k} \colon \tilde{h} \in \gH_\delta, \tilde{k} \in \gK_\delta\}$ forms an $\epsilon$-cover for $\gF$.
This gives the bound
\[
\gN(\epsilon, \gF, \|\cdot\|_{n})
\le \gN\Big(\frac{\epsilon}{M_{H}+M_{K}}, \gH, \|\cdot\|_{n}\Big) 
\cdot \gN\Big(\frac{\epsilon}{M_{H}+M_{K}}, \gK, \|\cdot\|_{n}\Big).
\]
This completes the proof.    
\end{proof}

\begin{lemma}[Covering number of pointwise maximum functions]
Let $\gF_1,\ldots, \gF_K$ be $K$ classes of functions.
Define the maximum class
\[
\gF_{\max} =  \left\{f(x) = \max_{k\in[K]}f_k(x) : f_k \in \gF_k \right\}.
\]
Then, for any $\epsilon > 0$, the covering number of $\gF_{\max}$ w.r.t. the empirical $L_2$ norm in~\eqref{eq:l2_norm} satisfies
\[
\gN(\epsilon, \gF_{\max}, \|\cdot\|_{n}) 
\le \prod_{k\in [K]}\gN\Big(\frac{\epsilon}{K}, \gF_k, \|\cdot\|_{n}\Big).
\]    
\end{lemma}

\begin{proof}
For any $\epsilon > 0$, let $\delta = \frac{\epsilon}{K}$.  
Let $\gF_{k,\delta}$ be a $\delta$-cover of $\gF_{k}$ under $\|\cdot\|_{n}$.  
For any $f = \max_{k} f_k \in \gF_{\max}$, select approximations $\tilde{f}_k \in \gF_{k,\delta}$ such that $\|f_k - \tilde{f}_k\|_{n} \le \delta$.
Define $\tilde{f} := \max_{k} \tilde{f}_k$.
Then for any $x$, we have
\[
|f(x) - \tilde{f}(x)|\leq \max_{k\in[K]}|f_k(x)-\tilde{f}_k(x)| \leq \sum_{k\in [K]} |f_k(x)-\tilde{f}_k(x)|. 
\]
Combine this with the triangular inequality gives
\[
\begin{split}
\|f - \tilde{f}\|_{n} \leq&~\sqrt{\frac{1}{n}\sum_{i=1}^{n}\left(\sum_{k\in [K]} |f_k(X_i)-\tilde{f}_k(X_i)|\right)^2}\leq\sum_{k=1}^{K}\sqrt{\frac{1}{n}\sum_{i=1}^{n}\left(f_k(X_i)-\tilde{f}_k(X_i)\right)^2}\\
=&~\sum_{k\in [K]}\|f_k - \tilde{f}_k\|_{n}\leq K\delta = \epsilon.
\end{split}
\]
In summary, the collection $\{\tilde{f}_k \colon \tilde{f}_k \in \gF_{k,\delta}\}$ forms an $\epsilon$-cover for $\gF_{\max}$.
This gives the bound
\[
\gN(\epsilon, \gF_{\max}, \|\cdot\|_{n})
\le \prod_{k\in [K]}\gN\Big(\frac{\epsilon}{K}, \gF_{k}, \|\cdot\|_{n}\Big).
\]
This completes the proof.        
\end{proof}

%% file: sections/appendices/error_control_binary.tex
\subsubsection{Proof of type I error under binary classification}
\begin{proof}[Proof of Theorem~\ref{thm:binary_type1}]
To make the notation clearer, let $\pi^*(x)$ and $w^*$ denote the true values of $\pi(x)$ and $w$ in the data generating process.
Let $\widehat\pi(x)$ and $\widehat{w}$ be the corresponding estimators for $\pi^*(x)$ and $w^*$ in~\eqref{eq:binary_conditional} and~\eqref{eq:mele_estimator}, respectively. 
The final classifier is
\[
\widehat\phi_{\widehat\lambda}(x) = \mathbbm{1}\left(\widehat\lambda\leq r(x,\widehat\pi, \widehat{w})\right),
\]
where $\widehat\lambda$ is the solution to 
\[
\frac{1}{n}\sum_{i=1}^n \{1-\widehat{\pi}(X_i)\}\mathbbm{1}(\widehat\phi_{\lambda}(X_i)=1)=\frac{1}{n}\sum_{i=1}^n \{1-\widehat{\pi}(X_i)\}\mathbbm{1}\left(\lambda \leq r(X_i,\widehat{\pi}, \widehat{w})\right) = \alpha(1-\widehat w).
\]

Our goal is to show that the type I error $P_{0}(\{\widehat\phi_{\widehat\lambda}(X)\neq 0\})$ is controlled at level $\alpha$ in probability. 
Note that for any classifier $\phi$, the type I error can be rewritten as 
\[
P_{0}(\phi(X)\neq 0) = \frac{\sE_{X}[\{1-\pi^*(X)\}\mathbbm{1}(\phi(X)=1)]}{\sP(Y=0)}.
\]
We can therefore equivalently show that 
\[
\sE_{X}[\{1-\pi^*(X)\}\mathbbm{1}(\widehat\phi_{\widehat\lambda}(X)=1)]\leq \alpha(1-w^*) + O_{p}(n^{-1/2}).
\]

The difference of interest can be decomposed into
\[
\begin{split}
&\sE_{X}[\{1-\pi^*(X)\}\mathbbm{1}(\widehat\phi_{\widehat\lambda}(X)=1)] - \alpha(1-w^*) \\
= & \underbrace{\sE_{X}[\{\widehat\pi(X)-\pi^*(X)\}\mathbbm{1}(\widehat\phi_{\widehat\lambda}(X)=1)]}_{T_1} + \underbrace{\alpha(\widehat{w}-w^*)}_{T_2} \\
& + \underbrace{\left\{\sE_{X}[\{1-\widehat\pi(X)\}\mathbbm{1}(\widehat\phi_{\widehat\lambda}(X)=1)] - \frac{1}{n}\sum_{i=1}^n \{1-\widehat{\pi}(X_i)\}\mathbbm{1}(\widehat\phi_{\widehat\lambda}(X_i)=1)\right\}}_{T_3}.
\end{split}
\]
We now bound each term separately.

\noindent
\textbf{Bound on $T_1$:}
By the Lipschitz property of $\pi(\cdot)$ and the Cauchy-Schwarz inequality, we have
\[
|T_1| \leq \sE_{X}[|\widehat\pi(X)-\pi^*(X)|] \leq \frac{1}{4}\bigl(|\widehat\gamma - \gamma^*| + \|\widehat\beta - \beta^*\|_2\sE_{X}\|g(X)\|_2\bigr).
\]
Since $\sE_{X}\|g(X)\|_2<\infty$, this along with Theorem~\ref{thm:estimator_asymptotic_normality} implies that $|T_1| = O_p(n^{-1/2})$.

\noindent
\textbf{Bound on $T_2$:}
This term is straightforward.
Note that
\[
|T_2| = \alpha|\widehat{w}-w^*|.
\]
By the consistency of $\widehat{w}$, we have $|T_2| = O_p(n^{-1/2})$.

\noindent
\textbf{Bound on $T_3$:}    
The classifier can be expressed as a thresholding rule:
\[
\mathbbm{1}(\widehat\phi_{\widehat\lambda}(x)=1) = \mathbbm{1}\left(\widehat\pi(x) \geq \frac{\widehat\lambda \widehat{w}}{(1-\widehat{w}) + \widehat\lambda \widehat{w}}\right).
\]
Consider the class of functions:
\[
\gF = \left\{f_{\gamma,\beta,t}(x) = (1-\pi_{\gamma,\beta}(x))\mathbbm{1}(\pi_{\gamma,\beta}(x) \geq t) : \gamma \in \mathbb{R}, \|\beta\|_2 \leq B, t \in [0,1]\right\}
\]
where $\pi_{\gamma,\beta}(x) = \{1+\exp(-\gamma-\beta^{\top}g(x))\}^{-1}$.
We then have 
\[
T_3 \leq \sup_{f\in \gF} \left|\frac{1}{n}\sum_{i=1}^{n}f(X_i) - \sE f(X)\right|.
\]
Since $\gF$ is uniformly bounded by 1, by Lemma~\ref{lemma:ULLN_rademacher}, we have with probability at least $1-\delta$:
\[
T_3 \leq 2\mathfrak{R}_{n}(\gF) + \sqrt{\frac{2\log(1/\delta)}{n}},
\]
where $\mathfrak{R}_n(\gF)$ is the Rademacher complexity of $\gF$.
We have shown in Lemma~\ref{lemma:rademacher_weighted_logistic_threshold} that 
\[\mathfrak{R}_n(\gF) = C'\sqrt{\frac{d+1}{n}}\]
where $C'$ is the universal constant specified in Lemma~\ref{lemma:rademacher_weighted_logistic_threshold}.
Combining all bounds, we conclude that:
\[
\sE_{X}[\{1-\pi^*(X)\}\mathbbm{1}(\widehat\phi_{\widehat\lambda}(X)=1)] \leq \alpha(1-w^*) + O_p(n^{-1/2}),
\]
which completes the proof.
\end{proof}

\subsubsection{Rademacher complexity of weighted logistic classifiers}
\begin{lemma}[Covering number of weighted logistic threshold functions]
\label{lemma:covering_number_weighted_logistic_threshold}
Let $\|\cdot\|_{n}$ denote the empirical $L_2$-norm defined in Lemma~\ref{lemma:covering_number_product}. Consider the function class 
\[
\gF = \Bigl\{f_{\gamma,\beta,t}(x) = \bigl(1 - \pi_{\gamma,\beta}(x)\bigr)\, \mathbbm{1}\bigl(\pi_{\gamma,\beta}(x) \geq t\bigr) : 
|\gamma| \leq B, \|\beta\|_2 \leq B, t \in [0,1] \Bigr\},
\]
where $B > 0$ is a constant, and $\pi_{\gamma,\beta}(x) = \bigl\{1 + \exp\bigl(-\gamma - \beta^\top g(x)\bigr)\bigr\}^{-1}$ is the logistic function with $g(x)\in \sR^{d}$. 
Then, for any $\epsilon \in (0,1)$, the $\epsilon$-covering number of $\gF$ w.r.t. $\|\cdot\|_{n}$ satisfies
\[
\log \gN\bigl(\epsilon, \gF, \|\cdot\|_{n}\bigr) 
\leq 2C(d+1)\log\left(1+\frac{4B\max\{G_n,1\}}{\epsilon}\right).
\]
where $C > 0$ is a universal constant independent of $d$, $B$, $\epsilon$, and $G_n^2 = n^{-1}\sum_{i=1}^{n}\|g(X_i)\|_2^2$.
\end{lemma}

\begin{proof}[Proof of Lemma~\ref{lemma:covering_number_weighted_logistic_threshold}]
Let $\sigma(z) = \{1 + \exp(-z)\}^{-1}$ denote the logistic sigmoid function. 
Define the auxiliary function classes:
\[
\begin{split}
\gH 
=&~\Bigl\{x\mapsto 1 - \sigma(\gamma+\beta^\top g(x)): |\gamma|\leq B, \|\beta\|_2 \leq B \Bigr\}, \\ 
\gK 
=&~\Bigl\{x\mapsto \mathbbm{1}\bigl(\sigma(\gamma+\beta^\top g(x)) \geq t\bigr): |\gamma|\leq B, \|\beta\|_2 \leq B, t\in[0,1]\Bigr\}.
\end{split}
\]
Since both $\gH$ and $\gK$ are uniformly bounded by $1$, Lemma~\ref{lemma:covering_number_product} implies
\[
\gN(\epsilon, \gF, \|\cdot\|_{n}) \leq \gN(\epsilon/2, \gH, \|\cdot\|_{n}) \cdot \gN(\epsilon/2, \gK, \|\cdot\|_{n}).
\]
We now bound each term separately.

\medskip
\noindent\textbf{Covering number of $\gH$.}  
Let $h_{\gamma,\beta}\in \gH$ denote the function $x\mapsto  1 - \sigma(\gamma + \beta^{\top}g(x))$.
Note that $\|\sigma'\|_{\infty} \leq 1/4$.
Thus, for any two parameters $(\gamma, \beta)$ and $(\gamma', \beta')$, we have
\[
\begin{split}
|\sigma(\gamma+\beta^\top g(x)) - \sigma(\gamma' + \beta'^\top g(x))|
\leq \frac{1}{4}\bigl(|\gamma - \gamma'| + |\beta - \beta'|\|g(x)\|_2\bigr).    
\end{split}
\]
This along with $(a+b)^2\leq 2a^2 +2b^2$ imply that the empirical $L_2$ norm satisfies
\[
\begin{split}
\|h_{\gamma,\beta} - h_{\gamma',\beta'}\|_n^2 & = \frac{1}{n}\sum_{i=1}^{n}\left(\sigma(\gamma+\beta^\top g(X_i)) - \sigma(\gamma' + \beta'^\top g(X_i))\right)^2 \\
& \leq\frac{1}{2}|\gamma - \gamma'|^2 +\frac{1}{2}\underbrace{\left\{\frac{1}{n}\sum_{i=1}^{n}\|g(X_i)\|_2^2\right\}}_{G_n^2}\|\beta - \beta'\|_2^2. 
\end{split}
\]
Therefore, if $|\gamma-\gamma'| \leq \epsilon$ and $\|\beta-\beta'\|_2 \leq \epsilon/G_n$, then  $\|h_{\gamma,\beta} - h_{\gamma',\beta'}\|_n \leq \epsilon$.
This implies that the $\epsilon$-covering number of $\gH$ is bounded by the product of:
\begin{itemize}
    \item The $\epsilon$-covering number of $[-B,B] \subset \mathbb{R}$
    \item The $\epsilon/G_n$-covering number of the radius $B$ ball in $\mathbb{R}^{d}$
\end{itemize}

Using standard bounded ball covering number results (see~\citet[Lemma 5.7]{wainwright2019high} for example), the $\epsilon$-covering of a radius $B$ ball in $\sR^{d}$ is at most 
\[\left(1+\frac{2B}{\epsilon}\right)^{d}.\]
Therefore, we have 
\[
\gN(\epsilon, \gH, \|\cdot\|_{n}) 
\leq \left(1+\frac{2B}{\epsilon}\right)\left(1+\frac{2BG_n}{\epsilon}\right)^{d}.
\]

\medskip
\noindent\textbf{Covering number of $\gK$.}  
The class $\gK$ consists of indicator functions of half-spaces in $\mathbb{R}^{d+1}$, as it can be rewritten as:
\[
\gK = \Bigl\{x \mapsto \mathbbm{1}(\gamma + \beta^{\top}g(x) \geq \log(t/(1-t))) : |\gamma| \leq B, \|\beta\|_2 \leq B, t \in [0,1]\Bigr\}.
\]
Hence, its  VC dimension is $d+1$.

By~\citet[Theorem 8.3.18]{vershynin2018high}, there exists some universal constant $C>0$ such that 
\[\gN(\epsilon, \gK, \|\cdot\|_{n}) \leq (2/\epsilon)^{C\cdot\text{VC}(\gK)},\]
where $\text{VC}(\gK)$ is the VC dimension of $\gK$.
Therefore, we have 
\[
\gN(\epsilon, \gK, \|\cdot\|_{n})
\leq \left(\frac{2}{\epsilon}\right)^{C(d+1)}.
\]
\medskip
\noindent\textbf{Covering number of $\gF$.}  
Putting everything together yields:
\[
\begin{split}
\log\gN(\epsilon, \gF, \|\cdot\|_{n})
\leq&~(d+1)\log\left(1+\frac{4B\max\{G_n,1\}}{\epsilon}\right)+C(d+1)\log\left(\frac{4}{\epsilon}\right)\\
\leq&~2C(d+1)\log\left(1+\frac{4B\max\{G_n,1\}}{\epsilon}\right).
\end{split}
\]
This completes the proof.
\end{proof}

\begin{lemma}[Rademacher complexity of weighted binary classifier]
\label{lemma:rademacher_weighted_logistic_threshold}
Let $\gF$ be the same function class defined in Lemma~\ref{lemma:covering_number_weighted_logistic_threshold}.
Assume that $\mathbb{E}\bigl[\|g(X)\|_2^2\bigr] \leq R^2$ for some $R > 0$, we have 
\[\mathfrak{R}_{n}(\gF) \leq C'\sqrt{\frac{d+1}{n}},\]
where $C'=24\sqrt{\pi}(1+2BR)C$ and $C$ is the universal constant in Lemma~\ref{lemma:covering_number_weighted_logistic_threshold}.
\end{lemma}

\begin{proof}[Proof of Lemma~\ref{lemma:rademacher_weighted_logistic_threshold}]
Note that we have $\sup_{f,g\in \gF}\|f-g\|_{n} = 1$.
By Lemma~\ref{lemma:dudley},  we have 
\[
\begin{split}
\widehat{\mathfrak{R}}_{n}(\gF)\leq&~\frac{24}{\sqrt{n}}\int_{0}^{1}\sqrt{\log \gN(t;\gF, \|\cdot\|_n)}\,dt\\ 
\leq&~
24\sqrt{2}C\sqrt{\frac{d+1}{n}}\int_{0}^{1}\sqrt{\log(1+4B\max\{G_n,1\}/t)}\,dt\\
\leq&~24\sqrt{\pi/2}C\sqrt{\frac{d+1}{n}}(1+4B\max\{G_n,1\})
\end{split}
\]

Since $\sE\max\{G_n, 1\}\leq \sE(G_n)+1$, we get
\[
\begin{split}
\mathfrak{R}_{n}(\gF) =&~\sE_{X_1,\ldots, X_n}\{\widehat{\mathfrak{R}}_{n}(\gF)\}\\ 
\leq &~24\sqrt{\pi/2}C\sqrt{\frac{d+1}{n}}\{2+4B\sE_{X_1,\ldots, X_n}(G_n)\} \leq C'\sqrt{\frac{d+1}{n}}
\end{split}
\]
where $C' = 24\sqrt{\pi}(1+2BR)C$ and the last inequality is based on Jensen's inequality.
This completes the proof.
\end{proof}


%% file: sections/appendices/error_control_multiclass.tex
\subsubsection{Proof of class-specific error control}
\begin{proof}[Proof of Theorem~\ref{thm:multiclass_error}]
Recall that 
\[R_k(\phi) = \sP(\phi(X) \neq k | Y=k)\]
Then by the strong duality, the optimal classifier is a feasible solution that satisfies $R_k(\phi^*_{\blambda^*}) \leq \alpha_k$.
Therefore, we have 
\[R_k(\widehat\phi_{\widehat\blambda}) - \alpha_k = R_k(\widehat\phi_{\widehat\blambda}) - R_k(\phi^*_{\blambda^*}) + R_k(\phi^*_{\blambda^*}) - \alpha_k \leq R_k(\widehat\phi_{\widehat\blambda}) - R_k(\phi^*_{\blambda^*}).\]
Hence, we only need to bound
\[R_k(\widehat\phi_{\widehat\blambda}) - R_k(\phi^*_{\blambda^*}) \leq \sP(\widehat\phi_{\widehat\blambda}(X)\neq \phi^*_{\blambda^*}(X)|Y=k).\]

Note that the RHS is a random variable that depends on the training data $\gD^{\text{train}}$.
We wish to show that it converges to $0$ in probability as $n \to \infty$.
Since the random variable is bounded uniformly by 1, by Markov's inequality it suffices to show that
\[
\lim_{n\to \infty}\sE_{\gD^{\text{train}}}\left\{ \sP_X(\widehat\phi_{\widehat\blambda}(X)\neq \phi^*_{\blambda^*}(X)|Y=k) \right\} = 0.
\]

Using the dominant convergence theorem, we can exchange limits and expectations to see that the LHS is equal to
\[
\sE_{X}\left\{\lim_{n\to \infty}\sP_{\gD^{\text{train}}}(\widehat\phi_{\widehat\blambda}(X)\neq \phi^*_{\blambda^*}(X)\big|Y=k) \right\}
\]
if the inner limit exists pointwise.
It thus suffices to show that
\[
\lim_{n\to \infty} \sP_{\gD^{\text{train}}}(\widehat\phi_{\widehat\blambda}(X)\neq \phi^*_{\blambda^*}(X)) = 0
\]
holds almost surely in $X$.
For fixed $X$, we may write
\[
\phi^*_{\blambda^*}(X) = \arg\max_{k} Z_k, \text{ with } Z_k = \frac{\rho_k + \blambda_k^*\mathbbm{1}(k\in \gS)}{w_k^*}\pi_k^*(X)
\]
\[
\widehat\phi_{\widehat\blambda}(X) = \arg\max_{k} \widehat{Z}_k, \text{ with } \widehat{Z}_k = \frac{\rho_k + \widehat{\blambda}_k\mathbbm{1}(k\in \gS)}{\widehat{w}_k}\widehat{\pi}_k(X).
\]
Almost surely, there are no ties among $\{Z_k\}_{k=0}^{K-1}$, that is, $\eta > 0$ where $\eta$ is the gap between the largest and second largest entries in this set.
By Lemma~\ref{lemma:lagrange_multiplier} and Theorem~\ref{thm:estimator_asymptotic_normality}, we have 
\[
\widehat{Z}_k \overset{p}{\to} Z_k
\]
as $n\to\infty$.
As such, for any $\epsilon > 0$, for $n$ large enough, we have
\[
\sP\left(\max_{k\in[K]} |\widehat{Z}_k - Z_k | < \eta/2\right) \geq 1 - \epsilon.
\]
On this high probability set, we have that $\phi^*_{\blambda^*}(X) = \widehat\phi_{\widehat\blambda}(X)$, which implies that
\[
\sP_{\gD^{\text{train}}}\left(\widehat\phi_{\widehat\blambda}(x)= \phi^*_{\blambda^*}(X)\right) \geq 1 - \epsilon.
\]
Since this holds for all $\epsilon > 0$, the conclusion follows.
\end{proof}

\subsubsection{Asymptotic convergence of dual problem}
\begin{lemma}[Uniform convergence of dual objective]
\label{lemma:uniform_convergence}
Under Assumptions~\ref{assumption:bounded_second_moment} and~\ref{assump:compactness_multiclass}, for any bounded set $\Lambda\subseteq \sR_{+}^{|\sS|}$, the empirical dual objective $\widehat{G}(\lambda)$ converges uniformly to $G(\lambda)$ almost surely:
\[\sP\left(\lim_{n\to\infty} \sup_{\lambda\in\Lambda}|\widehat{G}(\blambda) - G(\blambda)|=0\right)=1.\]
\end{lemma}

\begin{proof}
Let $\pi_k^*(x) = \mathbb{P}(Y=k|X=x)$ denote the true conditional probability and $\widehat{\pi}_k(x) = \widehat{\mathbb{P}}(Y=k|X=x)$ its empirical estimate. 
We decompose the difference as follows:
\[
\begin{split}
|\widehat{G}(\blambda) - G(\blambda)|
=&~\left|\frac{1}{n}\sum_{i=1}^{n}\max_{k} \left\{c_{k}(\blambda, \widehat{\vw})\widehat\pi_k(X_i)\right\} - \sE_{X}\left[\max_{k} \left\{c_{k}(\blambda, \vw^*)\pi_k^*(X)\right\}\right]\right|\\
\leq &~\left|\frac{1}{n}\sum_{i=1}^{n}\max_{k} \left\{c_{k}(\blambda, \widehat{\vw})\widehat\pi_k(X_i)\right\} - \sE_{X}\left[\max_{k} \left\{c_{k}(\blambda, \widehat\vw)\widehat\pi_k(X)\right\}\right]\right|\\
&+\left|\sE_{X}\left[\max_{k} \left\{c_{k}(\blambda, \widehat{\vw})\widehat\pi_k(X)\right\}\right] - \sE_{X}\left[\max_{k} \left\{c_{k}(\blambda, \vw^*)\pi_k^*(X)\right\}\right]\right|
\end{split}
\]
where $c_k(\blambda, \vw) = \{\rho_k + \lambda_k \mathbbm{1}(k\in \gS)\}/w_k.$
Taking supremum over $\blambda\in \Lambda$ yields:
\[
\begin{split}
\sup_{\blambda\in \Lambda}|\widehat{G}(\blambda) - G(\blambda)| \leq&~\underbrace{\sup_{f\in \gF}\left|\frac{1}{n}\sum_{i=1}^{n} f(X_i) - \sE f(X)\right|}_{T_1} \\
&+ \underbrace{\sup_{\blambda\in \Lambda}\left|\sE_{X}\left[\max_{k} \left\{c_{k}(\blambda, \widehat{\vw})\widehat\pi_k(X)\right\}\right] - \sE_{X}\left[\max_{k} \left\{c_{k}(\blambda, \vw^*)\pi_k^*(X)\right\}\right]\right|}_{T_2},
\end{split}
\]
where 
\[
\gF=\left\{f_{\bgamma,\bbeta,\vw, \blambda}(x) = \max_{k}\{c_k(\blambda, \vw)\pi_k(x)\}: 
\|\gamma\|_2 \leq B, \|\beta\|_2 \leq B, w_k \geq c_k>0,\blambda\in\Lambda\right\}
\]
and $\pi_k(x) = \exp(\gamma_k + \beta_k^{\top} g(x))/\sum_{k'}\exp(\gamma_{k'} + \beta_{k'}^{\top} g(x))$.
We now bound each of $T_1$ and $T_2$ respectively.

\noindent
\textbf{Bound $T_1$:}
Since $\gF$ is uniformly bounded and satisfies the conditions of Lemma~\ref{lemma:ULLN_rademacher}, we have with probability at least $1-\delta$:
\[
T_1 \leq 2\mathfrak{R}_{n}(\gF) + \sqrt{\frac{2\log(1/\delta)}{n}},
\]
where $\mathfrak{R}_n(\gF)$ is the Rademacher complexity of $\gF$.
By Lemma~\ref{lemma:rademacher_multiclass}, there exists a constant $C'>0$ such that:
\[
\mathfrak{R}_n(\mathcal{F}) \leq C'\sqrt{\frac{d+3}{n}}.
\]
Thus, $T_1 \to 0$ almost surely as $n \to \infty$.

\noindent
\textbf{Bound $T_2$:}
Using the inequality $|\max_k f_k - \max_k g_k| \leq \max_k |f_k - g_k|\leq \sum_{k} |f_k-g_k|$, we obtain:
\[
\begin{split}
T_2\leq & \sum_{k} \sup_{\blambda\in \Lambda}\sE_{X}\left|c_{k}(\blambda, \widehat{\vw})\widehat\pi_k(X) - c_{k}(\blambda, \vw^*)\pi_k^*(X)\right|\\
\leq & \sum_{k} \sup_{\blambda\in \Lambda}c_{k}(\blambda, \widehat{\vw})\sE_{X}\left|\widehat\pi_k(X)-\pi_k^*(X)\right| + \sum_{k} \sE_{X}(\pi_k^*(X)) \sup_{\blambda\in \Lambda}|c_{k}(\blambda, \widehat\vw)-c_{k}(\blambda, \vw^*)|.
\end{split}
\]
Since $w_k^*$ is bounded away from zero and $\pi_k(x)$ is Lipschitz in its parameters, and $\widehat{\vw} \to \vw^*$ by asymptotic normality, we have:
\[
T_2 \leq C\left[\|\widehat\bgamma-\bgamma^*\|_2 + \|g(X)\|_2 \|\widehat\bbeta -\bbeta^*\|_2 + \sum_{k} \left\{\frac{1}{\widehat{w}_k} - \frac{1}{w_k^*}\right\}\right] \to 0
\]
almost surely as $n \to \infty$.

Combining these results yields the desired uniform convergence:
\[\sP\left(\lim_{n\to\infty} \sup_{\lambda\in\Lambda}|\widehat{G}(\blambda) - G(\blambda)|=0\right)=1.\]
\end{proof}

\begin{lemma}[Convergence of Lagrange multiplier]
\label{lemma:lagrange_multiplier}
Under Assumptions~\ref{assumption:bounded_second_moment},~\ref{assump:compactness_multiclass}, and~\ref{assump:local_strong_concavity}, we have
\[\widehat\blambda\overset{a.s.}{\longrightarrow} \blambda^*\quad\text{as}\quad n\to\infty.\]
\end{lemma}

\begin{proof}[Proof of Lemma~\ref{lemma:lagrange_multiplier}]
For any $\epsilon>0$, to show that $\|\widehat\blambda - \blambda^*\|_2 \leq \epsilon$, it suffices to show that 
\begin{equation} \label{eq:lagrange_multiplier_conv_helper}
    \widehat{G}(\blambda^*) > \sup_{\blambda\in \partial B_{\epsilon}(\blambda^*)} \widehat{G}(\blambda).
\end{equation}
Then by the concavity of $\widehat{G}$, we must have $\widehat\blambda\in B_{\epsilon}(\blambda^*)$.
As such, if we can show that \eqref{eq:lagrange_multiplier_conv_helper} holds almost surely as $n \to \infty$, we have
\[\sP\left(\limsup_{n\to\infty} \|\widehat{\blambda}-\blambda^*\|_2 >\epsilon\right)=0.\]
If this holds for all $\epsilon > 0$, we get $\widehat\blambda\overset{a.s.}{\longrightarrow} \blambda^*$ as $n\to\infty$.

To verify \eqref{eq:lagrange_multiplier_conv_helper}, observe that
\[
\begin{split}
&~\widehat{G}(\blambda^*) - \sup_{\blambda\in \partial B_{\tau}(\blambda^*)} \widehat{G}(\blambda) \\
\geq &~\left\{G(\blambda^*) - \sup_{\blambda\in \partial B_{\tau}(\blambda^*)} G(\blambda)\right\} - 2\sup_{\blambda\in \bar{B}_{\epsilon}(\blambda^*)}|\widehat{G}(\blambda)-G(\blambda)|\\ 
>&~ - 2\sup_{\blambda\in \bar{B}_{\epsilon}(\blambda^*)}|\widehat{G}(\blambda)-G(\blambda)|,
\end{split}
\]
where the strict inequality follows from $\nabla^2 G(\blambda^*) \prec 0$, which guarantees $\blambda^*$ is the unique maximizer of $G$ in $\bar{B}_{\epsilon}(\blambda^*)$.
By Lemma~\ref{lemma:uniform_convergence}, $\sup_{\blambda \in \overline{B}_{\epsilon}(\blambda^*)} |\widehat{G}(\blambda) - G(\blambda)| \to 0$ almost surely as $n \to \infty$. Thus, \eqref{eq:lagrange_multiplier_conv_helper} holds almost surely, completing the proof.
\end{proof}

\subsubsection{Rademacher complexity of weighted multinomial logistic classifiers}

\begin{lemma}[Covering number of weighted multinomial logistic classifier]
\label{lemma:covering_number_multiclass}
Let $\|\cdot\|_{n}$ denote the empirical $L_2$-norm defined in Lemma~\ref{lemma:covering_number_product}. 
Let $B>0$ be some fixed constant, $\Lambda\subseteq \sR_{+}^{|\sS|}$ be any bounded set, and
\[\gF=\left\{f_{\bgamma,\bbeta,\vw, \blambda}(x) = \max_{k}\{c_k(\blambda, \vw)\pi_k(x)\}: 
\|\gamma\|_2 \leq B, \|\beta\|_2 \leq B, w_k \geq 1/B,\blambda\in\Lambda\right\}
\]
and $\pi_k(x) = \exp(\gamma_k + \beta_k^{\top} g(x))/\sum_{k'}\exp(\gamma_{k'} + \beta_{k'}^{\top} g(x))$.

Then, for any $\epsilon>0$, the $\epsilon$-covering number of $\gF$ w.r.t. $\|\cdot\|_{n}$ satisfies
\[
\log \gN\bigl(\epsilon, \gF, \|\cdot\|_{n}\bigr) 
\leq K^2(d+3)\log\left(1+\frac{2BG_n}{\epsilon}\right)
\]
where $G_n^2=K(A/4)^2 \left\{1 + n^{-1}\sum_{i=1}^{n}\|g(X_i)\|_2^2\right\} + A^2B^2 + B^2$ and $A=\max_{k} (\rho_k + B)B$.
\end{lemma}

\begin{proof}[Proof of Lemma~\ref{lemma:covering_number_multiclass}]
Since $\sup_{f,g\in \gF}\|f-g\|_n$ is bounded by a finite constant, we show the upper bound for the Rademacher complexity of $\gF$ via Lemma~\ref{lemma:dudley}.
As such, we consider the covering number of $\gF$.
Denote
\[
\gG_{k} = \left\{g^{(k)}_{\bgamma,\bbeta, w, \lambda}(x) = C_{k}(\lambda, w)\pi_k(x): 
\|\gamma\|_2 \leq B, \|\beta\|_2 \leq B, w \geq 1/B,0\leq \lambda \leq B\right\},
\]
where $C_{k}(\lambda, w) = \frac{\rho_k + \lambda \mathbbm{1}(k\in\gS)}{w}$.
By Lemma~\ref{lemma:covering_number_product}, for any $t>0$, we have 
\be
\label{eq:covering_f_bounded_by_g}
\log \gN(t;\gF, \|\cdot\|_n) = \sum_{k} \log \gN(t;\gG_{k}, \|\cdot\|_n).
\ee
Therefore, we only need to consider the covering number of $\gG_{k}$.

We first find the Lipschitz-constant of $g^{(k)}$ for any fixed $x$.
Note that the $C_{k}(\lambda, w) \leq (\rho_k + B)B$.
For any $x$, the softmax derivative satisfies $\partial \pi_k/\partial \eta_j = \pi_k(\delta_{kj} - \pi_j)$ with $\left|\partial \pi_k/\partial \eta_j \right| \leq 1/4$.
Therefore,
\[
\left| \frac{\partial g^{(k)}}{\partial \gamma_j} \right|
= C_{k}(\lambda, w) \left| \frac{\partial \pi_k}{\partial \gamma_j} \right|
\leq \frac{C_{k}(\lambda, w)}{4}\leq (\rho_k + B)B/4,
\]
and
\[
\left\| \frac{\partial g^{(k)}}{\partial \beta_j} \right\|
= C_{k}(\lambda, w) \left\| \frac{\partial \pi_k}{\partial \beta_j} \right\|
\leq \frac{C_{k}(\lambda, w)}{4} \, \|g(x)\|_2\leq (\rho_k + B)B\|g(x)\|_2/4.
\]
Moreover,
\[
\left| \frac{\partial g^{(k)}}{\partial w} \right|
= \frac{|\rho_k| + \lambda \,\mathbbm{1}(k \in \mathcal S)}{w^2} \pi_k(x)
\leq (\rho_k + B)B^2,
\]
and
\[
\left| \frac{\partial g^{(k)}}{\partial \lambda} \right|
= \frac{\mathbbm{1}(k \in \mathcal S)}{w} \pi_k(x) \leq B.
\]
For any $\theta = (\bgamma,\bbeta, w, \lambda), \theta'=(\bgamma',\bbeta', w', \lambda')$ in the parameter space, we have that 
\[
\begin{split}
\left|g^{(k)}_{\theta}(x) - g^{(k)}_{\theta'}(x)\right| \leq&~\sup_{\tilde\theta} \left\| \nabla_{\theta} g^{(k)}_{\tilde\theta}(x) \right\|_2 \cdot \|\theta - \theta'\|_2 \\
\leq &~\left(K \left( \frac{A}{4} \right)^2
+ K \left(\frac{A}{4} \right)^2 \|g(x)\|_2^2
+ A^2B^2 + B^2\right)^{1/2}\|\theta - \theta'\|_2
\end{split}
\]
where $A := \sup_{\lambda,w} C_{k}(\lambda, w) \leq (\rho_k+B)B$.

This implies that the empirical $L_2$ norm satisfies
\[
\begin{split}
\|g^{(k)}_{\theta} - g^{(k)}_{\theta'}\|_n^2 & = \frac{1}{n}\sum_{i=1}^{n}\left(\sigma(\gamma+\beta^\top g(X_i)) - \sigma(\gamma' + \beta'^\top g(X_i))\right)^2 \\
& \leq \underbrace{\left\{K \left( \frac{A}{4} \right)^2
+ K \left(\frac{A}{4} \right)^2 \frac{1}{n}\sum_{i=1}^{n}\|g(X_i)\|_2^2
+ A^2B^2 + B^2\right\}}_{G_n^2}\|\theta - \theta'\|_2^2. 
\end{split}
\]
Therefore, if $\|\theta-\theta'\|_2 \leq \epsilon/G_n$, then  $\|g^{(k)}_{\theta} - g^{(k)}_{\theta'}\|_n \leq \epsilon$.
This implies that the $\epsilon$-covering number of $\gG_{k}$ is bounded by the $\epsilon/G_n$-covering number of the radius $B$ ball in $\mathbb{R}^{d}$.

Using standard bounded ball covering number results (see~\citet[Lemma 5.7]{wainwright2019high} for example), the $\epsilon$-covering of a radius $B$ ball in $\sR^{(K-1)(d+3)}$ is at most 
\[\left(1+\frac{2B}{\epsilon}\right)^{(K-1)(d+3)}.\]
Therefore, we have 
\[
\gN(\epsilon, \gG_{k}, \|\cdot\|_{n}) 
\leq \left(1+\frac{2BG_n}{\epsilon}\right)^{(K-1)(d+3)}.
\]
This combines with~\eqref{eq:covering_f_bounded_by_g} gives
\[\log \gN(\epsilon;\gF, \|\cdot\|_n) \leq K(K-1)(d+3)\log\left(1+\frac{2BG_n}{\epsilon}\right),\]
which completes the proof of the lemma.
\end{proof}

\begin{lemma}[Rademacher complexity of weighted multinomial logistic classifier]
\label{lemma:rademacher_multiclass}
Let $\gF$ be the same as that defined in~\ref{lemma:covering_number_multiclass}.
Under Assumption~\ref{assumption:bounded_second_moment}, we have 
\[\mathfrak{R}_{n}(\gF) \leq C\sqrt{\frac{d+3}{n}},\]
for some universal constant $C>0$.
\end{lemma}

\begin{proof}[Proof of Lemma~\ref{lemma:rademacher_multiclass}]
Note that we have $\sup_{f,g\in \gF}\|f-g\|_{n} \leq M$ for some finite constant $M>0$.
By Lemma~\ref{lemma:dudley}, we have 
\[
\begin{split}
\widehat{\mathfrak{R}}_{n}(\gF)\leq&~\frac{24}{\sqrt{n}}\int_{0}^{M}\sqrt{\log \gN(t;\gF, \|\cdot\|_n)}\,dt\\ 
\leq&~
\frac{24K\sqrt{d+3}}{\sqrt{n}}\int_{0}^{M}\sqrt{\log\left(1+\frac{2BG_n}{t}\right)}\,dt\\
\leq&~\frac{48K\sqrt{2(d+3)BMG_n}}{\sqrt{n}}\leq \frac{48K\sqrt{2(d+3)BM}G_n}{\sqrt{n}}
\end{split}
\]
where $G_n^2=K(A/4)^2 \left\{1 + n^{-1}\sum_{i=1}^{n}\|g(X_i)\|_2^2\right\} + A^2B^2 + B^2$ and $A=\max_{k} (\rho_k + B)B$.

Then by Cauchy-Schwarz inequality, we have
\[
\begin{aligned}
\sE\big[\widehat{\mathfrak{R}}_{n}(\mathcal{F})\big]
&\leq \frac{48K\sqrt{2(d+3)BM}\sE(G_n)}{\sqrt{n}}\leq \frac{48K\sqrt{2(d+3)BM}\{\sE(G_n^2)\}^{1/2}}{\sqrt{n}}\leq C'\sqrt{\frac{d+3
}{n}},
\end{aligned}
\]
where the last inequality holds by Assumption~\ref{assumption:bounded_second_moment} and $C'>0$ is some fixed constant.
\end{proof}